\documentclass[a4paper]{article}

\usepackage[utf8]{inputenc}

\usepackage{microtype}

\usepackage{fullpage}

\usepackage[hidelinks,colorlinks,allcolors=blue]{hyperref}

\usepackage[usenames,dvipsnames]{xcolor}

\usepackage{amsmath,amssymb,amsthm}

\usepackage{algorithm2e}
 
\usepackage{graphicx,caption,subfig}

\usepackage{tabularx,enumerate}

\newtheorem{lemma}{Lemma}[section]
\newtheorem{theorem}[lemma]{Theorem}
\newtheorem{claim}[lemma]{Claim}

\newtheorem{corollary}[lemma]{Corollary}
\newtheorem{proposition}[lemma]{Proposition}

\theoremstyle{definition}
\newtheorem{definition}[lemma]{Definition}

\theoremstyle{remark}
\newtheorem{observation}[lemma]{Observation}

\usepackage{tikz} 
\usetikzlibrary{matrix}
\usetikzlibrary{arrows,shapes,positioning}
\usetikzlibrary{decorations.pathreplacing}

\DeclareMathOperator{\fn}{fn}

\newcommand{\F}{\mathcal{F}}

\newcommand{\sm}{\setminus}
\newcommand{\bigO}{\mathcal{O}}

\renewcommand{\phi}{\varphi}

\newcommand{\fd}{{\textsc{$\F$-Completion}}}

\newcommand{\yes}{\textit{\textbf{yes}}}

\newcommand{\no}{\textit{\textbf{no}}}

\newcommand{\defparproblem}[4]{% PGD Version
  \hfill\\\smallskip\noindent%
  \begin{tabularx}{\textwidth}{|l X|}%
    \hline%
    \multicolumn{2}{|l|}{\pname{#1}}\\%
    \textup{Input:}&#2\\%
    \textup{Parameter:}&#3\\%
    \textup{Question:}&#4\\\hline%
  \end{tabularx}%
  \smallskip%
}%

\DeclareMathOperator{\uni}{uni}

\newcommand{\pname}{\textsc}
\newcommand{\cclass}{\textsf}

\newcommand{\vpmc}{vital potential maximal clique}
\newcommand{\Vpmc}{Vital potential maximal clique}

\DeclareMathOperator{\dpt}{dp}

\title{Exploring Subexponential Parameterized Complexity of Completion
  Problems\thanks{Supported by Rigorous Theory of Preprocessing, ERC
    Advanced Investigator Grant 267959}}

\author{
  Pål Grønås Drange%
  \thanks{University of Bergen, Norway,
    \texttt{\{Pal.Drange|Fedor.Fomin|Michal.Pilipczuk|Yngve.Villager\}%
      @ii.uib.no}}
  \addtocounter{footnote}{-1}
  \and Fedor V.~Fomin\footnotemark \addtocounter{footnote}{-1}
  \and Michał Pilipczuk\footnotemark \addtocounter{footnote}{-1}
  \and Yngve Villanger\footnotemark \addtocounter{footnote}{-1}
}

\date{}

\begin{document}
\maketitle

\begin{abstract}
  Let $\F$ be a family of graphs.  In the \textsc{$\F$-Completion}
  problem, we are given an $n$-vertex graph $G$ and an integer $k$ as
  input, and asked whether at most $k$ edges can be added to $G$ so
  that the resulting graph does not contain a graph from $\F$ as an
  induced subgraph.  It appeared recently that special cases of
  \textsc{$\F$-Completion}, the problem of completing into a chordal
  graph known as \textsc{Minimum Fill-in}, corresponding to the case
  of $\F=\{C_4,C_5,C_6,\ldots\}$, and the problem of completing into a
  split graph, i.e., the case of $\F=\{C_4, 2K_2, C_5\}$, are solvable
  in parameterized subexponential time
  $2^{\bigO(\sqrt{k}\log{k})}n^{\bigO(1)}$.  The exploration of this
  phenomenon is the main motivation for our research on \fd{}.
  
  In this paper we prove that completions into several well studied
  classes of graphs without long induced cycles also admit
  parameterized subexponential time algorithms by showing that:
  \begin{itemize}
  \item The problem \pname{Trivially Perfect Completion} is solvable
    in parameterized subexponential time
    $2^{\bigO(\sqrt{k}\log{k})}n^{\bigO(1)}$, that is \fd{} for $\F
    =\{C_4, P_4\}$, a cycle and a path on four vertices.
  \item The problems known in the literature as \textsc{Pseudosplit
      Completion}, the case where $\mathcal{F} = \{2K_2, C_4\}$, and
    \pname{Threshold Completion}, where $\F = \{2K_2, P_4, C_4\}$, are
    also solvable in time $2^{\bigO(\sqrt{k}\log{k})} n^{\bigO(1)}$.
  \end{itemize}
  
  We complement our algorithms for \textsc{$\F$-Completion} with the
  following lower bounds:
  \begin{itemize}
  \item For $\F = \{2K_2\}$, $\F = \{ C_4\}$, $\F = \{P_4\}$, and $\F
    = \{2K_2, P_4\}$, \textsc{$\F$-Completion} cannot be solved in
    time $2^{o(k)} n^{\bigO(1)}$ unless the Exponential Time
    Hypothesis~(ETH) fails.
  \end{itemize}

  Our upper and lower bounds provide a complete picture of the
  subexponential parameterized complexity of \textsc{$\F$-Completion}
  problems for $\mathcal{F}\subseteq \{2K_2, C_4, P_4\}$.
\end{abstract}

\section{Introduction}
\label{sec:introduction}

Let $\F$ be a family of graphs.  In this paper we study the following
\fd{} problem.

\defparproblem{\fd{}} {A graph $G=(V,E)$ and a non-negative integer
  $k$.}{$k$}{Does there exist a supergraph $H=(V,E\cup S)$ of $G$,
  such that $|S| \leq k$ and $H$ contains no graph from $\F$ as an
  induced subgraph?}

The \fd{} problems form a subclass of graph modification problems
where one is asked to apply a bounded number of changes to an input
graph to obtain a graph with some property.  Graph modification
problems arise naturally in many branches of science and have been
studied extensively during the past 40 years.  Interestingly enough,
despite the long study of the problem, there is no known dichotomy
classification of \fd{} explaining for which classes $\F$ the problem
is solvable in polynomial time and for which the problem is
\cclass{NP}-complete~\cite{Yannakakis81Edge,mancini2008graph,BurzynBD06}.

One of the motivations to study completion problems in graph
algorithms comes from their intimate connections to different width
parameters.  For example, the treewidth of a graph, one of the most
fundamental graph parameters, is the minimum over all possible
completions into a chordal graph of the maximum clique size minus
one~\cite{Bodlaender98}.  The treedepth of a graph, also known as the
vertex ranking number, the ordered chromatic number, and the minimum
elimination tree height, plays a crucial role in the theory of sparse
graphs developed by Ne{\v{s}}et{\v{r}}il and Ossona de
Mendez~\cite{NesetrilOdM12}.  Mirroring the connection between
treewidth and chordal graphs, the treedepth of a graph can be defined
as the largest clique size in a completion to a \emph{trivially
  perfect graph}.  Similarly, the vertex~cover number of a graph is
equal to the minimum of the largest clique size taken over all
completions to a \emph{threshold graph}, minus one.

Recent developments have also led to subexponential parameterized
algorithms for the problems \pname{Interval
  Completion}~\cite{bliznets2014interval} and \pname{Proper Interval
  Completion}~\cite{bliznets2014proper}.  Both these problems have
strong connection to width parameters just like the ones mentioned
above: The \emph{pathwidth} of a graph is the minimum over the maximum
clique size in an \emph{interval completion} of the graph, minus one,
whereas the \emph{bandwidth} mirrors this relation for \emph{proper
  interval completions} of the graph.

%% Begin figure: C4, P4 & Dart
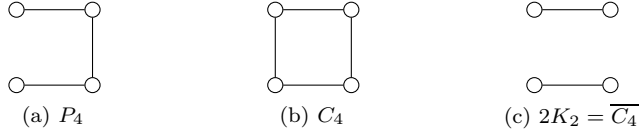
\begin{figure}[t]
  \centering \subfloat[.3\columnwidth][$P_4$] {
    \begin{tikzpicture}[every node/.style={circle, draw, scale=.6},
      scale=1]
      \fill[white] (-0.5,0.4) rectangle (1.5,0.6); \node (1) at (0,0)
      {}; \node (2) at (1,0) {}; \node (3) at (0,1) {}; \node (4) at
      (1,1) {}; \draw (1) -- (2); \draw (2) -- (4); \draw (3) -- (4);
    \end{tikzpicture}
  }\hspace{3em} \subfloat[.3\columnwidth][$C_4$]{
    \begin{tikzpicture}[every node/.style={circle, draw, scale=.6},
      scale=1]
      \fill[white] (-0.5,0.4) rectangle (1.5,0.6); \node (1) at (0,0)
      {}; \node (2) at (1,0) {}; \node (3) at (0,1) {}; \node (4) at
      (1,1) {}; \draw (1) -- (2); \draw (2) -- (4); \draw (3) -- (4);
      \draw (1) -- (3);
    \end{tikzpicture}
  }\hspace{3em} \subfloat[.3\columnwidth][$2K_2=\overline{C_4}$]{
    \begin{tikzpicture}[every node/.style={circle, draw, scale=.6},
      scale=1]
      \fill[white] (-0.5,0.4) rectangle (1.5,0.6); \node (1) at (0,0)
      {}; \node (2) at (1,0) {}; \node (3) at (0,1) {}; \node (4) at
      (1,1) {}; \draw (1) -- (2); \draw (3) -- (4);
    \end{tikzpicture}
  }
  \caption{Forbidden induced subgraphs.  \emph{Trivially perfect graphs}
    are $\{C_4, P_4\}$-free, threshold graphs are $\{2K_2,
    P_4,C_4\}$-free, and cographs are $P_4$-free.} \label{fig:3graphs}
\end{figure}

\vskip -0.5cm

\paragraph{Parameterized algorithms for completion problems}
For a long time in parameterized complexity, the main focus of studies
in \fd{} was for the case when~$\F$ was an infinite family of graphs,
e.g., \pname{Minimum Fill-in} or \pname{Interval
  Completion}~\cite{KaplanST99,NatanzonSS00,villanger2009interval}.
This was mainly due to the fact that when~$\F$ is a finite family,
\fd{} is solvable on an $n$-vertex graph in time $f(k) \cdot
n^{\bigO(1)}$ for some function~$f$ by a simple branching argument;
This was first observed by Cai~\cite{cai1996fixed}.  More precisely,
if the maximum number of non-edges in a graph from~$\F$ is~$d$, then
the corresponding \fd{} is solvable in time~$d^k \cdot n^{\bigO(1)}$.
The interest in \fd{} problems started to increase with the advance of
kernelization.  It appeared that from the perspective of
kernelization, even for the case of finite families~$\F$ the problem
is far from trivial.  Guo~\cite{guo2007problem} initiated the study of
kernelization algorithms for \fd{} in the case when the forbidden
set~$\F$ contains the graph~$C_4$, see Figure~\ref{fig:3graphs}.  (In
fact, Guo considered edge deletion problems, but they are polynomial
time equivalent to completion problems to the complements of the
forbidden induced subgraphs.)  In the literature, the most studied
graph classes containing no induced~$C_4$ are the \emph{split graphs},
i.e., $\{2K_2, C_4, C_5\}$-free graphs, \emph{threshold graphs}, i.e.,
$\{2K_2, P_4, C_4\}$-free graphs, and $\{C_4, P_4\}$-free graphs, that
is, \emph{trivially perfect graphs}~\cite{brandstadt1999graph}.  Guo
obtained polynomial kernels for the completion problems for
\emph{chain graphs}, split graphs, threshold graphs and trivially
perfect graphs and concluded that, as a consequence of his polynomial
kernelization, the corresponding \fd{} problems: \pname{Chain
  Completion}, \pname{Split Completion}, \pname{Threshold Completion}
and \pname{Trivially Perfect Completion} are solvable in times
$\bigO(2^k + mnk)$, $\bigO(5^k + m^4n)$, $\bigO(4^k + kn^4)$, and
$\bigO(4^k + kn^4)$, respectively.

The work on kernelization of \fd{} problems was continued by Kratsch
and Wahlstr{\"o}m~\cite{kratsch2009two} who showed that there exists a
set $\F$ consisting of one graph on seven vertices for which \fd{}
does not admit a polynomial kernel.  Guillemot et
al.~\cite{guillemot2013non} showed that \pname{Cograph Completion},
i.e., the case $\F=\{ P_4\}$, admits a polynomial kernel, while for
$\F=\{ \overline{P_{13}}\}$, the complement of a path on 13 vertices,
\fd{} has no polynomial kernel.  These results were significantly
improved by Cai and Cai~\cite{cai2013incompressibility}: For $\F = \{
P_\ell \}$ or $\F = \{ C_\ell \}$, the problems \fd{} and
\pname{$\F$-Edge Deletion} admit a polynomial kernel if and only if
the forbidden graph has at most three edges.

\begin{figure}[h]
  \centering
  \begin{tabular}{l | l | l }
    Obstruction set $\F$& Graph class name            & Complexity\\
    \hline
    $C_4,C_5, C_6, \dots$ & Chordal                &
    \cclass{SUBEPT}~\cite{fomin2012subexponential}\\
    $C_4,P_4$     & Trivially Perfect          &
    \cclass{SUBEPT}~(Theorem~\ref{thm:co-trivially-perfect-subept})\\
    $2K_2,C_4,C_5$ & Split                     &
    \cclass{SUBEPT}~\cite{ghosh2012faster}\\
    $2K_2,C_4,P_4$ & Threshold                 &
    \cclass{SUBEPT}~(Theorem~\ref{thm:threshold-completion-subept})\\
    $2K_2,C_4$     & Pseudosplit               &
    \cclass{SUBEPT}~(Theorem~\ref{thm:pseudosplit-completion-subept})\\
    $\overline{P_3}, K_t$, $t=o(k)$ & Co-$t$-cluster &
    \cclass{SUBEPT}~\cite{fomin2011subexponential} \\
    $\overline{P_3}$ & Co-cluster                   &
    \cclass{E}~\cite{komusiewicz2012cluster}\\
    $2K_2$         & $2K_2$-free               &
    \cclass{E}~(Theorem~\ref{thm:2k2-completion-exp}) \\
    $C_4$          & $C_4$-free                &
    \cclass{E}~(Theorem~\ref{thm:c4completion-exp}) \\
    $P_4$          & Cograph                   &
    \cclass{E}~(Theorem~\ref{thm:p4-completion-exp})\\
    $2K_2,P_4$      & Co-Trivially Perfect     &
    \cclass{E}~(Theorem~\ref{thm:2k2-p4-completion-exp})
  \end{tabular}
  \caption{Known subexponential complexity of \fd{} for different
    sets~$\F$.  All problems in this table are \cclass{NP}-hard and in
    \cclass{FPT}.  The entry \cclass{SUBEPT} means the problem is
    solvable in subexponential time $2^{o(k)} n^{\bigO(1)}$
    whereas~\cclass{E} means that the problem is not solvable in
    subexponential time unless~ETH
    fails.}  \label{fig:table-completion-complexity}
\end{figure}

It appeared recently that for some choices of $\F$, \fd{} is solvable
in \emph{subexponential} time.  The exploration of this phenomenon is
the main motivation for our research on this problem.  The last
chapter of Flum and Grohe's textbook on parameterized complexity
theory~\cite[Chapter~16]{flum2006parameterized} concerns
subexponential fixed parameter tractability, the complexity class
\cclass{SUBEPT}, which, loosely speaking---we skip here some technical
conditions---is the class of problems solvable in time
$2^{o(k)}n^{\bigO(1)}$, where $n$ is the input length and $k$ is the
parameter.  Until recently, the only notable examples of problems in
\cclass{SUBEPT} were problems on planar graphs, and more generally, on
graphs excluding some fixed graph as a
minor~\cite{demaine2005subexponential}.  In 2009, Alon et
al.~\cite{alon2009fast} used a novel application of color coding,
dubbed \emph{chromatic coding}, to show that parameterized
\pname{Feedback Arc Set in Tournaments} is in \cclass{SUBEPT}.  As
Flum and Grohe~\cite{flum2006parameterized} observed, for most of the
natural parameterized problems, already the classical
\cclass{NP}-hardness reductions can be used to refute the existence of
subexponential parameterized algorithms, unless the following
well-known complexity hypothesis formulated by Impagliazzo, Paturi,
and Zane~\cite{impagliazzo2001which} fails.

%%% ENVIRONMENT FOR ETH (used in intro and lower-bounds)

\medskip

\textsc{Exponential Time Hypothesis \textit{(ETH)}}.
\textit{There exists a positive real number $s$ such that 3-CNF-SAT
  with $n$ variables cannot be solved in time $2^{sn}$.}

\medskip

%%% END ENVIRONMENT FOR ETH

Thus, it is most likely that the majority of parameterized problems
are not solvable in subexponential parameterized time and until very
recently no natural parameterized problem solvable in subexponential
parameterized time on general graphs was known.  A subset of the
authors recently showed that \pname{Minimum Fill-in}, also known as
\pname{Chordal Completion}, which is equivalent to \fd{} with $\F$
consisting of cycles of length at least four, is in
\cclass{SUBEPT}~\cite{fomin2012subexponential}, simultaneously
establishing that \pname{Chain Completion} is solvable in
subexponential time.  Later, Ghosh et al.~\cite{ghosh2012faster}
showed that \pname{Split Completion} is solvable in subexponential
time.  On the other hand, Komusiewicz and
Uhlmann~\cite{komusiewicz2012cluster}, showed that an edge
modification problem known as \pname{Cluster Deletion} does not belong
to \cclass{SUBEPT} unless ETH fails.  Note that \pname{Cluster
  Deletion} is equivalent to \fd{} when $\F=\{\overline{P_3}\}$, the
complement of the path $P_3$.  On the other hand, it is interesting to
note that by the result of Fomin et
al.~\cite{fomin2011subexponential}, \pname{Cluster Deletion into $t$
  Clusters}, i.e., the complement problem for \fd{} for
$\F=\{\overline{P_3}, K_t\}$, is in \cclass{SUBEPT} for $t=o(k)$.

\paragraph{Our results}
In this work we extend the class of \fd{} problems admitting
subexponential time algorithms, see
Figure~\ref{fig:table-completion-complexity}.  Our main algorithmic
result is the following:
\begin{quote}
  \pname{Trivially Perfect Completion} is solvable in time
  $2^{\bigO(\sqrt{k}\log{k})} n^{\bigO(1)}$ and is thus in
  \cclass{SUBEPT}.
\end{quote}
This problem is the \fd{} problem for $\F=\{C_4, P_4\}$.

On a very high level, our algorithm is based on the same strategy as
the algorithm for completion into chordal
graphs~\cite{fomin2012subexponential}.  Just like in that algorithm,
we enumerate subexponentially many special objects, here called
\emph{trivially perfect potential maximal cliques} which are the
maximal cliques in some minimal completion into a trivially perfect
graph that uses at most~$k$ edges.  As far as we succeed in
enumerating these objects, we apply dynamic programming in order to
find an optimal completion.  But here the similarities end.  To
enumerate trivially perfect potential maximal cliques (henceforth
referred to as only \emph{potential maximal cliques}) for trivially
perfect graphs, we have to use completely different structural
properties from those used for the case of chordal graphs.

We also show that within the same running time, the \fd{} problem is
solvable for $\F=\{2K_2,C_4\}$, and $\F=\{2K_2, P_4, C_4\}$.  This
corresponds to completion into threshold and pseudosplit graphs,
respectively.  Let us note that combined with the results of Fomin and
Villanger~\cite{fomin2012subexponential} and Ghosh et
al.~\cite{ghosh2012faster}, this implies that all four problems
considered by Guo in~\cite{guo2007problem} are in \cclass{SUBEPT}, in
addition to admitting a polynomial kernel.  We finally complement our
algorithmic findings by showing the following:
\begin{quote}
  For $\F=\{2K_2\}$, $\F=\{C_4\}$, $\F=\{P_4\}$ and $\F=\{2K_2,P_4\}$,
  the \fd{} problem cannot be solved in time $2^{o({k})} n^{\bigO(1)}$
  unless ETH fails.
\end{quote}
Thus, we obtain a complete classification for all $\F \subseteq
\{2K_2, P_4, C_4\}$.

\paragraph{Organization of the paper}
In Section~\ref{sec:cotrivially-perfect-graphs} we give some
structural results about trivially perfect graphs and their
completions, and give the main result of the paper: an algorithm
solving \pname{Trivially Perfect Completion} in subexponential time.
This section also contains some structural results on trivially
perfect graphs that might be interesting on its own.  In
Sections~\ref{sec:threshold-completion}
and~\ref{sec:pseudosplit-completion} we give subexponential time
algorithms for \pname{Threshold Completion} and \pname{Pseudosplit
  Completion}.

In Section~\ref{sec:lower-bounds}, we give the lower bounds on \fd{}
when~$\F$ is $\{2K_2\}$, $\{C_4\}$, $\{P_4\}$, and $\{2K_2, P_4\}$.
Finally, in Section~\ref{sec:conclusion} we give some concluding
remarks and state some interesting remaining questions and future
directions.

\paragraph{Notation and preliminaries on parameterized complexity}
We consider only finite simple undirected graphs.  We use $n_G$ to
denote the number of vertices and $m_G$ the number of edges in a
graph~$G$.  If $G = (V,E)$ is a graph, and $A,B \subseteq V$, we write
$E(A,B)$ for the edges with one endpoint in~$A$ and the other in~$B$,
and we write $E(A) = m_A = m_{G[A]}$ for the edges inside~$A$.

Given a graph $G = (V,E)$, recall that $N_G(v)$ for a vertex $v \in V$
denotes the set of neighbors of~$v$ in~$G$.  We write $N_G[v]$ to mean
the set $N_G(v) \cup \{v\}$.  For sets of vertices $U \subseteq V$, we
write $N_G(U)$ to denote the open neighborhood $\bigcup_{v \in
  U}(N_G(v)) \setminus U$, and $N_G[U] = N_G(U) \cup U$ to denote the
closed neighborhood.  For a set of pairs of vertices $S$, we write $G
+ S = (V, E \cup S)$ and if $U \subseteq V$ is a set of vertices, then
$G - U = G[V \setminus U]$.  We will skip the subscripts when this
will not cause any confusion.

A \emph{universal vertex} in a graph~$G$ is a vertex~$v$ such that
$N[v] = V(G)$.  Let $\uni(G)$ denote the set of universal vertices
of~$G$.  Observe that $\uni(G)$, when non-empty, is always a clique,
and we will refer to it as the \emph{(maximal) universal clique}.  The
maximal universal cliques play an important role in the trivially
perfect graphs; They are the main building blocks we will use to
achieve the algorithm.

We here provide a simplified definition of parameterized problems,
kernels and the class of parameterized subexponential time algorithms.
A \emph{parameterized problem}~$\Pi$ is a problem whose input is a
pair $(x,k)$, where $k \in \mathbb{N}$.  The problem $\Pi$ is
\emph{fixed-parameter tractable}, and thus belongs to the
class~\cclass{FPT}, if there is an algorithm solving this problem in
time $f(k) \cdot x^{O(1)}$ for some function~$f$, depending only
on~$k$.  A \emph{kernelization algorithm} for~$\Pi$ is a polynomial
time algorithm which on input $(x,k)$ gives an output $(x',k')$ such
that $|x'| \leq g(k)$ and $k'\leq g(k)$ for some function~$g$
depending only on~$k$, and such that $(x,k)$ is a \yes{} instance
for~$\Pi$ if and only if $(x',k')$ is a \yes{} instance for~$\Pi$.  We
call the output the \emph{kernel}.  We say a problem admits a
\emph{polynomial kernel} if the function~$g$ is polynomial.

The complexity class \cclass{SUBEPT} is contained in \cclass{FPT}; It
is the class of problems $\Pi$ for which there exists an algorithm
with running time $2^{o(k)} \cdot n^{O(1)}$.  That is, the parameter
function~$f$ is \emph{subexponential}.  Note that if the exponential
time hypothesis is true, then \cclass{SUBEPT}~$\subsetneq$~\cclass{FPT}.

\section{Completion to trivially perfect graphs}
\label{sec:cotrivially-perfect-graphs}
In this section we study the \pname{Trivially Perfect Completion}
problem which is \fd{} for $\F = \{C_4, P_4\}$.  The decision version
of the problem was shown to be \cclass{NP}-complete by
Yannakakis~\cite{yannakakis1981computing}.  As already stated in the
introduction, trivially perfect graphs are characterized by a finite
set of forbidden induced subgraphs, and thus it follows from
Cai~\cite{cai1996fixed} that the problem also is fixed parameter
tractable, i.e., it belongs to the class \cclass{FPT}.

The main result of this section is the following theorem:
\begin{theorem}\label{thm:main-algo}
  \label{thm:co-trivially-perfect-subept} For an input $(G,k)$,
  \pname{Trivially Perfect Completion} is solvable in time
  $2^{\bigO(\sqrt{k}\log{k})} + \bigO(kn^{4})$.
\end{theorem}

Throughout this section, an edge set~$S$ is called a \emph{completion}
for~$G$ if $G+S$ is trivially perfect.  Furthermore, a completion~$S$
is called a \emph{minimal completion} for~$G$ if no proper subset
of~$S$ is a completion for $G$.  The main outline of the algorithm is
as follows:
\begin{enumerate}[Step A:]
\item On input $(G, k)$, we first apply the algorithm by
  Guo~\cite{guo2007problem} to obtain a kernel $\bigO(k^3)$ vertices.
  The running time of this algorithm is $\bigO(kn^{4})$.  The
  kernelization algorithm of Guo can only reduce the parameter, i.e.,
  $k'\leq k$ where $k'$ is the new parameter.  Moreover, the output
  kernel is in fact of size $\bigO(k'^3)$.  Therefore, due to this
  preprocessing step we may assume without loss of generality that we
  work on an instance $(G,k)$ with $|V(G)|\leq \bigO(k^3)$.
\item Assuming our input instance has $\bigO(k^3)$ vertices, we show
  how to generate all special vertex subsets of the kernel which we
  call \emph{\vpmc s} in time $2^{\bigO(\sqrt{k}\log{k})}$.  A \vpmc{}
  $\Omega \subseteq V(G)$ is a vertex subset which is a maximal clique
  in some minimal completion of size at most~$k$.
\item Using dynamic programming, we show how to compute an optimal
  solution or to conclude that $(G,k)$ is a \no{} instance, in time
  polynomial in the number of \vpmc s.
\end{enumerate}

\subsection{Structure of trivially perfect graphs}

Apart from the aforementioned characterization by forbidden induced
subgraphs, an inherently local characterization, several other
equivalent definitions of trivially perfect graphs are known.  These
definitions reveal more structural properties of this graph class
which will be essential in our algorithm.  Therefore, before
proceeding with the proof of
Theorem~\ref{thm:co-trivially-perfect-subept}, we establish a number
of results on the global structure of trivially perfect graphs and
minimal completions which will be useful.

The trivially perfect graphs have a rooted decomposition tree, which
we call a \emph{universal clique decomposition}, in which each node
corresponds to a maximal set of vertices that all are universal for
the graph induced by the vertices in the subtree rooted at this node.
This decomposition is similar to that of a \emph{treedepth
  decomposition}.  We refer to Figure~\ref{fig:structure-of-pmc} for
an example of the concepts that we introduce next.  The following
recursive definition is often used as an alternative definition of
trivially perfect graphs.

\begin{proposition}[\cite{jing1996quasi}]
  \label{def:tp-recursive}
  The class of trivially perfect graphs can be defined recursively as
  follows:
  \begin{itemize}
  \item $K_1$ is a trivially perfect graph.
  \item Adding a universal vertex to a trivially perfect graph results
    in a trivially perfect graph.
  \item The disjoint union of two trivially perfect graphs is a
    trivially perfect graph.
  \end{itemize}
\end{proposition}

Let~$T$ be a rooted tree and~$t$ be a node of~$T$.  We denote by~$T_t$
the maximal subtree of~$T$ rooted in~$t$.  We can now use the
universal clique $\uni(G)$ of a trivially perfect graph $G = (V,E)$ to
make a decomposition structure.

\begin{definition}[Universal clique decomposition]
  \label{def:univeral-clique-decomposition}
  A \emph{universal clique decomposition} of a connected trivially
  perfect graph $G = (V, E)$ is a pair $(T = (V_T, E_T), \mathcal{B} =
  \{B_{t}\}_{t \in V_T})$, where~$T$ is a rooted tree
  and~$\mathcal{B}$ is a partition of the vertex set~$V$ into disjoint
  non-empty subsets, such that
  \begin{itemize}
  \item if $vw \in E(G)$ and $v \in B_t$ and $w \in B_s$, then~$s$ and~$t$
    are on a path from a leaf to the root, with possibly $s = t$, and
  \item for every node $t \in V_T$, the set of vertices $B_t$ is the
    maximal universal clique in the subgraph $G[\bigcup_{s \in V(T_t)}
    B_s]$.
  \end{itemize}
\end{definition}
We call the vertices of~$T$ \emph{nodes} and the sets in~$\mathcal{B}$
\emph{bags} of the universal clique decomposition $(T, \mathcal{B})$.
By slightly abusing the notation, we often do not distinguish between
nodes and bags.  Note that by the definition, in a universal clique
decomposition every non-leaf node has at least two children, since
otherwise the universal clique contained in the corresponding bag
would not be maximal.

\begin{lemma}
  \label{lem:trivial_complete_part}
  A connected graph~$G$ admits a universal clique decomposition if and
  only if it is trivially perfect.  Moreover, such a decomposition is
  unique up to isomorphisms.
\end{lemma}
\begin{proof}
  From right to left, we proceed by induction on the number of
  vertices using Proposition~\ref{def:tp-recursive}.  The base case is
  when we have one vertex,~$K_1$ which is a trivially perfect graph
  and also admits a unique universal clique decomposition.  The
  induction step is when we add a vertex~$v$, and by the definition of
  trivially perfect graphs,~$v$ is a universal vertex.  Either we add
  a universal vertex to a connected trivially perfect graph, in which
  case we simply add the vertex to the root bag, or we add a universal
  vertex to the disjoint union of two or more trivially perfect
  graphs.  In this case, we create a new tree, with~$r_v$ being the
  root connected to the root of each of the trees for the disjoint
  union.  Since~$v$ is the only universal vertex in the graph, the
  constructed structure is a universal clique decomposition.  Observe
  that the constructed decompositions are unique (up to isomorphisms).
  
  From left to right, we proceed by induction on the height of the
  universal clique decomposition.  Suppose $(T, \mathcal{B})$ is a
  universal clique decomposition of a graph~$G$.  Consider the case
  when~$T$ has height~$1$, i.e., we have only one single tree node
  (and one bag).  Then this bag, by
  Proposition~\ref{def:tp-recursive}, is a clique (every vertex in the
  bag is universal), and since a complete graph is trivially perfect,
  the base case holds.  Consider now the case when~$T$ has height at
  least~$2$.  Let~$r$ be the root of~$T$, and let $x_1,x_2,\ldots,x_p$
  be children of~$r$ in~$T$.  Observe that the tree~$T_{x_i}$ is a
  universal clique decomposition for the graph $G[\bigcup_{t\in
    V(T_{x_i})} B_t]$ for each $i=1,2,\ldots,p$.  Hence, by the
  induction hypothesis we have that $G[\bigcup_{t\in V(T_{x_i})} B_t]$
  is trivially perfect.  To see that $G$ is trivially perfect as well,
  observe that~$G$ can be obtained by taking the disjoint union of
  graphs $G[\bigcup_{t\in V(T_{x_i})} B_t]$ for $i=1,2,\ldots,p$, and
  adding $|B_r|$ universal vertices.
\end{proof}

For the purposes of the dynamic programming procedure, we define the
following notion.

\begin{definition}[Block]
  \label{def:tp-block}
  Let $(T=(V_T,E_T), \mathcal{B}=\{B_{t}\}_{t\in V_T})$ be the universal
  clique decomposition of a connected trivially perfect graph $G =
  (V,E)$.  For each node $t \in V_T$, we associate a \emph{block} $L_t
  = ( B_t, D_t)$, where
  \begin{itemize}
  \item $B_t$ is the subset of $V$ contained in the bag corresponding
    to $t$, and
  \item $D_t$ is the set of vertices of $V$ contained in the bags
    corresponding to the nodes of the subtree $T_t$.
  \item The \emph{tail} of a block $L_t$ is the set of vertices $Q_t$
    contained in the bags corresponding to the nodes of the path from
    $t$ to $r$ in $T$, where $r$ is the root of $T$, including~$B_t$
    and~$B_r$.
  \end{itemize}
\end{definition}

%%%%%%%%%%%%%%%%%%%%%%%%% HERE IS FIGURE OF BLOCK AND DECOMPOSITION

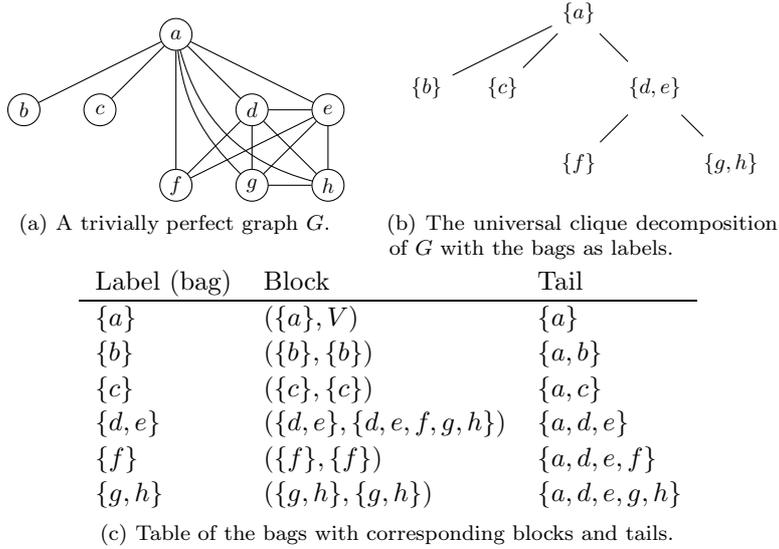
\begin{figure}[t]
  \centering
  \subfloat[A trivially perfect graph~$G$.] {
    \centering
    \begin{tikzpicture}[every
      node/.style={circle,draw,scale=.8},scale=1, minimum size=1.5em,
      inner sep=0pt]
      
      \node (a) at (2,2) {$a$};
      \node (b) at (0,1) {$b$};
      \node (c) at (1,1) {$c$};
      \node (d) at (3,1) {$d$};
      \node (e) at (4,1) {$e$};
      \node (f) at (2,0) {$f$};
      \node (g) at (3,0) {$g$};
      \node (h) at (4,0) {$h$};
      
      \draw (b) -- (a) -- (c);
      \draw (a) -- (f);
      \draw (a) to[out=-83] (g);
      \draw (a) to[bend right] (h);
      \draw (a) -- (d);
      \draw (a) -- (e);
      
      \draw (d) -- (f) -- (e);
      \draw (d) -- (g) -- (h) -- (e) -- (d);
      \draw (d) -- (h);
      \draw (e) -- (g);
    \end{tikzpicture}
    \label{fig:structure-of-pmc-graph}
  }\hspace{1em}
  \subfloat[The universal clique decomposition of~$G$ with the bags as
  labels.] {
    \centering
    \begin{tikzpicture}[every node/.style={circle, scale=.8}, scale=1]
      
      \node (a)  at (2,2) {$\{a\}$};
      \node (b)  at (0,1) {$\{b\}$};
      \node (c)  at (1,1) {$\{c\}$};
      \node (de) at (3,1) {$\{d,e\}$};
      
      \node (f)  at (2,0) {$\{f\}$};
      
      \node (gh) at (4,0) {$\{g,h\}$};
      
      \draw (b) -- (a) -- (c);
      
      \draw (a) -- (de) -- (gh);
      
      \draw (de) -- (f);
    \end{tikzpicture}
    \label{fig:tree-structure-of-pmc-tree}
  }\hspace{1em}
  \subfloat[Table of the bags with corresponding blocks and tails.] {
    \centering
    \begin{tabular}{l l l}
      Label (bag) & Block & Tail\\
      \hline
      $\{a\}$ & $(\{a\}, V)$ & $\{a\}$ \\
      $\{b\}$ & $(\{b\}, \{b\})$ & $\{a,b\}$ \\
      $\{c\}$ & $(\{c\}, \{c\})$ & $\{a,c\}$ \\
      $\{d,e\}$ & $(\{d,e\}, \{d,e,f,g,h\})$ & $\{a,d,e\}$ \\
      $\{f\}$ & $(\{f\}, \{f\})$ & $\{a,d,e,f\}$ \\
      $\{g,h\}$ & $(\{g,h\}, \{g,h\})$ & $\{a,d,e,g,h\}$
    \end{tabular}
  }
  \caption{In the first figure, we have a trivially perfect graph, and in
    the second, a \emph{universal clique decomposition} of the graph
    with the bags as labels.  Finally we have a table of the bags and
    the corresponding blocks and tails.  Notice that for a block
    $(B,D)$ and tail $Q$, $B \subseteq D$ and $B \subseteq Q$.
    Furthermore, in any leaf block it holds that $B = D$, and in the
    root block it holds that $D=V$.}
  \label{fig:structure-of-pmc}
\end{figure}

%%%%%%%%%%%%%%%%%%%%%%%%% HERE ENDS FIGURE OF BLOCK AND DECOMPOSITION

When~$t$ is a leaf of~$T$, we have that $B_t=D_t$ and we call the
block $L_t = (B_t,D_t)$ a \emph{leaf block}.  If $t$ is the root, we
have that $D_t=V(G)$ and we call~$L_t$ the \emph{root block}.
Otherwise, we call $L_t$ an \emph{internal block}.  Observe that for
every block $L_t = (B_t,D_t)$ with tail $Q_t$ we have that $B_t
\subseteq Q_t$, $B_t \subseteq D_t$, and $D_t\cap Q_t = B_t$, see
Figure~\ref{fig:structure-of-pmc}.  Note also that~$Q_t$ is a clique
and the vertices of~$Q_t$ are universal to $D_t\setminus B_t$.  The
following lemma summarizes the properties of universal clique
decompositions, maximal cliques, and blocks used in our proof.
\begin{lemma}
  \label{lemma:leaf_part}
  Let $(T, \mathcal{B})$ be the universal clique decomposition of a
  connected trivially perfect graph $G$ and let $L=(B,D)$ be a block
  with~$Q$ as its tail.
  \begin{enumerate}[(i)]
  \item\label{lemma:leaf_part_i} If $L$ is a leaf block, then $Q =
    N_G[v]$ for every $v \in B$.
  \item\label{lemma:leaf_part_ii} The following are equivalent:
    \begin{enumerate}[(1)]
    \item $L$ is a leaf block,
    \item $D = B$, and
    \item $Q$ is a maximal clique of $G$.
    \end{enumerate}
  \item\label{lemma:leaf_part_iii} If $L$ is a non-leaf block, then
    for every two vertices $u,v$ from different connected components
    of $G[D \setminus B]$, we have that $Q= N_G(u) \cap N_G(v)$.
  \end{enumerate}
\end{lemma}
\begin{proof}
  \textit{(\ref{lemma:leaf_part_i})} Since~$Q$ is a clique, we have
  that $Q \subseteq N_G[v]$.  On the other hand, since $v \in B$
  and~$L$ is a leaf block, we have that $Q \supseteq N_G[v]$ by the
  definition of universal clique decomposition.
  
  \textit{(\ref{lemma:leaf_part_ii})} We prove the chain $(1) \to (2)
  \to (3) \to (1)$.  Suppose that~$L$ is a leaf block, and~$D$ is the
  set of vertices in the bags in the subtree rooted at~$L$, then $B =
  D$.  Then by~\textit{(\ref{lemma:leaf_part_i})} we have that
  $N_G[v]=Q$ for any $v \in B$, hence~$Q$ is maximal.  Finally, if~$Q$
  is a maximal clique in the graph, i.e., it cannot be extended, by
  definition~$L$ cannot have any children so~$L$ must be a leaf block.
  
  \textit{(\ref{lemma:leaf_part_iii})} Suppose $L = (B, D)$ is a
  non-leaf block and~$D_1$ and~$D_2$ are two connected components of
  $G' = G[D \setminus B]$.  Let $v \in D_1$ and $u \in D_2$ and
  observe that since they are in different connected components of
  $G[D \setminus B]$, $N_{G'}(v) \cap N_{G'}(u) = \emptyset$.  By the
  universality of~$Q$, the result follows: $Q = N_G(v) \cap N_G(u)$.
\end{proof}

\subsection{Structure of minimal completions}
Before we proceed with the algorithm, we provide some properties of
minimal completions.  The following lemma gives insight to the
structure of a \yes{} instance.
\begin{lemma}
  \label{lem:connectedness-of-minimal-completion}
  Let $G = (V, E)$ be a connected graph, $S$ a minimal completion and
  $H = G+S$.  Suppose $L = (B, D)$ is a block in some universal clique
  decomposition of $H$ and denote by $D_1, D_2,\dots , D_\ell$ the
  connected components of $H[D]- B$.
  \begin{enumerate}[(i)]
  \item\label{lem:conn-min-com:1} If $L$ is not a leaf block, then
    $\ell > 1$;
  \item\label{lem:conn-min-com:2} If $\ell > 1$, then in $G$ every
    vertex $v \in B$ has at least one neighbor in each set $D_1, D_2,\dots
    , D_\ell$;
  \item\label{lem:conn-min-com:3} The graph $G[D_i]$ is connected for
    every $i\in \{1,\dots, \ell\}$;
  \item\label{lem:conn-min-com:4} For every $i\in \{1,\dots, \ell\}$, $B
    \subseteq N_G(D \sm (B \cup D_i))$.
  \end{enumerate}
\end{lemma}
\begin{proof} We prove this case by case.
  % 1: If $(Q, B, D)$ is not a leaf block, then $l > 1$,
  \textit{(\ref{lem:conn-min-com:1})} Let $(B,D)$ be a non-leaf block.
  Since~$B$ is maximal,~$D$ is not a clique, so by the recursive
  definition of trivially perfect graphs, $H[D]-B$ is the disjoint
  union of two or more trivially perfect graphs, hence $\ell > 1$.
  
  % 2: If $\ell > 1$, then in $G$ every vertex $v \in B$ has at least
  % one neighbor in each set $D_1, D_2,\dots , D_\ell$,
  \textit{(\ref{lem:conn-min-com:2})} Suppose, without loss of
  generality, that there exists a vertex $v\in B$ that has no neighbor
  in~$D_1$.  Let $S'=S\setminus (\{v\}\times V(D_1))$; Note that
  since~$v$ is universal to $V(D_1)$ in $H$ and completely
  non-adjacent to~$V(D_1)$ in~$G$, then $\{v\}\times V(D_1)\subseteq
  S$ and~$S'$ is a proper subset of~$S$.  We claim that $H'=G+S'$ is
  also a trivially perfect graph, which contradicts the minimality
  of~$S$.  Indeed, consider a universal clique decomposition obtained
  from the universal clique decomposition of~$H$ by \textit{(a)}, in
  case $\ell=2$, moving~$v$ from~$B$ to the root bag of~$D_2$, or
  \textit{(b)}, in case $\ell>2$, moving~$v$ from~$B$ to a new bag
  $B'=\{v\}$ attached below~$B$, with all the root bags of
  $D_2,D_3,\ldots,D_\ell$ re-attached from below~$B$ to below~$B'$.
  It can be easily seen that this new universal clique decomposition
  is indeed a universal clique decomposition of~$H'$, which proves
  that~$H'$ is trivially perfect.
  
  % 3 $G[D_j]$ is connected for every $j \leq l$.
  \textit{(\ref{lem:conn-min-com:3})} For the sake of a contradiction,
  suppose $G[D_a]$ was disconnected.  Let $(D_{a_1}, D_{a_2})$ be a
  partition of $D_a$ such that there is no edge between $D_{a_1}$ and
  $D_{a_2}$ in $G$.  Clearly, $H[D_{a_1}]$ and $H[D_{a_2}]$ are
  trivially perfect graphs as induced subgraphs of $H$, hence they
  admit some universal clique decompositions.  Since $H[D_a]$ is
  connected, we infer that $S$ contains some edges between $D_{a_1}$
  and $D_{a_2}$.  Let now $S' = S \setminus \{uv \mid u \in D_{a_1}, v
  \in D_{a_2}, uv \in S \}$; By the previous argument we have that $S'
  \subsetneq S$.  Modify now the given universal clique decomposition
  of $H$ by removing the subtree below $B$ that corresponds to $D_a$,
  and attaching instead two subtrees below $B$ that are universal
  clique decompositions of $H[D_{a_1}]$ and $H[D_{a_2}]$.  Observe
  that thus we obtain a universal clique decomposition of $G+S'$,
  which shows that $G+S'$ is trivially perfect.  This is a
  contradiction with the minimality of $S$.
  
  % for every $i \leq \ell$, $B \subseteq N_G(D \sm (B \cup D_i))$.
  \textit{(\ref{lem:conn-min-com:4})} Follows directly from
  \textit{(\ref{lem:conn-min-com:1})} and
  \textit{(\ref{lem:conn-min-com:2})}: if $\ell>0$, then $\ell>1$ and
  every vertex of $B$ has edges in $G$ to all different connected
  components of $D \setminus B$.
\end{proof}

%%%%%%%%%%%%%%%%%%%%%%%%%%%%%%%%%%%%%%%%%%%%%%%%%%%%%%%% 
\subsection{The algorithm}

%%%%%%%%%%%%%%%%%%%%% FIGURE OF POTENTIAL MAXIMAL CLIQUES

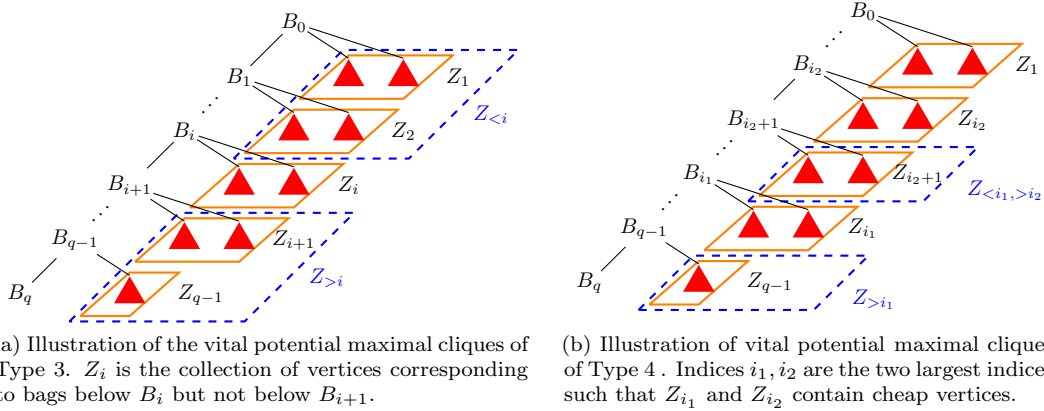
\begin{figure}[t]
  \centering
  \subfloat[Illustration of the \vpmc s of Type~3.  $Z_i$ is the
  collection of vertices corresponding to bags below $B_i$ but not
  below $B_{i+1}$.]{
    \centering
    \begin{tikzpicture}[every node/.style={circle, scale=.8}, scale=.72,
      minimum size=1.3em, inner sep=0pt]

      \foreach \x/\y in {6/4,3/1}
      {
        \draw[blue, dashed, thick] (\x-2.1,\y-1.5) -- (\x-0.1,\y+0.5) -- (\x+1+2.1,\y+0.5) -- (\x+1+0.1,\y-1.5) -- (\x-2.1,\y-1.5);
      }

      \foreach \x/\y/\z in {6/4/1,5/3/1,4/2/1,3/1/1,2/0/0}
      {      
      	\draw [orange,thick] (\x-0.9,\y-0.4) -- (\x,\y+0.4) -- (\x+\z+0.9,\y+0.4) -- (\x+\z,\y-0.4) -- (\x-0.9,\y-0.4);
      }

      \node [draw=none] at (8,4) {$Z_1$};

      \node (b1)  at (5,5) {$B_0$};
      \node[fill=red,regular polygon, regular polygon sides=3,inner sep=3pt] (z11) at (6,4) {};
      \node[fill=red,regular polygon, regular polygon sides=3,inner sep=3pt] (z12) at (7,4) {};

      \node [draw=none] at (7,3) {$Z_2$};
      \node [draw=none,blue] at (8.6,3.3) {$Z_{<i}$};

      \node (b2)  at (4,4) {$B_1$};
      \node[fill=red,regular polygon, regular polygon sides=3,inner sep=3pt] (z21) at (5,3) {};
      \node[fill=red,regular polygon, regular polygon sides=3,inner sep=3pt] (z22) at (6,3) {};    
      
      \node [draw=none,rotate=45] at (3.5,3.5) {$\cdots$};
      
      \node [draw=none] at (6,2) {$Z_i$};
      
      \node (bi)  at (3,3) {$B_i$};
      \node[fill=red,regular polygon, regular polygon sides=3,inner sep=3pt] (zi1) at (4,2) {};
      \node[fill=red,regular polygon, regular polygon sides=3,inner sep=3pt] (zi2) at (5,2) {};

      \node [draw=none] at (5,1) {$Z_{i+1}$};
      
      \node (bi1) at (2,2) {$B_{i+1}$};
      \node[fill=red,regular polygon, regular polygon sides=3,inner sep=3pt] (zi11) at (3,1) {};
      \node[fill=red,regular polygon, regular polygon sides=3,inner sep=3pt] (zi12) at (4,1) {};
      
      \node [draw=none,rotate=45] at (1.5,1.5) {$\cdots$};    
      
      \node [draw=none] at (3.3,0) {$Z_{q-1}$};
      \node [draw=none,blue] at (5.6,0.3) {$Z_{>i}$};

      \node (bt1) at (1,1) {$B_{q-1}$};
      \node[fill=red,regular polygon, regular polygon sides=3,inner sep=3pt] (zt11) at (2,0) {};
      
      \node (bt) at (0,0) {$B_q$};
      
      \draw (b1) -- (b2);
      \draw (b1) -- (z11.north);
      \draw (b1) -- (z12.north);
      
      \draw (b2) -- (z21.north);
      \draw (b2) -- (z22.north);
      
      \draw (bi) -- (bi1);
      \draw (bi) -- (zi1.north);
      \draw (bi) -- (zi2.north);
      
      \draw (bi1) -- (zi11.north);
      \draw (bi1) -- (zi12.north);
      
      \draw (bt1) -- (bt);
      \draw (bt1) -- (zt11.north);
      
    \end{tikzpicture}
    \label{fig:vpmc-tree-zi}
  }
  \hspace{1em}
  \subfloat[t][Illustration of \vpmc s of Type~4 .  Indices $i_1,i_2$
  are the two largest indices such that $Z_{i_1}$ and $Z_{i_2}$
  contain cheap vertices.]{
    \begin{tikzpicture}[every node/.style={circle, scale=.8}, scale=.72,
      minimum size=1.3em, inner sep=0pt]
      
      \foreach \x/\y in {4/2,2/0}
      {
        \draw[blue, dashed, thick] (\x-1.1,\y-0.5) -- (\x-0.1,\y+0.5) -- (\x+1+1+1.1,\y+0.5) -- (\x+1+1+0.1,\y-0.5) -- (\x-1.1,\y-0.5);
      }

      \foreach \x/\y/\z in {6/4/1,5/3/1,4/2/1,3/1/1,2/0/0}
      {      
      	\draw [orange,thick] (\x-0.9,\y-0.4) -- (\x,\y+0.4) -- (\x+\z+0.9,\y+0.4) -- (\x+\z,\y-0.4) -- (\x-0.9,\y-0.4);
      }

      \node [draw=none] at (8,4) {$Z_1$};

      \node (b1)  at (5,5) {$B_0$};
      \node[fill=red,regular polygon, regular polygon sides=3,inner sep=3pt] (z11) at (6,4) {};
      \node[fill=red,regular polygon, regular polygon sides=3,inner sep=3pt] (z12) at (7,4) {};

      \node [draw=none,rotate=45] at (4.5,4.5) {$\cdots$};    

      \node [draw=none] at (7,3) {$Z_{i_2}$};

      \node (b2)  at (4,4) {$B_{i_2}$};
      \node[fill=red,regular polygon, regular polygon sides=3,inner sep=3pt] (z21) at (5,3) {};
      \node[fill=red,regular polygon, regular polygon sides=3,inner sep=3pt] (z22) at (6,3) {};

      \node [draw=none] at (6,2) {$Z_{i_2+1}$};
      \node [draw=none,blue] at (7.6,1.7) {$Z_{<i_1,>i_2}$};
      
      \node (bi)  at (3,3) {$B_{i_2+1}$};
      \node[fill=red,regular polygon, regular polygon sides=3,inner sep=3pt] (zi1) at (4,2) {};
      \node[fill=red,regular polygon, regular polygon sides=3,inner sep=3pt] (zi2) at (5,2) {};

      \node [draw=none,rotate=45] at (2.5,2.5) {$\cdots$};

      \node [draw=none] at (5,1) {$Z_{i_1}$};
      
      \node (bi1) at (2,2) {$B_{i_1}$};
      \node[fill=red,regular polygon, regular polygon sides=3,inner sep=3pt] (zi11) at (3,1) {};
      \node[fill=red,regular polygon, regular polygon sides=3,inner sep=3pt] (zi12) at (4,1) {};
      
      \node [draw=none,rotate=45] at (1.5,1.5) {$\cdots$};    
      
      \node [draw=none] at (3.3,0) {$Z_{q-1}$};
      \node [draw=none,blue] at (5.2,-0.3) {$Z_{>i_1}$};

      \node (bt1) at (1,1) {$B_{q-1}$};
      \node[fill=red,regular polygon, regular polygon sides=3,inner sep=3pt] (zt11) at (2,0) {};
      
      \node (bt) at (0,0) {$B_q$};
      
      \draw (b1) -- (z11.north);
      \draw (b1) -- (z12.north);
      
      \draw (b2) -- (bi);
      \draw (b2) -- (z21.north);
      \draw (b2) -- (z22.north);
      
      \draw (bi) -- (zi1.north);
      \draw (bi) -- (zi2.north);
      
      \draw (bi1) -- (zi11.north);
      \draw (bi1) -- (zi12.north);
      
      \draw (bt1) -- (bt);
      \draw (bt1) -- (zt11.north);
      
    \end{tikzpicture}
    \label{fig:vpmc-tree-zij}
  }
  \caption{Illustration of the different neighborhoods of the maximal
    clique that we use to find the maximal cliques of Types~3 and~4.
    The figure shows a universal clique decomposition of a completed
    graph, where $\Omega = B_0 \cup \ldots \cup B_q$ is a maximal
    clique.  Observe that the leaf block is $L_q = (B_q, B_q)$ and
    that its tail is $\Omega$.}
  \label{fig:vpmc-tree-omega}
\end{figure}

%%%%%%%%%%%%%%%% END FIGURE OF POTENTIAL MAXIMAL CLIQUES

As has been already mentioned, the following concept is crucial for
our algorithm.  Recall that when~$\Omega$ is a set of vertices in a
graph~$G$, by~$m_\Omega$ we mean the number of edges in~$G[\Omega]$.
\begin{definition}[\Vpmc]\label{def:vital_pmc}
  Let $(G,k)$ be an input instance to \pname{Trivially Perfect
    Completion}.  A vertex set $\Omega \subseteq V(G)$ is a
  \emph{trivially perfect potential maximal clique} or simply
  \emph{potential maximal clique}, if~$\Omega$ is a maximal clique in
  some minimal trivially perfect completion of~$G$.  If moreover this
  trivially perfect completion contains at most~$k$ edges, then the
  potential maximal clique is called \emph{vital}.
\end{definition}

Observe that given a \yes{} instance $(G, k)$ and a minimal
completion~$S$ of size at most~$k$, every maximal clique in $G + S$ is
a vital potential maximal clique in~$G$.  Note also that in
particular, any \vpmc{} contains at most~$k$ non-edges.  The following
definition will be useful:

\begin{definition}[Fill number]
  Let $G = (V,E)$ be a graph,~$S$ a completion and $H = G+S$.  We
  define the \emph{fill} of a vertex~$v$, denoted by $\fn^G_H(v)$ as
  the number of edges incident to~$v$ in~$S$.
\end{definition}

\begin{observation}
  There are at most $2\sqrt{k}$ vertices~$v$ such that $\fn_H^G(v) >
  \sqrt{k}$.
\end{observation}
It follows that for every set $U \subseteq V$ such that $|U| > 2\sqrt{k}$,
there is a vertex $u \in U$ with $\fn_H^G(u) \leq \sqrt{k}$.  Any
vertex~$u$ such that $\fn_H^G(u) \leq \sqrt{k}$ will be referred to as
a \emph{cheap} vertex.

Everything is settled to start the proof of
Theorem~\ref{thm:co-trivially-perfect-subept}.  Our algorithm consists
of three steps.  We first compress the instance to an instance of size
$\bigO(k^3)$, then we enumerate all (subexponentially many) \vpmc s in
this new instance, and finally we do a dynamic programming procedure
on these objects.

% 
% STEP A --- KERNEL
% 
\paragraph{Step~A.\ Kernelization}
For a given input $(G,k)$, we start by applying the kernelization
algorithm by Guo~\cite{guo2007problem} to construct in time
$\bigO(kn^{4})$ an equivalent instance $(G',k')$, where~$G'$ has
$\bigO(k^3)$ vertices and~$k'\leq k$.  Thus, from now on we can assume
that the input graph~$G$ has~$\bigO(k^3)$ vertices.  Without loss of
generality, we will also assume that~$G$ is connected, since we can
treat each connected component of~$G$ separately.

% 
% STEP B --- ENUMERATION
% 
\paragraph{Step~B.\ Enumeration}
In this step, we give an algorithm that in time $2^{\bigO(\sqrt{k}
  \log k)}$ outputs a family~$\mathcal{C}$ of vertex subsets of~$G$
such that
\begin{itemize}
\item the size of~$\mathcal{C}$ is $2^{\bigO(\sqrt{k} \log k)}$, and
\item every vital potential maximal clique belongs to~$\mathcal{C}$.
\end{itemize}

We identify four different types of \vpmc s.  For each type~$i$,
$1\leq i \leq 4$, we list a family~$\mathcal{C}_i$ of
$2^{\bigO(\sqrt{k} \log k)} $ subsets containing all vital potential
maximal cliques of this type.  Finally, $\mathcal{C} = \mathcal{C}_1
\cup \dots \cup \mathcal{C}_4$.  We show that every vital potential
maximal clique of~$(G,k)$ is of at least one of these types and that
all objects of each type can be enumerated in $2^{\bigO(\sqrt{k} \log
  k)} $ time.

Let~$\Omega$ be a vital potential maximal clique.  By the definition
of~$\Omega$, there exists a minimal completion with at most~$k$ edges
into a trivially perfect graph~$H$ such that~$\Omega$ is a maximal
clique in~$H$.  Let $(T=(V_T,E_T), \mathcal{ B}=\{B_{t}\}_{t\in V_T})$
be the universal clique decomposition of~$H$.  Recall that by
Lemma~\ref{lemma:leaf_part},~$\Omega$ corresponds to a path
$P_{rt}=B_{t_0}B_{t_1}\cdots B_{t_q}$ in~$T$ from the root $r=t_0$ to
a leaf $t=t_q$.  Then for the corresponding leaf block $(B_t,D_t)$
with tail~$Q_t$, we have that~$\Omega=Q_t$.  To simplify the notation,
we use~$B_i$ for~$B_{t_i}$.

Note that the algorithm does not know neither the clique~$\Omega$ nor the
completed trivially perfect graph~$H$.  However, in the analysis we
may partition all the vital potential maximal cliques~$\Omega$ with
respect to structural properties of~$\Omega$ and~$H$, and then provide
simple enumeration rules that ensure that all vital potential maximal
cliques of each type are indeed enumerated.  We now proceed to the
description of the types and enumeration rules and refer to
Figure~\ref{fig:vpmc-tree-omega} for a visualization of the concepts.
In the sequel, whenever we are referring to cheap or expensive
vertices, we mean being cheap/expensive with respect to the fixed
completion to $H$.

\medskip\noindent\textit{Type~1.}  Potential maximal cliques of the
first type are such that $|V \setminus \Omega| \leq 2\sqrt{k}$.  The
family $\mathcal{C}_1$ consists of all sets $W \subseteq V$ such that
$|V \setminus W| \leq 2\sqrt{k}$.  There are at most
$(2\sqrt{k}+1)\cdot\binom{\bigO(k^3)}{2\sqrt{k}}$ such sets and we
enumerate all of them in time $2^{\bigO(\sqrt{k} \log k)}$ by trying
all vertex subsets of size at least $|V| - 2\sqrt{k}$.  Thus every
Type~1 vital potential maximal clique is in~$\mathcal{C}_1$.

\medskip\noindent\textit{Type~2.}  By
Lemma~\ref{lemma:leaf_part}~\textit{(\ref{lemma:leaf_part_i})}, we
have that $\Omega = Q_t= N_H[v]$ for each vertex $v \in D_t=B_t$.
\Vpmc s of the second type are such that $|B_t| > 2\sqrt{k}$.  Observe
that then at least one vertex $v \in B_t$ should be \emph{cheap},
i.e., $\fn^G_H(v) \leq \sqrt{k}$.  We generate the family
$\mathcal{C}_2$ as follows.  Every set in~$\mathcal{C}_2$ is of the
form $W_1\cup W_2$, where $W_1 = N_G[v]$ for some $v \in V$, and
$|W_2| \leq \sqrt{k}$.  There are at most
$\binom{\bigO(k^3)}{\sqrt{k}}k^3$ such sets and they can be enumerated
by computing for every vertex $v$ the set $W_1 = N_G[v]$ and adding to
each such set all possible subsets of size at most~$\sqrt{k}$.  Hence
every Type~2 vital potential maximal clique is in~$\mathcal{C}_2$.

\medskip

Thus if~$\Omega$ is not of Types~1 or~2, then $|V \setminus \Omega| > 2 \sqrt{k}$ and
for the corresponding leaf block we have $|B_t|\leq 2\sqrt{k}$.  Since
$|V \setminus \Omega| > 2\sqrt{k}$ it follows that $V \setminus
\Omega$ contains at least one cheap vertex, i.e., a vertex with fill
number at most~$\sqrt{k}$.

We partition the nodes of~$T$ that are not on the path $B_0,
B_1,\dots, B_q$ into~$q$ disjoint sets $Z_0, Z_1, \dots, Z_{q-1}$
according to the nodes of the path~$P_{rt}$.  Node $x \notin
V(P_{rt})$ belongs to~$Z_i$, $i \in \{0,\dots, q - 1\}$, if~$i$ is the
largest integer such that~$t_i$ is an ancestor of~$x$ in~$T$.  In
other words,~$Z_i$ consists of bags of subtrees outside~$P_{rt}$
attached below~$t_i$, see Figure~\ref{fig:vpmc-tree-zi}.  For integers
$p_1,p_2$, we shall denote $B_{p_1,p_2}=\bigcup_{j=p_1}^{p_2} B_j$

For the remaining two types of \vpmc s we distinguish cases depending
on whether all cheap vertices in $V\setminus \Omega$ are located in
exactly one set $Z_i$, or not.  Recall that all \vpmc s for which
$V\setminus \Omega$ does not contain any cheap vertex are already
contained in Type~1.

\medskip\noindent\textit{Type~3.} \Vpmc s $\Omega$ of the third type are
the ones that do not belong to Type 1 or 2, but there exists an index
$i\in\{0,1,\ldots,q-1\}$ such that all cheap vertices of $V\setminus
\Omega$ belong to $Z_i$.  Since $\Omega$ is not of Type 1, $Z_i$ is
non-empty.  Also, since $\Omega$ is not of Type~2, we have that $|B_q|
\leq 2\sqrt{k}$.  Let us denote $Z_{<i}=\bigcup_{j=0}^{i-1} Z_j$ and
$Z_{>i}=\bigcup_{j=i+1}^{q-1} Z_j$ (see
Figure~\ref{fig:vpmc-tree-zi}).  By our assumption, we have that
$Z_{<i}$ and $Z_{>i}$ contain only expensive vertices, and hence
$|Z_{<i}|,|Z_{>i}|\leq 2\sqrt{k}$.  Let $u$ be any cheap vertex
belonging to $Z_i$, and observe that the following equalities and
inclusions are implied by
Lemma~\ref{lem:connectedness-of-minimal-completion}~\textit{(\ref{lem:conn-min-com:2})}:
\begin{itemize}
\item $B_{0,i-1}=N_G(Z_{<i})$;
\item $B_{i+1,q-1}\subseteq N_G(Z_{>i})\subseteq \Omega$;
\item $B_i\subseteq N_G(B_q\cup (N_G[Z_{>i}]\setminus N_H(u))) \subseteq \Omega$.
\end{itemize}
It follows that
\begin{equation}
  \label{eq:type3}
  \Omega=N_G(Z_{<i})\cup N_G(Z_{>i})\cup N_G(B_q\cup (N_G[Z_{>i}]\setminus N_H(u)))\cup B_q .
\end{equation}
Given~\eqref{eq:type3}, we may define family~$\mathcal{C}_3$.  Family
$\mathcal{C}_3$ comprises all the sets that can be constructed as
follows:
\begin{itemize}
\item Pick three disjoint sets $W_1,W_2,W_3\subseteq V$ of size at most
  $2\sqrt{k}$ each.  This corresponds to the choice of $Z_{<i}$,
  $Z_{>i}$ and $B_q$, respectively.
\item Pick a vertex $v\in V$ and a set $A\subseteq V$ of size at most
  $\sqrt{k}$.  This corresponds to the choice of $u$ and fill-in edges
  adjacent to $u$.  Let $N_v=N_G(v)\cup A$.
\item Put the set $N_G(W_1)\cup N_G(W_2)\cup N_G(W_3\cup (N_G[W_2]\setminus N_v))\cup W_3$
  into the family $\mathcal{C}_3$.
\end{itemize}
Observe that since $|V|=\bigO(k^3)$, the number of sets included in
$\mathcal{C}_3$ is at most $2^{\bigO(\sqrt{k}\log k)}$, and that this
family can be enumerated within the same asymptotic running time.
From~\eqref{eq:type3} it follows immediately that each \vpmc{} of
Type~3 is contained in $\mathcal{C}_3$.

\medskip\noindent\textit{Type~4.} \Vpmc s $\Omega$ of the fourth type are
the ones that do not belong to Type~1 or~2, but there exist at least
two indices $i_1$ and $i_2$ such that $Z_{i_1}$ and $Z_{i_2}$ both
contain a cheap vertex.  Let $i_1,i_2$ be the two largest such
indices, where $i_1>i_2$.  Let $Z_{<i_1,>i_2} =
\bigcup_{j=i_2+1}^{i_1-1} Z_j$ and $Z_{>i_1}=\bigcup_{j=i_1+1}^{q-1}
Z_j$.  See Figure~\ref{fig:vpmc-tree-zij} for an illustration.  By the
maximality of $i_1,i_2$ we have that $Z_{<i_1,>i_2}$ and $Z_{>i_1}$
contain only expensive vertices, and hence
$|Z_{<i_1,>i_2}|,|Z_{>i_1}|\leq 2\sqrt{k}$.  Again, since $\Omega$ is
not of Type~2, we have that $|B_q|\leq 2\sqrt{k}$.  Let $u_1\in
Z_{i_1}$ and $u_2\in Z_{i_2}$ be two cheap vertices.  Observe that the
following equalities and inclusions are implied by
Lemma~\ref{lem:connectedness-of-minimal-completion}~\textit{(\ref{lem:conn-min-com:2})}:
\begin{itemize}
\item $B_{0,i_2}=N_H(u_1)\cap N_H(u_2)$;
\item $B_{i_2+1,i_1-1}\subseteq N_G(Z_{<i_1,>i_2})\subseteq \Omega$;
\item $B_{i_1+1,q-1}\subseteq N_G(Z_{>i_1})\subseteq \Omega$;
\item $B_{i_1}\subseteq N_G(B_q\cup (N_G[Z_{>i_1}]\setminus N_H(u_1))) \subseteq \Omega$.
\end{itemize}
It follows that
\begin{equation}\label{eq:type4}
  \Omega=(N_H(u_1)\cap N_H(u_2))\cup N_G(Z_{<i_1,>i_2})\cup N_G(Z_{>i_1})\cup N_G(B_q\cup
  (N_G[Z_{>i_1}]\setminus N_H(u_1)))\cup B_q.
\end{equation}
Given~\eqref{eq:type4}, we may define the family~$\mathcal{C}_4$.  This
family comprises all the sets that can be constructed as follows:
\begin{itemize}
\item Pick three disjoint sets $W_1,W_2,W_3\subseteq V$ of size at most
  $2\sqrt{k}$ each.  This corresponds to the choice of
  $Z_{<i_1,>i_2}$, $Z_{>i_1}$ and $B_q$, respectively.
\item Pick two vertices $v_1,v_2\in V$ and two sets $A_1,A_2\subseteq V$, each
  of size at most $\sqrt{k}$.  This corresponds to the choice of $u_1$
  and $u_2$, and of the neighbors in $H$ adjacent to $u_1$ and $u_2$.
  Let $N_{v_i}=N_G(v_i)\cup A_i$, for $i=1,2$.
\item Put the set $(N_{v_1}\cap N_{v_2})\cup N_G(W_1)\cup N_G(W_2)\cup N_G(W_3\cup
  (N_G[W_2]\setminus N_{v_1}))\cup W_3$ into the family $\mathcal{C}_4$.
\end{itemize}
Observe that since $|V|=\bigO(k^3)$, the number of sets included in
$\mathcal{C}_4$ is at most $2^{\bigO(\sqrt{k}\log k)}$, and that this
family can be enumerated within the same asymptotic running time.
From~\eqref{eq:type4} it follows immediately that each \vpmc{} of
Type~4 is contained in $\mathcal{C}_4$.

Summarizing, every \vpmc{} of Type~1, 2, 3, and~4 is included in the
family~$\mathcal{C}_1$, $\mathcal{C}_2$, $\mathcal{C}_3$,
and~$\mathcal{C}_4$, respectively.  Since every \vpmc{} is of Type~1,
2, 3, or~4, by taking $\mathcal{C}=\mathcal{C}_1\cup \mathcal{C}_2\cup
\mathcal{C}_3\cup \mathcal{C}_4$ we can infer the following lemma that
formalizes the result of Step~B.

\begin{lemma}[Enumeration Lemma]
  \label{lem:enumeration}
  Let $(G,k)$ be an instance of \pname{Trivially Perfect Completion}
  such that $|V(G)| = \bigO(k^3)$.  Then in time
  $2^{\bigO(\sqrt{k}\log{k})}$, we can construct a family
  $\mathcal{C}$ consisting of $2^{\bigO(\sqrt{k}\log{k})}$ subsets of
  $V(G)$ such that every \vpmc{} of $(G,k)$ is in $\mathcal{C}$.
\end{lemma}

% 
% STEP C --- DYNAMIC PROGRAMMING
% 
\paragraph{Step~C.\ Dynamic programming}
We first give an intuitive idea of the dynamic procedure: We start off
by assuming that we have the family~$\mathcal{C}$ containing all \vpmc
s of $(G,k)$.  We start by generating in time
$2^{\bigO(\sqrt{k}\log{k})}$ a family~$\mathcal{S}$ of pairs $(X,Y)$,
where $X,Y \subseteq V(G)$, such that for every minimal completion~$S$
of size at most~$k$, and the corresponding universal clique
decomposition $(T,\mathcal{B})$ of~$H=G+S$, it holds that every
block~$(B,D)$ is in~$\mathcal{S}$, and the size of~$\mathcal{S}$ is
$2^{\bigO(\sqrt{k}\log{k})}$.  (See Definition~\ref{def:tp-block} for
the definition of a block.)

The construction of~$\mathcal{S}$ is based on the following
observations about blocks and \vpmc s: Let~$G$ be a graph,~$S$ a
minimal completion and $L=(B,D)$ a block of the universal clique
decomposition of $H=G+S$, where~$H$ is not a complete graph, with~$Q$
being its tail.  Then the following holds:
\begin{itemize}
\item If~$L$ is a leaf block, then $B = \Omega_1 \setminus \Omega_2$ for some \vpmc
  s~$\Omega_1$ and~$\Omega_2$, and $D = B$.
\item If~$L$ is the root block, then the tail of~$L$ is~$B$, $B =
  \Omega_1 \cap \Omega_2$ for some \vpmc s~$\Omega_1$ and~$\Omega_2$, and $D = V$.
\item If~$L$ is an internal block, then~$Q$ is the intersection of two
  \vpmc s~$\Omega_1$ and~$\Omega_2$ of~$G$, $B = Q \setminus \Omega_3$
  for some \vpmc{}~$\Omega_3$, and~$D$ is the connected component of
  $G - (Q \setminus B)$ containing~$B$.
\end{itemize}

From this observation, we can conclude that by going through all
triples $\Omega_1,\Omega_2,\Omega_3$, we can compute the
set~$\mathcal{S}$ consisting of all blocks~$(B,D)$ of minimal
completions.  We now define the value~$\dpt(B,D)$ as follows:
$\dpt(B,D)$ is equal to the minimum number of edges needed to be added
to~$G[D]$ to make it a trivially perfect graph with~$B$ being the
universal clique contained in the root of the universal clique
decomposition, unless this minimum number is larger than $k$; In this
case we put $\dpt(B,D)=+\infty$.  We later derive recurrence equations
that enable us to compute all the relevant values
of~$\dpt(\cdot,\cdot)$ using dynamic programming.  Finally, the
minimum cost of completing~$G$ to a trivially perfect graph is equal
to $\min_{(B,V(G)) \in \mathcal{S}}\dpt(B,V(G))$.  If this minimum is
equal to $+\infty$, then no completion of size at most $k$ exists and
we can conclude that $G,k$ is a \no-instance.

\bigskip

We now proceed to a formal proof of the correctness of the dynamic
programming procedure.  Suppose that we have the family~$\mathcal{C}$
containing all \vpmc s of~$(G,k)$.  We start by generating in time
$2^{\bigO(\sqrt{k}\log{k})}$ a family~$\mathcal{S}$ of pairs~$(X,Y)$,
where $X,Y \subseteq V$, where $V = V(G)$, such that
\begin{itemize}
\item for every minimal completion~$H$ that adds at most $k$ edges,
  every block~$(B,D)$ of the universal clique decomposition of $H$
  belongs to~$\mathcal{S}$, and
\item the size of $\mathcal{S}$ is $2^{\bigO(\sqrt{k}\log{k})}$.
\end{itemize}
The construction of $\mathcal{S}$ is based on the following lemmata.

\begin{lemma}
  \label{lemma:constructingBlocks1}
  Let $G$ be a graph, $S$ a minimal completion of size at most $k$,
  and~$(B,D)$ a non-leaf and non-root block of the universal clique
  decomposition of~$H=G+S$, with~$Q$ being its tail.  Then
  \begin{enumerate}[(i)]
  \item\label{lem:conBl:item1:1} $Q$ is the intersection of two \vpmc
    s~$\Omega_1$ and~$\Omega_2$ of~$G$,
  \item\label{lem:conBl:item1:2} $B = Q \setminus \Omega_3$ for some \vpmc{}~$\Omega_3$,
    and
  \item\label{lem:conBl:item1:3} $D$ is the connected component of $G
    - (Q \setminus B)$ containing~$B$.
  \end{enumerate}
\end{lemma}

\begin{proof}
  \textit{(\ref{lem:conBl:item1:1})} Consider two connected
  components~$D_1$ and~$D_2$ of~$H[D \setminus B]$ and let~$\Omega'_1$
  and~$\Omega'_2$ be maximal cliques in~$D_1$ and~$D_2$.  Observe that
  $\Omega_1 = \Omega'_1 \cup Q$ and $\Omega_2 = \Omega'_2 \cup Q$ are
  maximal cliques in~$H$.  By definition,~$\Omega_1$ and~$\Omega_2$
  are \vpmc s in~$G$ and $\Omega_1 \cap \Omega_2 = Q$.
  
  \textit{(\ref{lem:conBl:item1:2})} Let $\hat{L} = (\hat{B},
  \hat{D})$ be the parent block of $(B,D)$.  Since $\hat{L}$ is not a
  leaf-block, $\hat{L}$ has at least two children and thus there is a
  block $(B',D')$ which is also a child of $\hat{L}$.  By the previous
  point,~$\hat{Q}$, the tail of~$\hat{L}$ is exactly $\hat{Q} =
  \Omega_1 \cap \Omega_3$ for some vital potential maximal
  clique~$\Omega_3$.  It follows that $B = Q \setminus \Omega_3$.
  
  \textit{(\ref{lem:conBl:item1:3})} It follows from
  Lemma~\ref{lem:connectedness-of-minimal-completion} that $G[D]$ is
  connected.  Then it follows immediately that $D$ is the unique
  connected component of $G - (Q \setminus B)$ containing $B$.
\end{proof}

\begin{lemma}
  \label{lemma:constructingBlocks2}
  Let~$G$ be a graph,~$S$ a minimal completion of size at most $k$,
  and $L = (B,D)$ a leaf block of the universal clique decomposition
  of $H = G + S$.  If~$H$ is not a complete graph, then
  \begin{enumerate}[(i)]
  \item\label{lem:conBl:item2:1} $B = \Omega_1 \setminus \Omega_2$ for some \vpmc
    s~$\Omega_1$ and~$\Omega_2$, and
  \item\label{lem:conBl:item2:2} $D = B$.
  \end{enumerate}
\end{lemma}

\begin{proof}
  \textit{(\ref{lem:conBl:item2:1})} Let $\hat{L} = (\hat{B},
  \hat{D})$ be the parent block of~$L$, which exists since~$L$ is not
  the root block.  Let $L' = (B', D')$ be a child of~$\hat{L}$ which
  is not~$L$.  If $L' = (B', D')$ is a leaf, then set $L'' = L$, and
  if not, then let $L''= (B'',D'')$ be a leaf having~$L'$ as an
  ancestor.  The blocks~$L'$ and~$L''$ exist since~$\hat{L}$ is not a
  leaf.  Furthermore, let~$\hat{Q}$ be the tail of~$\hat{L}$, and let
  $\Omega_1 = N_H[B]$ and $\Omega_2 = N_H[B'']$ be two maximal cliques
  in~$H$.  We know from above that $\hat{Q} = \Omega_1 \cap \Omega_2$
  and hence $B = \Omega_1 \setminus \Omega_2$.
  
  \textit{(\ref{lem:conBl:item2:2})} This follows immediately from
  Lemma~\ref{lemma:leaf_part}.
\end{proof}

\begin{lemma}
  \label{lemma:constructingBlocks3}
  Let $G$ be a connected graph, $S$ a minimal completion of size at
  most $k$, and $L = (B,D)$ the root block of the universal clique
  decomposition of $H = G + S$.  If $H$ is not a complete graph, then
  \begin{enumerate}[(i)]
  \item\label{lem:conBl:item3:1} the tail of $L$ is $B$,
  \item\label{lem:conBl:item3:2} $B = \Omega_1 \cap \Omega_2$ for some \vpmc s
    $\Omega_1$ and $\Omega_2$, and
  \item\label{lem:conBl:item3:3} $D = V$.
  \end{enumerate}
\end{lemma}

\begin{proof}
  \textit{(\ref{lem:conBl:item3:1})} By definition, the tail is the
  collection of vertices from~$B$ to the root.  Since~$L$ is a root
  block, the tail is~$B$ itself.
  
  \textit{(\ref{lem:conBl:item3:2})} This follows in the same manner
  as in the proof of
  Lemma~\ref{lemma:constructingBlocks1}~\textit{(\ref{lem:conBl:item1:1})},
  since~$B$ is the tail of block~$L$.
  
  \textit{(\ref{lem:conBl:item3:3})} From the definition of universal
  clique decompositions we have that~$D$ is the connected component of
  $H[V \setminus (Q \setminus B)]$ containing~$B$, but $Q \setminus B
  = \emptyset$, hence~$D$ is the connected component of~$H$
  containing~$B$ and since~$H$ is connected, the result follows.
\end{proof}

By making use of
Lemmata~\ref{lemma:constructingBlocks1}--\ref{lemma:constructingBlocks3},
one can construct the required family~$\mathcal{S}$ by going through
all possible triples of elements of~$\mathcal{C}$.  The size
of~$\mathcal{S}$ is at most
$|\mathcal{C}|^3=2^{\bigO(\sqrt{k}\log{k})}$ and the running time of
the construction of~$\mathcal{S}$ is $2^{\bigO(\sqrt{k}\log{k})}$.
Note here that by
Lemma~\ref{lem:connectedness-of-minimal-completion}~\textit{(\ref{lem:conn-min-com:3})}
and the fact that~$G$ is connected, we may discard from~$\mathcal{S}$
every pair~$(B,D)$ where~$G[D]$ is not connected.

For every pair $(X,Y) \in \mathcal{S}$, with $X \subseteq Y \subseteq V$, we define
$\dpt\left(X,Y\right)$ to be the minimum number of edges required to
add to~$G[Y]$ to obtain a trivially perfect graph where~$X$ is the
maximal universal clique; If this minimum value exceeds $k$, we define
$\dpt\left(X,Y\right)=+\infty$.  Thus, to compute an optimal solution,
it is sufficient to go through the values $\dpt\left(X,Y\right)$,
where $(X,Y) \in \mathcal{S}$ with $Y = V$.  In other words, to
compute the size of a minimum completion we can find
\begin{equation}
  \label{eq:DP}
  \min_{(X,V) \in \mathcal{S}} { \dpt\left(X,V\right) } ,
\end{equation}
and if this value is $+\infty$, then the size of a minimum completion
exceeds~$k$.

In the following, for a subset of vertices~$A$ we write~$m_A$ to
denote the number of edges inside~$A$, i.e., $m_A =|E(A)|$.  We
compute~\eqref{eq:DP} by making use of dynamic programming over sets
of~$\mathcal{S}$.  For every pair $(X,Y) \in \mathcal{S}$ which can be
a leaf block for some completion, i.e., for all pairs with $X = Y$, we
put
\[
\dpt\left(X,X\right) = \binom{|X|}{2} - m_X .
\]
Of course, if the computed value exceeds $k$, then we put
$\dpt\left(X,X\right)=+\infty$.

For $(X,Y) \in \mathcal{S}$ with $X \subsetneq Y$, if $(X,Y)$ is a block of some
minimal completion~$H$, then in~$H$, we have that~$X$ is a universal
clique in~$H[Y]$, every vertex of~$X$ is adjacent to all vertices of
$Y \setminus X$ and the number of edges in $H[Y \setminus X]$ is the
sum of edges in the connected components of $H[Y \setminus X]$.  By
Lemma~\ref{lem:connectedness-of-minimal-completion}, the vertices of
every connected component~$Y'$ of~$H[Y \setminus X]$ induce a
connected component in~$G[Y \setminus X]$.  We can notice that for
each connected component~$Y'$ of~$H[Y \setminus X]$ the decomposition
of~$H$ contains a new block $(X',Y')$ and since~$\mathcal{S}$ contains
all blocks of minimal trivially perfect completions it follows that
$(X',Y') \in \mathcal{S}$.

Now for $(X,Y) \in \mathcal{S}$ in increasing size of~$Y$, we use $m_{X,
  Y \setminus X} = |E(X, Y \setminus X)|$ to denote the number of
edges between~$X$ and~$Y \setminus X$ in~$G$.  Let~$C$ be the set of
connected components of~$G[Y \setminus X]$.  Then we have
\[
\dpt\left(X,Y\right) = \binom{|X|}{2} - m_X + |X| \cdot |Y \setminus X| -
m_{X, Y \setminus X} + \sum_{G[Y'] \in C} \min_{(X',Y') \in
  \mathcal{S}} \dpt\left(X',Y'\right) .
\]
Again, if the value on the right hand side exceeds $k$, then we have
$\dpt\left(X,Y\right) = +\infty$.

The cardinality of~$Y'$ is less than~$|Y|$ since~$X \neq \emptyset$ and as blocks
are processed in increasing cardinality of~$Y$, the value for
$\dpt\left(X',Y'\right)$ has been calculated when it is needed for
$\dpt\left(X,Y\right)$.

The running time required to compute $\dpt\left(X,Y\right)$ is up to a
polynomial factor in~$k$ proportional to the number of sets~$(X', Y')
\in \mathcal{S}$, which is~$\bigO(|\mathcal{S}|)$.  Thus the total
running time of the dynamic programming procedure is up to a
polynomial factor in~$k$ proportional to~$\bigO(|\mathcal{S}|^2)$, and
hence~\eqref{eq:DP} can be computed in
time~$2^{\bigO(\sqrt{k}\log{k})}$.  This concludes Step~C and the
proof of Theorem~\ref{thm:co-trivially-perfect-subept}.

\newcommand{\funfamily}{\mathfrak{F}}

\section{Completion to threshold graphs}
\label{sec:threshold-completion}
In this section we give an algorithm which solves \pname{Threshold
  Completion}, which is \fd{} for the case when $\F = \{2K_2, C_4,
P_4\}$, in subexponential parameterized time.  More specifically, we
show the following theorem:

\begin{theorem}
  \label{thm:threshold-completion-subept}
  \pname{Threshold Completion} is solvable in time $2^{\bigO(\sqrt{k}
    \log k)} + \bigO(kn^{4})$.
\end{theorem}

The proof of Theorem~\ref{thm:threshold-completion-subept} is a
combination of the following known techniques: the kernelization
algorithm by Guo~\cite{guo2007problem}, the chromatic coding technique
of Alon et al.~\cite{alon2009fast}, also used in the subexponential
algorithm of Ghosh et al.~\cite{ghosh2012faster} for split graphs, and
the algorithm of Fomin and Villanger for chain
completion~\cite{fomin2012subexponential}.

For the kernelization part we use the following result from
Guo~\cite{guo2007problem}.  Guo stated and proved it for the
complement problem \pname{Threshold Edge Deletion}, but since the set
of forbidden subgraphs~$\F = \{2K_2, C_4, P_4\}$ is
self-complementary, the deletion and completion problems are
equivalent.

\begin{proposition}[\cite{guo2007problem}]
  \label{lem:Guo_tresh}
  \pname{Threshold Completion} admits a kernel with~$\bigO(k^3)$
  vertices.  The running time of the kernelization algorithm
  is~$\bigO(kn^4)$.
\end{proposition}

\paragraph{Universal sets}
We start with describing the \emph{chromatic coding} technique by Alon
et al.~\cite{alon2009fast}.  Let~$f$ be a coloring (not necessarily
proper) of the vertex set of a graph~$G = (V,E)$ into~$t$ colors.  We
call an edge~$e \in E$ \emph{monochromatic} if its endpoints have the
same color, and we call a set of edges~$F \subseteq E$ \emph{colorful}
if no edge in~$F$ is monochromatic.
\begin{definition}
  \label{def:universal-coloring-family}
  A \emph{universal~$(n,k,t)$-coloring family} is a
  family~$\funfamily$ of functions from~$[n]$ to~$[t]$ such that for
  any graph~$G$ with vertex set~$[n]$, and~$k$ edges, there is an~$f
  \in \funfamily$ such that~$f$ is a proper coloring of~$G$,
  i.e.,~$E(G)$ is colorful.
\end{definition}

\begin{proposition}[\cite{alon2009fast}]
  \label{prop:unsetfamily}
  For any~$n > 10k^2$, there exists an explicit universal
  $(n,k,\bigO(\sqrt{k}))$-coloring family~$\funfamily$ of size
  $|\funfamily| \leq 2^{ {\bigO}(\sqrt{k}\log{k})} \log n$.
\end{proposition}

Note that by \emph{explicit} we mean here that the family~$\funfamily$
not only exists, but can be constructed in $2^{
  {\bigO}(\sqrt{k}\log{k})} n^{\bigO(1)}$ time.

\subsection{Split, threshold and chain graphs.}
Here we give some known facts about split graphs, threshold graphs and
chain graphs which we will use to obtain the main result.

\begin{definition}
  \label{def:split-partition}
  Given a graph $G = (V, E)$, a partition of the vertex set into
  sets~$C$ and~$I$ is called a \emph{split partition} of~$G$ if~$C$ is
  a clique and~$I$ is an independent set.
\end{definition}
We denote by $(C,I)$ a split partition of a graph.
\begin{definition}[Split graph]
  \label{def:split-graph}
  A graph is a split graph if it admits a split partitioning.
\end{definition}

\begin{proposition}[Theorem 6.2, \cite{golumbic-book}]\label{prop:Ghosh} 
A split graph on~$n$ vertices
  has at most~$n + 1$ split partitions and these partitions can be
  enumerated in polynomial time.
\end{proposition}

\begin{definition}
  \label{def:chain-graph}
  A \emph{chain graph} is a bipartite graph $G = (A,B,E)$ where the
  neighborhoods of the vertices are nested, i.e., there is an ordering
  of the vertices in~$A$, $a_1,a_2,\dots,a_{n_1}$, such that for
  each~$i < n_1$ we have that~$N(a_i) \subseteq N(a_{i+1})$,
  where~$n_1 = |A|$.
\end{definition}

We will use the following result, which is often used as an
alternative definition of threshold graphs.

\begin{proposition}[\cite{mahadev1995threshold}]
  \label{lem:threshold-nested-neighborhoods}
  A graph~$G$ is a threshold graph if and only if~$G$ has a split
  partition~$(C,I)$ and the neighborhoods of the vertices of~$I$ are
  nested.
\end{proposition}
Thus, the class of threshold graphs is a subclass of split graphs and
by Proposition~\ref{prop:Ghosh}, threshold graphs on~$n$ vertices have
at most~$n + 1$ split partitions.

\subsection{The algorithm.}
We now proceed to the details of the algorithm which solves
\pname{Threshold Completion} in the time stated in the theorem.  Fomin
and Villanger~\cite{fomin2012subexponential} showed that the following
problem is solvable in subexponential time:

\defparproblem{\pname{Chain Completion}}{A bipartite graph $G=(A,B,E)$
  and integer~$k$.}{$k$}{Is there a set of edges~$S$ of size at
  most~$k$ such that~$(A,B, E \cup S)$ is a chain graph?}

Note that in the \pname{Chain Completion} problem, the resulting chain
graph must have the same bipartition as the input graph.  Thus,
despite the fact that chain graphs are exactly the $\{2K_2, C_3, C_4,
P_4\}$-free graphs, formally \pname{Chain Completion} is not an \fd{}
problem according to our definition.

\begin{proposition}[\cite{fomin2012subexponential}]\label{prop:FominVill}
  \textsc{Chain Completion} is solvable in $2^{\bigO(\sqrt{k}\log{k})}
  + \bigO(k^2 nm)$ time.
\end{proposition}

We now have the results needed to give an algorithm for
\pname{Threshold Completion}, thus proving
Theorem~\ref{thm:threshold-completion-subept}.
\begin{proof}[of Theorem~\ref{thm:threshold-completion-subept}]
  We start by using Proposition~\ref{lem:Guo_tresh} to obtain a
  polynomial kernel with~$\bigO(k^3)$ vertices in time~$\bigO(kn^4)$.
  We will therefore from now on assume that the input graph~$G$
  has~$n=\bigO(k^3)$ vertices.
  
  Suppose that $(G,k)$ is a \yes{} instance of \pname{Threshold
    Completion}.  Then there is an edge set~$S$ of size at most~$k$
  such that~$G + S$ is a threshold graph.  Without loss of generality,
  we can assume that $n > 10k^2$.  By
  Proposition~\ref{prop:unsetfamily}, we can construct in
  $2^{{\bigO}(\sqrt{k}\log{k})} n^{\bigO(1)} =
  2^{{\bigO}(\sqrt{k}\log{k})}$ time an explicit universal
  $(n,k,\bigO(\sqrt{k}))$-coloring family~$\funfamily$ of size
  $|\funfamily| \leq 2^{{\bigO}(\sqrt{k}\log{k})} \log n =
  2^{{\bigO}(\sqrt{k}\log{k})}$.  Since~$|S| \leq k$, there is a
  vertex coloring~$f \in \funfamily$ such that~$S$ is colorful.
  
  We iterate through all the colorings~$f \in \funfamily$.  Let us
  examine one coloring~$f \in \funfamily$, and let $V_1, V_2, \dots ,
  V_t$ be the partitioning of~$V(G)$ according~$f$, where $t =
  \bigO(\sqrt{k})$.  Then, since threshold graphs are hereditary and
  we assume~$S$ to be colorful, each~$V_i$ must induce a threshold
  graph---we cannot add edges within a color class.

  By Proposition~\ref{prop:Ghosh},~$G+S$ has~$\bigO(k^3)$ split
  partitions.  Each such split partition of~$G+S$ induces a split
  partition of~$G[V_i]$, $i\in\{1,\dots, t\}$.  Again by
  Proposition~\ref{prop:Ghosh}, each~$G[V_i]$ also has~$\bigO(k^3)$
  split partitions.  We use brute-force to generate the set of
  $\bigO((k^3)^t) = 2^{{\bigO}(\sqrt{k}\log{k})}$ partitions of~$G$,
  and the set of generated partitions contains all split partitions
  of~$G+S$.  By Proposition~\ref{lem:threshold-nested-neighborhoods},
  if~$(G,k)$ is a \yes{} instance and~$f$ is colorful, then for at
  least one of the split partitions~$(C,I)$ of~$G+S$ the neighborhoods
  of~$I$ are nested.  To check if a split partition can be turned into
  a nested partition, we use Proposition~\ref{prop:FominVill}.

  To summarize, we perform the following steps:
  
  \paragraph{Step~A.\ Kernelization}
  Apply Proposition~\ref{lem:Guo_tresh} to obtain in
  time~$\bigO(kn^4)$ a kernel with~$\bigO(k^3)$ vertices.  From now on
  we assume that the number of vertices~$n$ in~$G$ is~$\bigO(k^3)$.
  
  \paragraph{Step~B.\ Generating universal families}
  If necessary, we add a set of isolated vertices to~$G$ to guarantee
  that~$n > 10k^2$.  We apply Proposition~\ref{prop:unsetfamily} to
  construct a universal $(n,k,\bigO(\sqrt{k}))$-coloring
  family~$\funfamily$ of size $2^{{\bigO}(\sqrt{k}\log{k})}$.  For
  each generated coloring~$f$ and the corresponding vertex partition
  $V_1, V_2, \dots, V_t$, $t = \bigO(\sqrt{k})$, we perform the steps
  that follow.
  
  \paragraph{Step~C.\ Generating split partitions}
  We generate a set of partitions~$\mathcal{C}$ of~$V(G)$ as follows.
  Each partition $(C,I)\in \mathcal{C}$ is of the following form.  For
  $i\in\{1, \dots, t\}$, let~$\mathcal{C}_i$,
  $|\mathcal{C}_i|=\bigO(k^3)$, be the set of split partitions
  of~$G[V_i]$.  Then for each $i\in\{1, \dots, t\}$, $(C \cap V_i, I
  \cap V_i) \in \mathcal{C}_i$.  In other words, every partition
  of~$\mathcal{C}$ induces a split partition of~$G[V_i]$.  The time
  required to generate all partitions from $\mathcal{C}$ is
  $\bigO((k^3)^t)=2^{ {\bigO}(\sqrt{k}\log{k})}.$ We also perform a
  sanity check by excluding from~$\mathcal{C}$ all pairs~$(C,I)$,
  where~$I$ is not an independent set.  We perform the next step with
  each pair $(C,I)\in \mathcal{C}$.
  
  \paragraph{Step~D.\ Computing nested split partitions}
  For a pair $(C,I)\in \mathcal{C}$, such that~$I$ is an independent set
  in~$G$, we first compute the number of edges~$c$ needed to turn~$C$
  into a clique, i.e., $c = {|C| \choose 2} - m_C$.  Finally, we use
  Proposition~\ref{prop:FominVill} to check if the neighborhood of~$I$
  in~$C$ can be made nested by adding at most~$k-c$ edges.
  
  \medskip From the discussions above, if~$(G,k)$ is a~\yes{} instance
  of the problem, the solution will be found after completing the
  algorithm.  Otherwise, we conclude that~$(G,k)$ is a~\no{} instance.
  The running time to perform Step~A is $\bigO(kn^4)$ and Step~B is
  done in $2^{{\bigO}(\sqrt{k}\log{k})}$.  For every $f \in
  \funfamily$, in Step~C we generate $2^{{\bigO}(\sqrt{k}\log{k})}$
  partitions.  The total number of times Step~C is called is
  $|\funfamily|$ and the total number of partitions generated is
  $|\funfamily| \cdot 2^{{\bigO}(\sqrt{k}\log{k})} =
  2^{{\bigO}(\sqrt{k}\log{k})}$.  In Step~D, we run the algorithm with
  running time $2^{{\bigO}(\sqrt{k}\log{k})}$ on each of the
  $2^{{\bigO}(\sqrt{k}\log{k})}$ partitions, resulting in a total
  running time of $2^{\bigO(\sqrt{k} \log k)} + \bigO(kn^{4})$.
\end{proof}

\section{Completion to pseudosplit graphs}
\label{sec:pseudosplit-completion}
In this section we show that \pname{Pseudosplit Completion}, or \fd{}
for $\mathcal{F} = \{2K_2, C_4\}$, can be solved by first applying a
polynomial-time and parameter-preserving preprocessing routine, and
then using the subexponential time algorithm of Ghosh et
al.~\cite{ghosh2012faster} for \pname{Split Completion}.

The crucial property of pseudosplit graphs that will be of use is the
following characterization:
\begin{proposition}[\cite{maffray1994linearpseudo}]
  \label{lem:pseudo-charact}
  A graph~$G = (V,E)$ is pseudosplit if and only if one of the
  following holds
  \begin{itemize}
  \item $G$ is a split graph, or
  \item $V$ can be partitioned into $C,I,X$ such that $G[C \cup I]$ is a
    split graph with~$C$ being a clique and~$I$ being an independent
    set, $G[X] \cong C_5$, and moreover, there is no edge between~$X$
    and~$I$ and every edge is present between~$X$ and~$C$.
  \end{itemize}
\end{proposition}
In other words, a pseudosplit graph is either a split graph, or a
split graph containing one induced~$C_5$ which is completely
non-adjacent to the independent set of the split graph, and completely
adjacent to the clique set of the split graph.  We call a graph which
falls into the latter category a \emph{proper pseudosplit graph}.

In order to ease the argumentation regarding minimal completions, we
call a split partition $(C,I)$ \emph{$I$-maximal} if there is no
vertex $v \in C$ such that $(C \setminus \{v\}, I \cup \{v\})$ is a
split partition.  Our algorithm uses the subexponential algorithm of
Ghosh et al.~\cite{ghosh2012faster} for \pname{Split Completion} as a
subroutine.  We therefore need the following result:
\begin{proposition}[\cite{ghosh2012faster}]\label{proposition:subexpGhosh}
  \pname{Split Completion} is solvable in time $2^{\bigO(\sqrt{k} \log
    k)} n^{\bigO(1)}$.
\end{proposition}

Formally, in this section we prove the following theorem:

\begin{theorem}
  \label{thm:pseudosplit-completion-subept}
  \pname{Pseudosplit Completion} is solvable in time
  $2^{\bigO(\sqrt{k} \log k)} n^{\bigO(1)}$.
\end{theorem}

\begin{algorithm}[t]
  \begin{enumerate}
  \item Use the algorithm from
    Proposition~\ref{proposition:subexpGhosh} to check in time
    $2^{\bigO(\sqrt{k} \log k)} n^{\bigO(1)}$ if~$(G,k)$ is a \yes{}
    instance of \pname{Split Completion}.  If~$(G,k)$ is a \yes{}
    instance of \pname{Split Completion}, then return that~$(G,k)$ is
    a \yes{} instance of \pname{Pseudosplit Completion}.  Otherwise we
    complete to a proper pseudosplit graph.
    
  \item For each $X = \{x_1,x_2,\dots,x_5\} \subseteq V(G)$ such that there is a
    supergraph~$G_X \supseteq G[X]$ and~$G_X \cong C_5$, we construct
    an instance~$(G',k')$ to \pname{Split Completion} from~$(G,k)$ as
    follows:
    \begin{enumerate}
    \item Let $k'=k+|E(G[X])|-5$.
    \item Add all the possible edges between vertices of~$X$, so
      that~$X$ becomes a clique.
    \item Add a set~$A$ of~$k+2$ vertices to~$G$.
    \item Add every possible edge between~$A$ and~$N_G[X]$.
    \end{enumerate}
  \item Use Proposition~\ref{proposition:subexpGhosh} to check
    if~$(G',k')$ is a~\yes{} instance of \pname{Split Completion}.
    If~$(G',k')$ is a~\yes{} instance of \pname{Split Completion},
    then return that~$(G,k)$ is a~\yes{} instance of
    \pname{Pseudosplit Completion}.
  \item If for no set~$X$ the answer \yes{} was returned, then return
    \no.
  \end{enumerate}
  \caption{Algorithm solving \textsc{Pseudosplit Completion}.}
  \label{alg:pseudosplit-del}
\end{algorithm}

The algorithm whose existence is asserted in
Theorem~\ref{thm:pseudosplit-completion-subept} is given as
Algorithm~\ref{alg:pseudosplit-del}.  We now proceed to prove that
this algorithm is correct, and that its running time on input~$(G, k)$
is $2^{\bigO(\sqrt{k} \log k)} n^{\bigO(1)}$.  In the following we
adopt the notation from Algorithm~\ref{alg:pseudosplit-del}.

As in the algorithm, we denote by~$X$ the set of five vertices which
will be used as the set inducing a~$C_5$ (we try all possible subsets;
note that their number is bounded by~$\bigO(n^5)$).  Note here that
since~$G[X]$ admits a supergraph isomorphic to a~$C_5$, it follows
that $|E(G[X])|\leq 5$ and, consequently, $k'\leq k$.  Similarly,
by~$A$ we denote the set of~$k+2$ vertices we add that are adjacent
only to~$N_G[X]$.  Intuitively, this set will be used to force that in
any minimal split completion of size at most~$k$ it holds that $N_G[X]
\subseteq C$.  From now on~$G'$ is the graph as in the algorithm, that
is,~$G'$ is constructed from~$G$ by making~$X$ into a clique, adding
vertices~$A$ and all the possible edges between~$A$ and~$N_G[X]$.

The following lemma will be crucial in the proof of the correctness of
the algorithm.
\begin{lemma}\label{lem:pseudo-structural}
  Assume that~$S$ is a minimal split completion of~$G'$ of size at
  most~$k'$, and let $(C,I)$ be an $I$-maximal split partition of
  $G'+S$.  Then:
  \begin{enumerate}[(i)]
  \item $N_G[X] \subseteq C$,
  \item $A \subseteq I$,
  \item no edge of $S$ has an endpoint in $A$,
  \item $C \setminus X$ is fully adjacent to $X$ in $G'+S$, and
  \item $I \setminus A$ is fully non-adjacent to $X$ in $G'+S$.
  \end{enumerate}
\end{lemma}
\begin{proof}
  \textit{(i)} Aiming towards a contradiction, suppose that some $v
  \in N_G[X]$ is in~$I$.  Since $A \subseteq N(v)$, we must then have
  that $A \subseteq C$.  However, since~$A$ is stable in~$G$, this
  demands adding at least ${{k+2} \choose 2}>k'$ edges.
  
  \smallskip
  
  \noindent\textit{(ii)} Aiming towards a contradiction, assume that
  $A\cap C\neq \emptyset$.  Since $N_G(A) \subseteq C$ and~$A$ is stable in~$G$, it follows
  that $G'+S'$, where~$S'$ is~$S$ with all the edges incident to~$A$
  removed, is also a split graph with partition $(C',I')$, where
  $C'=C\setminus A$ and $I'=I\cup (A\cap C)$.  Since $S'\subseteq S$,
  we have that either $|S'|<|S|$ which is a contradiction with
  minimality of~$S$, or that $S'=S$ and we obtain a contradiction with
  the assumption that partition $(C,I)$ was $I$-maximal.
  
  \smallskip
  
  \noindent\textit{(iii)} Suppose that there is an edge $e \in S$
  incident to a vertex of~$A$.  Since $A\subseteq I$, we infer that
  $S\setminus \{e\}$ is still a split completion, which contradicts
  the minimality of~$S$.
  
  \smallskip
  
  \noindent\textit{(iv)}~$C$ is a clique in $G'+S$ and $X\subseteq C$, so this
  holds trivially.
  
  \smallskip
  
  \noindent\textit{(v)}
  Suppose for a contradiction that some $v^i \in I \setminus A$ is adjacent to
  some $v^x \in X$.  Since $N_G[X] \subseteq C$, we have that $v^iv^x
  \in S$.  But then $S \setminus \{ v^iv^x\}$ is also a split
  completion, and we have a contradiction with the minimality of~$S$.
\end{proof}

The correctness of the algorithm is implied by the following lemma:

\begin{lemma}
  \label{lem:pseudo-iff-alg-yes}
  The instance $(G,k)$ is a \yes{} instance of \pname{Pseudosplit
    Completion} if and only if Algorithm~\ref{alg:pseudosplit-del}
  returns \yes{} on input $(G,k)$.
\end{lemma}
\begin{proof}
  From left to right, let $(G,k)$ be a \yes{} instance for
  \pname{Pseudosplit Completion}.  We immediately observe that $(G,k)$
  is a~\yes{} instance for \pname{Split Completion} if and only if our
  algorithm returns~\yes{} in the first test.  We therefore assume
  that~$G$ has to be completed to a proper pseudosplit graph.
  
  Let $S_0$ be a completion set with $|S_0|\leq k$ such that $G_0=G+S_0$
  is a proper pseudosplit graph.  Let $(C,I,X)$ be the pseudosplit
  partition of $G+S_0$; hence $G_0[X]$ is isomorphic to a~$C_5$.  We
  claim that the algorithm will return~\yes{} when considering the
  set~$X$ in the second point; let then~$G'$ be the graph constructed
  in the algorithm for the set~$X$.  Let~$S$ be equal to~$S_0$ with
  all the edges of $G_0[X]$ that were not present in $G[X]$ removed;
  note that $|S| = |S_0| + |E(G[X])| - 5 \leq k'$.  We claim that $G'
  + S$ is a split graph with split partition $(C \cup X, I \cup A)$.
  Indeed, since $G'[X]$ is a clique,~$X$ is fully adjacent to~$C$ in
  $G_0 \subseteq G' + S$, and~$C$ is a clique in $G_0 \subseteq G' +
  S$, then $C \cup X$ is a clique in $G' + S$.  On the other hand, $I
  \cup A$ is independent in~$G'$ and all the edges of~$S$ have at
  least one endpoint belonging to $C \cup X$, so $I \cup A$ remains
  independent in $G' + S$.  As a result $G' + S$ is a split graph, and
  so the algorithm will return~\yes{} after the application of
  Proposition~\ref{proposition:subexpGhosh} in the third point.
  
  From right to left, assume that Algorithm~\ref{alg:pseudosplit-del}
  returns~\yes{} on input $(G,k)$.  If it returned~\yes{} already on
  the first test, then~$G$ may be completed into a split graph by
  adding at most~$k$ edges, so in particular $(G,k)$ is a~\yes{}
  instance of \pname{Pseudosplit Completion}.  From now on we assume
  that the algorithm returned~\yes{} in the third point.  More
  precisely, for some set~$X$ the application of
  Proposition~\ref{proposition:subexpGhosh} has found a minimal
  completion set~$S$ of size at most~$k'$ such that $G'+S$ is a split
  graph, with $I$-maximal split partition $(C,I)$.
  
  By Lemma~\ref{lem:pseudo-structural} we have that \textit{(i)}
  $N_G[X] \subseteq C$, \textit{(ii)} $A \subseteq I$, \textit{(iii)}
  no edge of~$S$ has an endpoint in~$A$, \textit{(iv)} $C \setminus X$
  is fully adjacent to~$X$ in $G' + S$, and \textit{(v)} $I \setminus
  A$ is fully non-adjacent to~$X$ in $G' + S$.  By the choice of~$X$,
  there exists a supergraph~$G_X$ of~$G[X]$ such that $G_X \cong C_5$.
  Let now~$S_0$ be equal to~$S$ with all the edges of~$G_X$ that were
  not present in~$G[X]$ included.  Observe that $|S_0| \leq k$ and
  that by \textit{(iii)}~$S_0$ contains only edges incident to
  vertices of~$G$.  Consider now the partition $(C \setminus X, I
  \setminus A, X)$ of $V(G + S_0)$.  Since $(C, I)$ was a split
  partition of $G' + S$, it follows that $C \setminus X$ is a clique
  in $G + S_0$ and $I \setminus A$ is an independent set in $G + S_0$.
  Moreover, from \textit{(iv)} and \textit{(v)} it follows that~$X$ is
  fully adjacent to $C \setminus X$ in $G + S_0$ and fully
  non-adjacent to $I \setminus A$ in $G + S_0$.  Finally, the graph
  induced by~$X$ in $G + S_0$ is $G_X \cong C_5$.  By
  Lemma~\ref{lem:pseudo-charact} we infer that $G + S_0$ is a
  pseudosplit graph, and so the instance $(G,k)$ is a \yes{} instance
  of \pname{Pseudosplit Completion}.
\end{proof}

As for the time complexity of the algorithm, we try sets of five
vertices for~$X$, which is $\bigO(n^5)$ tries.  For each such guess,
we construct the graph~$G'$, which has $n + k+2$ vertices.  Since
$k'\leq k$, by Proposition~\ref{proposition:subexpGhosh} solving
\pname{Split Completion} requires time $2^{\bigO(\sqrt{k} \log k)}
n^{\bigO(1)}$, both in the first and the third point of the algorithm.
Thus the total running time of Algorithm~\ref{alg:pseudosplit-del} is
$2^{\bigO(\sqrt{k} \log k)} n^{\bigO(1)}$.

\section{Lower bounds}
\label{sec:lower-bounds}

In this section we will give the promised lower bounds described in
Figure~\ref{fig:table-completion-complexity}, i.e., we will show that
\fd{} is not solvable in subexponential time for~$\mathcal{F}$ being
one of~$\{2K_2\}$,~$\{C_4\}$,~$\{P_4\}$, and~$\{2K_2,P_4\}$ under ETH.

Throughout this section we will reduce to the above problems from
\pname{3Sat}; We will assume that the input formula~$\phi$ is in
3-CNF, that is, it is a conjunction of a number of clauses, where each
clause is a disjunction of at most three literals.  By applying
standard regularizing preprocessing for~$\phi$ (see for
instance~\cite[Lemma 13]{fomin2011subexponential}) we may also assume
that each clause of~$\phi$ contains \emph{exactly} three literals, and
the variables appearing in these literals are pairwise different.

If~$\phi$ is a \pname{3Sat} instance, we denote by~$\mathcal{V}(\phi)$ the
variables in~$\phi$ and by~$\mathcal{C}(\phi)$ the clauses.  We assume
we have an ordering $c_1, c_2, \ldots, c_m$ for the clauses
in~$\mathcal{C}(\phi)$ and the same for the variables, $x_1, x_2,
\ldots, x_n$.  For simplicity, we also assume the literals in each
clause are ordered internally by the variable ordering.

We restate here the Exponential Time Hypothesis, as that will be the
crucial assumption for proving that the problems mentioned above do
not admit subexponential time algorithms.

%%% ENVIRONMENT FOR ETH (used in intro and lower-bounds)
\medskip
\textsc{Exponential Time Hypothesis \textit{(ETH)}}.
\textit{There exists a positive real number $s$ such that 3-CNF-SAT
  with $n$ variables cannot be solved in time $2^{sn}$.}
\medskip
%%% END ENVIRONMENT FOR ETH

\noindent
By the Sparsification Lemma of Impagliazzo, Paturi and
Zane~\cite{impagliazzo2001which}, unless ETH fails, \pname{3Sat}
cannot be solved in time~$2^{o(n+m)}(n + m)^{\bigO(1)}$.

For each considered problem we present a linear reduction from
\pname{3Sat}, that is, a reduction which constructs an instance whose
parameter is bounded linearly in the size of the input formula.
Pipelining such a reduction with the assumed subexponential
parameterized algorithm for the problem would give a subexponential
algorithm for \pname{3Sat}, contradicting ETH.

% 
% 
% 2K_2 free completion // C_4 free deletion
% 
% 
% 
\subsection{$2K_2$-free completion is not solvable in subexponential
  time}

For $\F=\{2K_2\}$, we refer to \fd{} as to \pname{$2K_2$-Free
  Completion}.  We show the following theorem.
\begin{theorem}
  \label{thm:2k2-completion-exp}
  The problem \pname{$2K_2$-Free Completion} is not solvable in
  $2^{o(k)}n^{\bigO(1)}$ time unless the Exponential Time Hypothesis
  (ETH) fails.
\end{theorem}

For the proof, however, instead of working directly on this problem,
we find it more convenient to show the hardness of the (polynomially)
equivalent problem \pname{$C_4$-Free Edge Deletion}.  We will
throughout this section write $G-S$ when $S \subseteq E(G)$ for the
graph~$(V(G), E \setminus S)$.

\paragraph{Construction}
We reduce from \pname{3Sat} and the gadgets can be seen in
Figure~\ref{fig:c4-del-variable-and-clause-gadget-enlarged}.
Let~$\phi$ be an instance of \pname{3Sat}.  We construct the
instance~$(G_\phi,k_\phi)$ for \pname{$C_4$-Free Edge Deletion} and we
begin by defining the graph~$G_\phi$.  For every variable~$x \in
\mathcal{V}(\phi)$, we construct a variable gadget graph~$G^x$.  The
graph~$G^x$ consists of six vertices~$w^x_{0}$, $w^x_{1}$, $w^x_{2}$,
$n^x$ (for {\emph{negative}}),~$p^x$ (for {\emph{positive}}),
and~$t^x$.  The three vertices~$w^x_{0}$,~$w^x_{1}$ and~$w^x_{2}$ will
induce a triangle whereas~$n^x$ and~$p^x$ are adjacent to the vertices
in the triangle and to~$t^x$.  We can observe that the four
vertices~$n^x$, $t^x$, $p^x$, $w^x_{i}$ induce a~$C_4$ for $i =
0,1,2$, and that no other induced~$C_4$ occurs in the gadget (see
Figure~\ref{fig:c4-del-variable-gadget-enlarged}).  It can also be
observed that by removing either one of the edges~$n^xt^x$ and
$p^xt^x$, the gadget becomes~$C_4$-free.  We will refer to the edge
$t^xp^x$ as the \emph{true edge} and to~$t^x n^x$ as the \emph{false
  edge}.  These edges are the thick edges in
Figure~\ref{fig:c4-del-variable-gadget-enlarged}.  This concludes the
variable gadget construction.

%%%%%%% FIGURE 5, VARIABLES

\begin{figure}[t]
  \centering \subfloat[The variable gadget~$G^x$ for a variable~$x$
  with three occurrences of~$C_4$.  The edge $t^xp^x$ is the {true
    edge} and the edge $t^xn^x$ is the {false edge}.  All~$C_4$s
  of~$G^x$ can be eliminated by removing the true or the false edge.]
  {
    \label{fig:c4-del-variable-gadget-enlarged}
    \begin{tikzpicture}[every node/.style={circle, draw, scale=.8},
      scale=1, inner sep=0pt, minimum size=1.5em]
      
      \node (t) at  (3,5) {$t^x$};
      \node (a) at  (1,4) {$n^x$};
      \node (b) at  (5,4) {$p^x$};
      \node (c1) at (2,1) {$w^x_0$};
      \node (c2) at (3,1.5) {$w^x_1$};
      \node (c3) at (4,1) {$w^x_2$};
      \draw[line width=1mm] (t) -- (a);
      \draw[line width=1mm] (t) -- (b);
      \draw (a) -- (c1);
      \draw (a) -- (c2);
      \draw (a) to[bend left] (c3);
      \draw (b) to[bend right] (c1);
      \draw (b) -- (c2);
      \draw (b) -- (c3);
      \draw (c1) -- (c2);
      \draw (c2) -- (c3);
      \draw (c1) -- (c3);
    \end{tikzpicture}
  }\hspace{2em} \subfloat[The clause gadget~$G^c$ for a clause~$c$ has
  three occurrences of~$C_4$, which can be eliminated by removing two
  of the variable-edges, the thick edges in the figure.]{
    \begin{tikzpicture}[every node/.style={circle, draw, scale=.9},
      scale=.3, xscale=-1, minimum size=1.5em, inner sep=0pt]
      
      \node (a1) at (1,7.5) {$a^c_1$};
      \node (b1) at (4,2.3) {$b^c_1$};
      
      \node (a2) at (10,13.2) {$a^c_2$};
      \node (b2) at (4,13.2) {$b^c_2$};

      \node (a0) at (10.3,2.6) {$a^c_0$};
      \node (b0) at (13,7.5) {$b^c_0$};
      
      \draw[line width=1mm] (a1) -- (b1);
      \draw[line width=1mm] (a2) -- (b2);
      \draw[line width=1mm] (a0) -- (b0);
      
      \draw (a1) -- (a2);
      \draw (b1) -- (b2);
      \draw (a1) -- (a0);
      \draw (b1) -- (b0);
      \draw (a2) -- (a0);
      \draw (b2) -- (b0);
    \end{tikzpicture}
    \label{fig:c4-del-clause-gadget-enlarged}
  }
  \caption{Variable (left) and clause (right) gadgets used in the
    reduction.}
  \label{fig:c4-del-variable-and-clause-gadget-enlarged}
\end{figure}
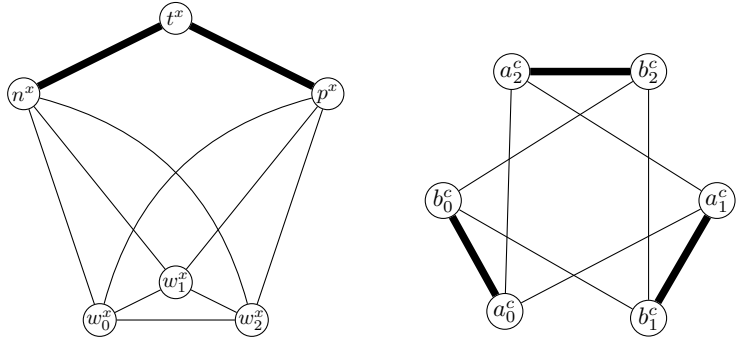

%%%%%%% END FIGURE 5, VARIABLES

%%%%%%% FIGURE 6, CLAUSES

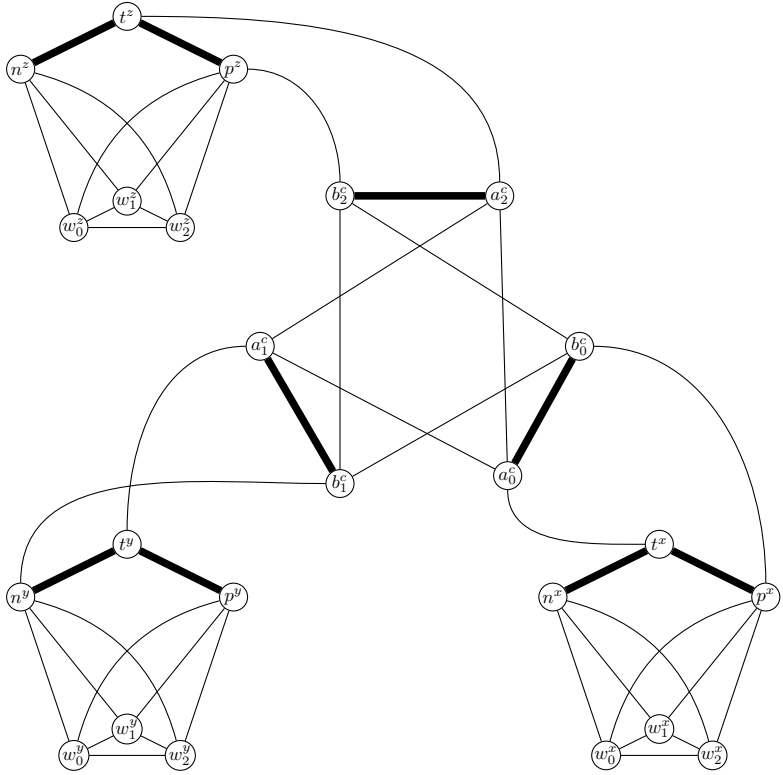
\begin{figure}[t]
  \centering
  \begin{tikzpicture}[every
    node/.style={circle,draw,scale=.7},scale=.7,minimum
    size=1.5em,inner sep=0pt]
    
    % x
    \node (tx) at  (13,5)   {$t^x$};
    \node (ax) at  (11,4)   {$n^x$};
    \node (bx) at  (15,4)   {$p^x$};
    \node (cx1) at (12,1)   {$w^x_0$};
    \node (cx2) at (13,1.5) {$w^x_1$};
    \node (cx3) at (14,1)   {$w^x_2$};
    
    % y
    \node (ty) at  (3,5)   {$t^y$};
    \node (ay) at  (1,4)   {$n^y$};
    \node (by) at  (5,4)   {$p^y$};
    \node (cy1) at (2,1)   {$w^y_0$};
    \node (cy2) at (3,1.5) {$w^y_1$};
    \node (cy3) at (4,1)   {$w^y_2$};
        
    % z
    \node (tz) at  (3,15) {$t^z$};
    \node (az) at  (1,14) {$n^z$};
    \node (bz) at  (5,14) {$p^z$};
    \node (cz1) at (2,11) {$w^z_0$};
    \node (cz2) at (3,11.5) {$w^z_1$};
    \node (cz3) at (4,11) {$w^z_2$};
    
    % c
    \node (a1) at ( 5.50, 8.75) {$a^c_1$};
    \node (b1) at ( 7.00, 6.15) {$b^c_1$};
    \node (a2) at (10.00, 11.6) {$a^c_2$};
    \node (b2) at ( 7.00, 11.6) {$b^c_2$};
    \node (a0) at (10.15, 6.3) {$a^c_0$};
    \node (b0) at (11.50, 8.75) {$b^c_0$};
    
    % draw c
    \draw[line width=1mm] (a1) -- (b1);
    \draw[line width=1mm] (a2) -- (b2);
    \draw[line width=1mm] (a0) -- (b0);
    \draw (a1) -- (a2);
    \draw (b1) -- (b2);
    \draw (a1) -- (a0);
    \draw (b1) -- (b0);
    \draw (a2) -- (a0);
    \draw (b2) -- (b0);

    % x
    \draw[line width=1mm] (tx) -- (ax);
    \draw[line width=1mm] (tx) -- (bx);
    \draw (ax) -- (cx1);
    \draw (ax) -- (cx2);
    \draw (ax) to[bend left] (cx3);
    \draw (bx) to[bend right] (cx1);
    \draw (bx) -- (cx2);
    \draw (bx) -- (cx3);
    \draw (cx1) -- (cx2);
    \draw (cx2) -- (cx3);
    \draw (cx1) -- (cx3);
    
    % y
    \draw[line width=1mm] (ty) -- (ay);
    \draw[line width=1mm] (ty) -- (by);
    \draw (ay) -- (cy1);
    \draw (ay) -- (cy2);
    \draw (ay) to[bend left] (cy3);
    \draw (by) to[bend right] (cy1);
    \draw (by) -- (cy2);
    \draw (by) -- (cy3);
    \draw (cy1) -- (cy2);
    \draw (cy2) -- (cy3);
    \draw (cy1) -- (cy3);

    % z
    \draw[line width=1mm] (tz) -- (az);
    \draw[line width=1mm] (tz) -- (bz);
    \draw (az) -- (cz1);
    \draw (az) -- (cz2);
    \draw (az) to[bend left] (cz3);
    \draw (bz) to[bend right] (cz1);
    \draw (bz) -- (cz2);
    \draw (bz) -- (cz3);
    \draw (cz1) -- (cz2);
    \draw (cz2) -- (cz3);
    \draw (cz1) -- (cz3);
    
    % connection
    \draw (tx) to[out=180, in=-90] (a0);
    \draw (bx) to[out=90, in=0] (b0);
    
    \draw (ty) to[out=90, in=180] (a1);
    \draw (ay) to[out=90, in=-180] (b1);
    
    \draw (tz) to[out=0, in=90] (a2);
    \draw (bz) to[out=0, in=90] (b2);
  \end{tikzpicture}
  \caption{The connections for a clause $c = x \lor \neg y \lor z$.
    For the negated variable,~$\neg y$, we connect the clause gadget
    to~$n^y$ and~$t^y$, whereas for the variables in the non-negated
    form we have the clause connected to the~$t$ and~$p$ vertices.
    Observe that a budget of five is sufficient and necessary for
    eliminating all occurrences of~$C_4$ in the depicted subgraph.}
  \label{fig:2k2-free-compl-connections-enlarged}
\end{figure}

%%%%%%% END FIGURE 6, CLAUSES

For every clause $c \in \mathcal{C}(\phi)$, we construct a clause gadget
graph~$G^c$ as follows.  The graph~$G^c$ consists of two triangles,
$a^c_0,a^c_1,a^c_2$ and $b^c_0,b^c_1,b^c_2$.  We also add the edges
$a^c_0 b^c_0$, $a^c_1 b^c_1$, and $a^c_2 b^c_2$.  These three latter
edges will correspond to the variables contained in~$c$ and we refer
to them as \emph{variable-edges} (the thick edges in
Figure~\ref{fig:c4-del-clause-gadget-enlarged}).  No more edges are
added.  The clause gadget can be seen in
Figure~\ref{fig:c4-del-clause-gadget-enlarged}.  Observe that there
are exactly three induced $C_4$s in $G^c$, all of the form
$a^c_i,a^c_{i+1},b^c_{i+1},b^c_i$ for $i=0,1,2$, where the indices
behave cyclically modulo $3$.  Moreover, removing any two edges of the
form $a^c_i b^c_i$ for $i=0,1,2$ eliminates all the induced $C_4$s
contained in $G^c$.

To conclude the construction, we give the connections between variable
gadgets and clause gadgets that encode literals in the clauses (see
Figure~\ref{fig:2k2-free-compl-connections-enlarged}).  If a variable
$x$ appears in a non-negated form as the $i$th (for $i = 0,1,2$)
variable in a clause $c$, we add the edges $t^xa^c_i$ and $p^xb^c_i$.
If it appears in a negated form, we add the edges $t^xa^c_i$ and
$n^xb^c_i$.  The connections can be seen in
Figure~\ref{fig:2k2-free-compl-connections-enlarged}.  Observe that we
get exactly one extra induced $C_4$ in the connection, and that this
can be eliminated by removing either one of the thick edges.

This concludes the construction.  We have now obtained a graph $G_\phi$
constructed from an instance $\phi$ of \pname{3Sat}.  We let $k_\phi =
|\mathcal{V}(\phi)| + 2|\mathcal{C}(\phi)|$ be the allowed (and
necessary) budget, and the instance of \pname{$C_4$-Free Edge
  Deletion} is then $(G_\phi, k_\phi)$.

We now proceed to prove the following lemma, which will give the
result.
\begin{lemma}\label{lem:phi-sat-iff-c4-del}
  A given \pname{3Sat} instance~$\phi$ has a satisfying assignment if and
  only if~$(G_\phi,k_\phi)$ is a~\yes{} instance of \pname{$C_4$-Free
    Edge Deletion}.
\end{lemma}
\begin{proof}
  Let $\phi$ be satisfiable and~$G_\phi$ and~$k_\phi$ be as above.  We show
  that $(G_\phi,k_\phi)$ is a \yes{} instance for \pname{$C_4$-Free
    Edge Deletion}.  Let $\alpha\colon \mathcal{V}(\phi) \to
  \{\mathtt{true}, \mathtt{false}\}$ be a satisfying assignment
  for~$\phi$.  For every variable $x \in \mathcal{V}(\phi)$, if
  $\alpha(x) = \mathtt{true}$, we remove the edge corresponding to
  true, i.e.  the edge $t^xp^x$, otherwise we remove the edge
  corresponding to false, i.e., the edge $t^xn^x$.  Every clause $c
  \in \mathcal{C}$ is satisfied by~$\alpha$; we pick an arbitrary
  variable~$x$ whose literal satisfies~$c$ and remove two edges
  corresponding to the two other literals.  If a clause is satisfied
  by more than one literal, we pick any of the corresponding
  variables.
  
  For every clause we deleted exactly two edges and for every variable
  exactly one edge.  Thus the total number of edges removed is
  $2|\mathcal{C}(\phi)| + |\mathcal{V}(\phi)| = k_\phi$.  We argue now
  that the remaining graph~$G'_\phi$ is~$C_4$-free.  Since variables
  appearing in clauses are pairwise different, it can be easily
  observed that every induced cycle of length four in~$G_\phi$ is
  either
  \begin{itemize}
  \item entirely contained in some clause gadget, or
  \item entirely contained in some variable gadget, or
  \item is of form $t^x\gamma^xb_i^ca_i^c$, where $x$ is the $i$th variable
    of clause $c$, and $\gamma\in \{n,p\}$ denotes whether the literal
    in $c$ that corresponds to $x$ is negated or non-negated.
  \end{itemize}
  By the construction of~$G'_\phi$, we destroyed all induced 4-cycles of
  the first two types.  Consider a 4-cycle $t^xp^xb_i^ca_i^c$ of the
  third type, where~$x$ appears positively in clause~$c$.  In the case
  when the literal of variable~$x$ was not chosen to satisfy~$c$, we
  have deleted the edge $a_i^cb_i^c$ and so this 4-cycle is removed.
  Otherwise we have that $\alpha(x) = \mathtt{true}$, and we have
  deleted the edge~$t^xp^x$, thus also removing the considered
  4-cycle.  The case of a 4-cycle of the form $t^xn^xb_i^ca_i^c$ is
  symmetric.
  
  Concluding, all the induced 4-cycles that were contained in~$G_\phi$
  are removed in~$G_\phi'$.  Since vertices in pairs~$(a^c_i,b^c_i)$
  and~$(\gamma^x,t^x)$ for~$\gamma\in \{n,p\}$ do not have common
  neighbors, it follows that no new~$C_4$ could be created when
  obtaining~$G_\phi'$ from~$G_\phi$ by removing the edges.  We infer
  that~$G_\phi'$ is indeed~$C_4$-free.
    
  \medskip We proceed with the opposite direction.  Let~$S$ be an edge
  set of~$G_\phi$ of size at most~$k_\phi$ such that~$G-S$
  is~$C_4$-free.  By the definition of the budget~$k_\phi$ and the
  observation that every variable gadget needs at least one edge to be
  in~$S$ and every clause gadget needs at least two edges to be in~$S$
  (note here that the edge sets of clause and variable gadgets are
  pairwise disjoint), we have that~$S$ contains \emph{exactly} one
  edge from each variable gadget, \emph{exactly} two edges from each
  clause gadget, and no other edges.
  
  We construct an assignment $\alpha\colon \mathcal{V}(\phi) \to
  \{\mathtt{true},\mathtt{false}\}$ for the formula~$\phi$ as follows.
  For a variable $x \in \mathcal{V}(\phi)$, put $\alpha(x) =
  \mathtt{false}$ if the false edge~$t^xn^x$ of~$G^x$ is in~$S$, put
  $\alpha(x) = \mathtt{true}$ if the true edge~$t^xp^x$ of~$G^x$ is
  in~$S$, and put an arbitrary value for~$\alpha(x)$ otherwise.  We
  claim that the assignment~$\alpha$ satisfies~$\phi$.

  Suppose for a contradiction that a clause $c \in \mathcal{C}$ is not
  satisfied.  Since exactly two edges in the clause gadget~$G^c$ are
  in~$S$, there is a variable~$x$ appearing in~$c$ such that the
  corresponding variable-edge of~$G^c$ is not in~$S$.  If $\alpha(x) =
  \mathtt{true}$, then because~$c$ is not satisfied, we have that
  $\neg x \in c$.  By the definition of~$\alpha$ we have that the
  false edge of~$G^x$ does not belong to~$S$.  Then in~$G_\phi$, the
  false edge of~$G^x$ and the variable-edge of~$G^c$ corresponding
  to~$x$ form an induced~$C_4$ that is not destroyed by~$S$, a
  contradiction.  The case $\alpha(x) = \mathtt{false}$ is symmetric.
  This concludes the proof of the lemma.
\end{proof}

Finally, the proof of Theorem~\ref{thm:2k2-completion-exp} follows
from Lemma~\ref{lem:phi-sat-iff-c4-del}; Combining the presented
reduction with an algorithm for \pname{$C_4$-Free Edge Deletion}
working in $2^{o(k)}n^{\bigO(1)}$ time would yield an algorithm for
\pname{3Sat} with time complexity $2^{o(n+m)}(n + m)^{\bigO(1)}$,
which contradicts ETH by the results of Impagliazzo, Paturi and
Zane~\cite{impagliazzo2001which}.

% 
% 
% C_4 free completion
% 
% 
% 
\subsection{$C_4$-free completion is not solvable in subexponential
  time}
\label{subsec:C_4}
For every \fd{} problem that so far turned out to be
solvable in subexponential time, we had~the graph $C_4$ in~$\F$ together
with some other graphs: trivially perfect graphs are the class
excluding~$C_4$ and~$P_4$, threshold graphs are the class
excluding~$2K_2$,~$P_4$ and~$C_4$, and pseudosplit graphs are the class
excluding~$2K_2$ and~$C_4$. Previous known
subexponentiality results in the area of graph modifications are
completing to chordal graphs and chain
graphs~\cite{fomin2012subexponential}, completing to split
graphs~\cite{ghosh2012faster} and recently, completing to interval
graphs~\cite{bliznets2014interval} and proper interval
graphs~\cite{bliznets2014proper}.  All these graph classes have~$C_4$
as a forbidden induced subgraph.

It is therefore natural to ask whether the~$C_4$ is the ``reason'' for
the existence of subexponential algorithms.  However, in this section
we show that excluding~$C_4$ alone is not sufficient for achieving a
subexponential time algorithm.  For $\F=\{C_4\}$, we refer to \fd{} as
\pname{$C_4$-Free Completion}.
\begin{theorem}
  \label{thm:c4completion-exp}
  The problem \pname{$C_4$-Free Completion} is not solvable in
  $2^{o(k)}n^{\bigO(1)}$ time unless the Exponential Time Hypothesis
  (ETH) fails.
\end{theorem}

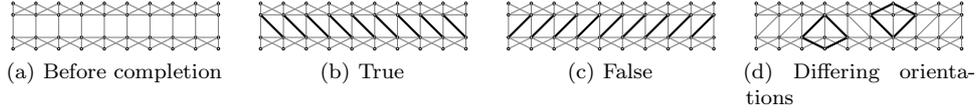
\begin{figure}[t]
  \centering
  \subfloat[Before completion] {
    \centering
    \begin{tikzpicture}[every
      node/.style={circle,draw,scale=.1},scale=.3]
      \foreach \s in {1,...,10}
      {
        \node (\s+u) at (\s,2) {};
        \node (\s+t) at (\s,1.5) {};
        \node (\s+b) at (\s,.5) {};
        \node (\s+d) at (\s,0) {};
        \draw[color=gray] (\s+u) -- (\s+t);
        \draw[color=gray] (\s+t) -- (\s+b);
        \draw[color=gray] (\s+b) -- (\s+d);
      }
      \foreach \s/\t in {1/2, 2/3, 3/4, 4/5, 5/6, 6/7, 7/8, 8/9, 9/10}
      {
        \draw[color=gray] (\s+u) -- (\t+t);
        \draw[color=gray] (\s+t) -- (\t+u);
        \draw[color=gray] (\s+t) -- (\t+t);
        \draw[color=gray] (\s+b) -- (\t+b);
        \draw[color=gray] (\s+b) -- (\t+d);
        \draw[color=gray] (\s+d) -- (\t+b);
      }
    \end{tikzpicture}
    \label{fig:c4compl-variable-gadget}
  }
  \hspace{.01\textwidth}
  \subfloat[True] {
    \centering
    \begin{tikzpicture}[every node/.style={circle,draw,scale=.1},scale=.3]
      \foreach \s in {1,...,10}
      {
        \node (\s+u) at (\s,2) {};
        \node (\s+t) at (\s,1.5) {};
        \node (\s+b) at (\s,.5) {};
        \node (\s+d) at (\s,0) {};
        \draw[color=gray] (\s+u) -- (\s+t);
        \draw[color=gray] (\s+t) -- (\s+b);
        \draw[color=gray] (\s+b) -- (\s+d);
      }
      \foreach \s/\t in {1/2, 2/3, 3/4, 4/5, 5/6, 6/7, 7/8, 8/9, 9/10}
      {
        \draw[color=gray] (\s+u) -- (\t+t);
        \draw[color=gray] (\s+t) -- (\t+u);
        \draw[color=gray] (\s+t) -- (\t+t);
        \draw[color=gray] (\s+b) -- (\t+b);
        \draw[color=gray] (\s+b) -- (\t+d);
        \draw[color=gray] (\s+d) -- (\t+b);
        \draw[thick] (\s+t) -- (\t+b); % TRUE
      }
    \end{tikzpicture}
  }
  \hspace{.01\textwidth}
  \subfloat[False]{
    \centering
    \begin{tikzpicture}[every node/.style={circle,draw,scale=.1},scale=.3]
      \foreach \s in {1,...,10}
      {
        \node (\s+u) at (\s,2) {};
        \node (\s+t) at (\s,1.5) {};
        \node (\s+b) at (\s,.5) {};
        \node (\s+d) at (\s,0) {};
        \draw[color=gray] (\s+u) -- (\s+t);
        \draw[color=gray] (\s+t) -- (\s+b);
        \draw[color=gray] (\s+b) -- (\s+d);
      }
      \foreach \s/\t in {1/2, 2/3, 3/4, 4/5, 5/6, 6/7, 7/8, 8/9, 9/10}
      {
        \draw[color=gray] (\s+u) -- (\t+t);
        \draw[color=gray] (\s+t) -- (\t+u);
        \draw[color=gray] (\s+t) -- (\t+t);
        \draw[color=gray] (\s+b) -- (\t+b);
        \draw[color=gray] (\s+b) -- (\t+d);
        \draw[color=gray] (\s+d) -- (\t+b);
        \draw[thick] (\s+b) -- (\t+t); % FALSE
      }
    \end{tikzpicture}
  }
  \hspace{.01\textwidth}
  \subfloat[Differing orientations]{
    \centering
    \begin{tikzpicture}[every node/.style={circle,draw,scale=.1},scale=.3]
      \foreach \s in {1,...,10}
      {
        \node (\s+u) at (\s,2) {};
        \node (\s+t) at (\s,1.5) {};
        \node (\s+b) at (\s,.5) {};
        \node (\s+d) at (\s,0) {};
        \draw[color=gray] (\s+u) -- (\s+t);
        \draw[color=gray] (\s+t) -- (\s+b);
        \draw[color=gray] (\s+b) -- (\s+d);
      }
      \foreach \s/\t in {1/2, 2/3, 3/4}
      {
        \draw[color=gray] (\s+u) -- (\t+t);
        \draw[color=gray] (\s+t) -- (\t+u);
        \draw[color=gray] (\s+t) -- (\t+t);
        \draw[color=gray] (\s+b) -- (\t+b);
        \draw[color=gray] (\s+b) -- (\t+d);
        \draw[color=gray] (\s+d) -- (\t+b);
        \draw[color=gray] (\s+b) -- (\t+t); % FALSE
      }
      \foreach \s/\t in {4/5, 5/6, 6/7}
      {
        \draw[color=gray] (\s+u) -- (\t+t);
        \draw[color=gray] (\s+t) -- (\t+u);
        \draw[color=gray] (\s+t) -- (\t+t);
        \draw[color=gray] (\s+b) -- (\t+b);
        \draw[color=gray] (\s+b) -- (\t+d);
        \draw[color=gray] (\s+d) -- (\t+b);
        \draw[color=gray] (\s+t) -- (\t+b); % TRUE
      }\foreach \s/\t in {7/8, 8/9, 9/10}
      {
        \draw[color=gray] (\s+u) -- (\t+t);
        \draw[color=gray] (\s+t) -- (\t+u);
        \draw[color=gray] (\s+t) -- (\t+t);
        \draw[color=gray] (\s+b) -- (\t+b);
        \draw[color=gray] (\s+b) -- (\t+d);
        \draw[color=gray] (\s+d) -- (\t+b);
        \draw[color=gray] (\s+b) -- (\t+t); % FALSE
      }
      % FALSE -- TRUE
      \draw[thick] (4+t) -- (5+b) -- (4+d) -- (3+b) -- (4+t);
      % TRUE -- FALSE
      \draw[thick] (7+b) -- (8+t) -- (7+u) -- (6+t) -- (7+b);
    \end{tikzpicture}
    \label{fig:c4compl-variable-gadget-differing}
  }
  \caption{The variable gadget $G_x$, before completion, its two
    completions corresponding to assignments and a completion with
    differing orientations.}
  \label{fig:c4compl-variable-gadgets}
\end{figure}

To prove the theorem, we reduce from \pname{3Sat}, and similarly as
before we start with a formula where each clause contains exactly
three literals corresponding to pairwise different variables.  By
duplicating clauses if necessary, we also assume that each variable
appears in at least two clauses.

We again need two types of gadgets, one gadget to emulate variables in
the formula and one type to emulate clauses.  Let~$\phi$ be the
\pname{3Sat} instance and denote by $\mathcal{V}(\phi)$ the variables
in~$\phi$ and by $\mathcal{C}(\phi)$ the clauses.  We construct the
graph~$G_\phi$ as follows:

For each variable $x \in \mathcal{V}(\phi)$ we construct a variable gadget
graph~$G^x$ as depicted in
Figure~\ref{fig:c4compl-variable-gadget-enlarged}.  Let~$p_x$ be the
number of clauses~$x$ occurs in; by our assumption we have that
$p_x\geq 2$.  The graph~$G^x$ consists of a ``tape'' of~$4 p_x$
squares arranged in a cycle, with additional vertices attached to the
sides of the tape.  The intuition is that every fourth square in~$G^x$
is reserved for a clause~$x$ occurs in.  Formally, the vertex set
of~$G^x$ consists of
\[
V(G^x) = \bigcup_{0 \leq i <4p_x} \{u^x_{i}, t^x_{i}, b^x_{i}, d^x_{i} \}
,
\] and the edge set of

\begin{align*}
  E(G^x) = \bigcup_{0 \leq i<4p_x} \{ &u^x_{i} t^x_{i}, u^x_{i}
  t^x_{{i+1}}, t^x_{i} u^x_{{i+1}},
  t^x_{i} t^x_{{i+1}},\\
  & t^x_{i} b^x_{i}, b^x_{i} b^x_{{i+1}}, b^x_{i} d^x_{{i+1}}, b^x_{i}
  d^x_{i}, d^x_{i} b^x_{{i+1}} \},
\end{align*}
where the indices behave cyclically modulo~$4p_x$.  The letters for
the vertices are chosen to correspond with top and bottom ($t^x$
and~$b^x$) of tape, and up and down ($u^x$ and~$d^x$).  The
construction is visualized in
Figures~\ref{fig:c4compl-variable-gadget}
and~\ref{fig:c4compl-variable-gadget-enlarged}.
\begin{figure}[t]
  \centering
  \begin{tikzpicture}[every
    node/.style={circle,draw,scale=.8},scale=1,minimum size=2em, inner
    sep=1pt]
    \node (u0) at (-1,4) {$u^x_{0}$};
    \node (t0) at (-1,3) {$t^x_{0}$};
    \node (b0) at (-1,1) {$b^x_{0}$};
    \node (d0) at (-1,0) {$d^x_{0}$};

    \node (u1) at (1,4) {$u^x_{1}$};
    \node (t1) at (1,3) {$t^x_{1}$};
    \node (b1) at (1,1) {$b^x_{1}$};
    \node (d1) at (1,0) {$d^x_{1}$};

    \node (u2) at (3,4) {$u^x_{2}$};
    \node (t2) at (3,3) {$t^x_{2}$};
    \node (b2) at (3,1) {$b^x_{2}$};
    \node (d2) at (3,0) {$d^x_{2}$};

    \node (u3) at (5,4) {$u^x_{3}$};
    \node (t3) at (5,3) {$t^x_{3}$};
    \node (b3) at (5,1) {$b^x_{3}$};
    \node (d3) at (5,0) {$d^x_{3}$};

    \node (u4) at (7,4) {$u^x_{4}$};
    \node (t4) at (7,3) {$t^x_{4}$};
    \node (b4) at (7,1) {$b^x_{4}$};
    \node (d4) at (7,0) {$d^x_{4}$};

    \node[scale=.3,fill] (dot1) at (8.0, 2) {};
    \node[scale=.3,fill] (dot2) at (8.5, 2) {};
    \node[scale=.3,fill] (dot3) at (9.0, 2) {};

    \draw (u0) -- (t0) -- (b0) -- (d0);
    \draw (u1) -- (t1) -- (b1) -- (d1);
    \draw (u2) -- (t2) -- (b2) -- (d2);
    \draw (u3) -- (t3) -- (b3) -- (d3);
    \draw (u4) -- (t4) -- (b4) -- (d4);
    \draw (t0) -- (t1) -- (t2) -- (t3) -- (t4);
    \draw (b0) -- (b1) -- (b2) -- (b3) -- (b4);
    
    \draw (t0) -- (u1) -- (t2) -- (u3) -- (t4);
    \draw (u0) -- (t1) -- (u2) -- (t3) -- (u4);

    \draw (b0) -- (d1) -- (b2) -- (d3) -- (b4);
    \draw (d0) -- (b1) -- (d2) -- (b3) -- (d4);
  \end{tikzpicture}
  \caption{Variable gadget $G_x$.}
  \label{fig:c4compl-variable-gadget-enlarged}
\end{figure}
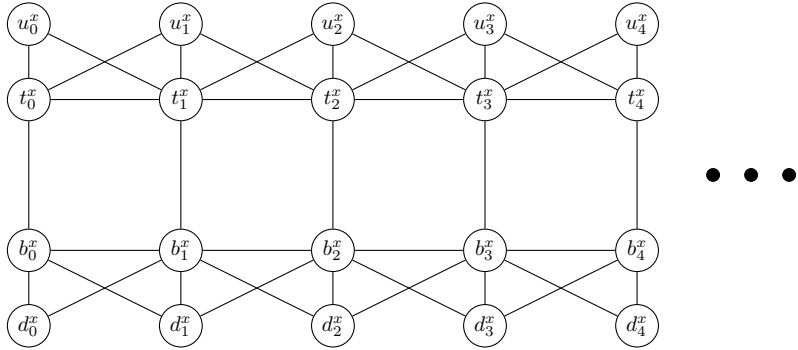

\begin{claim}
  \label{claim:c4-compl-variable-gadget-p-edges}
  The minimum number of edges required to add to~$G^x$ to make it
  $C_4$-free is~$4p_x$.  Moreover, there are exactly two ways of
  eliminating all~$C_4$s with~$4p_x$ edges, namely adding an edge on
  the diagonal for each square.  Furthermore, if we add one edge to
  eliminate some cycle, all the rest must have the same orientation,
  i.e., either all added edges are of the form $t_i^x b_{i+1}^x$ or of
  the form $t_{i+1}^x b_{i}^x$.  See
  Figure~\ref{fig:c4compl-variable-gadgets}.
\end{claim}
\begin{proof}[of claim] A gadget~$G^x$ contains~$4p_x$ induced~$C_4$,
  and no two of them can be eliminated by adding one edge.  Hence, to
  eliminate all~$C_4$s in~$G^x$, we need at least~$4p_x$ edges.  On
  the other hand, it is easy to verify that after adding~$4p_x$
  diagonals to~$C_4$s of the same orientation the resulting graph does
  not contain any induced~$C_4$, see
  Figure~\ref{fig:c4compl-variable-gadgets} for examples.  Whenever we
  have two consecutive cycles with completion edges of different
  orientation, we create a new~$C_4$ consisting of the two completion
  edges, and (depending on their orientation) either two edges
  incident to vertex~$u_i^x$ above their common vertex, or two edges
  incident to vertex~$d_i^x$ below.  See
  Figure~\ref{fig:c4compl-variable-gadget-differing}.
\end{proof}

\begin{corollary}\label{cor:variableEdges}
  The minimum number of edges required to eliminate all $C_4$s
  appearing inside all the variable gadgets is $12|\mathcal{C}(\phi)|$.
\end{corollary}
\begin{proof}
  Since each clause of $\mathcal{C}(\phi)$ contains exactly three
  occurrences of variables, it follows that $\sum_{x\in
    \mathcal{V}(\phi)} p_x=3|\mathcal{C}(\phi)|$.  The constructed
  variable gadgets are pairwise disjoint, so by
  Claim~\ref{claim:c4-compl-variable-gadget-p-edges} we infer that the
  minimum number of edges required in all the variable gadgets is
  equal to $\sum_{x \in \mathcal{V}(\phi)} 4p_x = 3\cdot 4
  |\mathcal{C}(\phi)| = 12|\mathcal{C}(\phi)|$.
\end{proof}

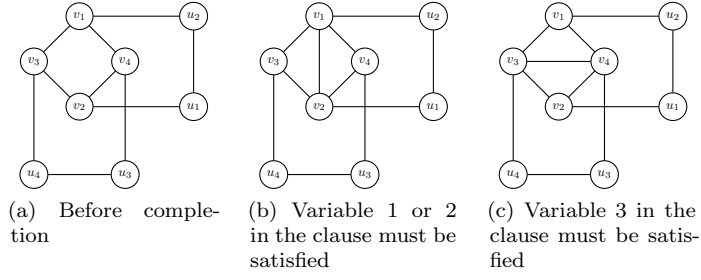
\begin{figure}[t]
  \centering
  \subfloat[Before completion] {
    \centering
    \begin{tikzpicture}[every node/.style={circle, draw, scale=.5},
      scale=.3]
      
      \node (v1) at (3,8) {$v_1$};
      \node (v2) at (3,4) {$v_2$};
      \node (v3) at (1,6) {$v_3$};
      \node (v4) at (5,6) {$v_4$};
      
      \node (u1) at (8,4) {$u_1$};
      \node (u2) at (8,8) {$u_2$};
      \node (u3) at (5,1) {$u_3$};
      \node (u4) at (1,1) {$u_4$};
      
      \draw (v1) -- (v3) -- (v2) -- (v4) -- (v1);
      \draw (v1) -- (u2) -- (u1) -- (v2);
      \draw (v3) -- (u4) -- (u3) -- (v4);
    \end{tikzpicture}
    \label{fig:c4compl-clause-gadget}
  }
  \hspace{.02\textwidth}
  \subfloat[Variable 1 or 2 in the clause must be satisfied] {
    \centering
    \begin{tikzpicture}[every node/.style={circle, draw, scale=.5},
      scale=.3]
      
      \node (v1) at (3,8) {$v_1$};
      \node (v2) at (3,4) {$v_2$};
      \node (v3) at (1,6) {$v_3$};
      \node (v4) at (5,6) {$v_4$};
      
      \node (u1) at (8,4) {$u_1$};
      \node (u2) at (8,8) {$u_2$};
      \node (u3) at (5,1) {$u_3$};
      \node (u4) at (1,1) {$u_4$};
      
      \draw (v1) -- (v3) -- (v2) -- (v4) -- (v1);
      \draw (v1) -- (u2) -- (u1) -- (v2);
      \draw (v3) -- (u4) -- (u3) -- (v4);
      
      \draw (v1) -- (v2);
      
    \end{tikzpicture}
    \label{fig:c4compl-clause-gadget-1-2}
  }
  \hspace{.02\textwidth}
  \subfloat[Variable 3 in the clause must be satisfied] {
    \centering
    \begin{tikzpicture}[every node/.style={circle, draw, scale=.5},
      scale=.3]
      
      \node (v1) at (3,8) {$v_1$};
      \node (v2) at (3,4) {$v_2$};
      \node (v3) at (1,6) {$v_3$};
      \node (v4) at (5,6) {$v_4$};
      
      \node (u1) at (8,4) {$u_1$};
      \node (u2) at (8,8) {$u_2$};
      \node (u3) at (5,1) {$u_3$};
      \node (u4) at (1,1) {$u_4$};
      
      \draw (v1) -- (v3) -- (v2) -- (v4) -- (v1);
      \draw (v1) -- (u2) -- (u1) -- (v2);
      \draw (v3) -- (u4) -- (u3) -- (v4);
      
      \draw (v3) -- (v4);
    \end{tikzpicture}
    \label{fig:c4compl-clause-gadget-3-4}
  }
  \caption{The clause gadget}
  \label{fig:c4compl-clause-gadgets}
\end{figure}

We now proceed to create the clause gadgets.  For each clause $c \in
\mathcal{C}(\phi)$, we create the graph $G^c$ as depicted in
Figure~\ref{fig:c4-free-clause-gadget}.  It consists of an induced
4-cycle $v^c_1 v^c_4 v^c_2 v^c_3$ and induced paths $v^c_2 u^c_1 u^c_2
v^c_1$ and $v^c_3 u^c_4 u^c_3 v^c_4$.  We also attach a gadget
consisting of $k_\phi$ internally disjoint induced paths of four
vertices with endpoints in $v^c_4$ and $u^c_4$, where $k_\phi$ is the
budget to be specified later.  That makes it impossible to add an edge
between $v^c_4$ and $u^c_4$ in any $C_4$-free completion with at most
$k_\phi$ edges.

By the $i$-th \emph{square} we mean a quadruple $(t^x_{i}, b^x_{i},
t^x_{{i+1}}, b^x_{{i+1}})$.  If a clause $c$ is the $\ell$-th clause
the variable $x$ appears in, we will use the vertices of the $4(\ell
-1)$-st square for connections to the gadget corresponding to $c$.
For ease of notation let $j = 4(\ell -1)$.  We also use pairs
$\{v^c_1, u^c_1\}$, $\{v^c_2, u^c_2\}$, and $\{v^c_3, u^c_3\}$ of
$G^c$ for connecting to the corresponding variable gadgets.  If a
variable $x$ appears in a non-negated form as the $i$th (for $i =
1,2,3$) literal of a clause $c$, then we add the edges
$t^x_{j+1}v^c_i$ and $b^x_{j}u^c_i$.  If it appears in a negated form,
we add the edges $t^x_{j}v^c_i$ and $b^x_{j+1}u^c_i$.  See
Figure~\ref{fig:c4-free-compl-connections-enlarged}.  This concludes
the construction of $G_\phi$.  Finally, we set the budget for the
instance equal to $k_\phi=14|\mathcal{C}(\phi)|$.

\begin{figure}[t]
  \centering
  \begin{tikzpicture} [every node/.style={circle, draw, scale=.8},
    scale=.8, minimum size=1.5em, inner sep=0pt]
    
    % CLAUSE!!!!
    \node (cv1) at (11,8) {$v^c_1$}; \node (cv2) at (11,6) {$v^c_2$};
    \node (cv3) at (10,7) {$v^c_3$}; \node (cv4) at (12,7) {$v^c_4$};
    
    \node (cu1) at (14,6) {$u^c_1$}; \node (cu2) at (14,8) {$u^c_2$};
    \node (cu3) at (12,4) {$u^c_3$}; \node (cu4) at (10,4) {$u^c_4$};
    
    \draw (cv1) -- (cv4) -- (cv2) -- (cv3) -- (cv1) -- (cu2) -- (cu1) -- (cv2);
    \draw (cv3) -- (cu4) -- (cu3) -- (cv4);

    % kK_2 clause gadget
    \node[scale=0.3] (k11) at (12.50, 3.50) {};
    \node[scale=0.3] (k12) at (13.00, 4.00) {};
    
    \node[scale=0.3] (k21) at (12.75, 3.25) {};
    \node[scale=0.3] (k22) at (13.25, 3.75) {};
    
    \node[scale=0.3] (k31) at (13.00, 3.00) {};
    \node[scale=0.3] (k32) at (13.50, 3.50) {};
    
    \node[scale=0.3] (k41) at (13.25, 2.75) {};
    \node[scale=0.3] (k42) at (13.75, 3.25) {};
    
    \node[scale=0.3] (k51) at (13.50, 2.50) {};
    \node[scale=0.3] (k52) at (14.00, 3.00) {};
    
    \draw (cu4) -- (k11) -- (k12) -- (cv4);
    \draw (cu4) -- (k21) -- (k22) -- (cv4);
    \draw (cu4) -- (k31) -- (k32) -- (cv4);
    \draw (cu4) -- (k41) -- (k42) -- (cv4);
    \draw (cu4) -- (k51) -- (k52) -- (cv4);
    
    \draw [decorate, decoration={brace}, xshift=-1.5em, yshift=-1.5em]
    (13.50, 2.5) -- (12.50, 3.50) 
    node [scale=1.5, draw=none, midway, xshift=-1.5em, yshift=-1.5em,
    rotate=45] {$k_\phi K_2$};
  \end{tikzpicture}
  \caption{The clause gadget $G^c$.  It contains one $C_4$, and if we add
    either edge $v^c_1v^c_2$ or edge $v^c_3v^c_4$, we get a new $C_4$
    that must be destroyed by adding one more edge.  The $k_\phi K_2$
    gadget makes sure we cannot add edge $v^c_4u^c_4$.}
  \label{fig:c4-free-clause-gadget}
\end{figure}
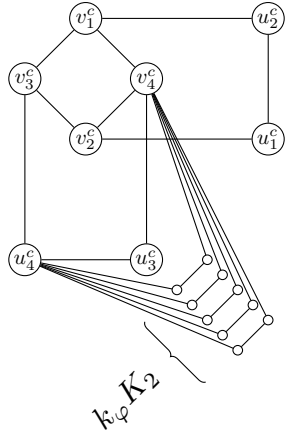

%
% 
% start connection of c4 completion
% 
%
\begin{figure}[t]
  \centering
  \begin{tikzpicture}[every node/.style={circle, draw, scale=.8},
    scale=.8, minimum size=1.8em, inner sep=0pt]
    
    \node (xd0)	at (0,10) {$d^x_0$};
    \node (xd1)	at (1,10) {$d^x_1$};
    \node (xd2)	at (2,10) {$d^x_2$};
    \node (xd3)	at (3,10) {$d^x_3$};
    \node (xd4)	at (4,10) {$d^x_4$};
    \node (xd5)	at (5,10) {$d^x_5$};
    \node (xb0)	at (0,11) {$b^x_0$};
    \node (xb1)	at (1,11) {$b^x_1$};
    \node (xb2)	at (2,11) {$b^x_2$};
    \node (xb3)	at (3,11) {$b^x_3$};
    \node (xb4)	at (4,11) {$b^x_4$};
    \node (xb5)	at (5,11) {$b^x_5$};
    \node (xt0)	at (0,12) {$t^x_0$};
    \node (xt1)	at (1,12) {$t^x_1$};
    \node (xt2)	at (2,12) {$t^x_2$};
    \node (xt3)	at (3,12) {$t^x_3$};
    \node (xt4)	at (4,12) {$t^x_4$};
    \node (xt5)	at (5,12) {$t^x_5$};
    \node (xu0)	at (0,13) {$u^x_0$};
    \node (xu1)	at (1,13) {$u^x_1$};
    \node (xu2)	at (2,13) {$u^x_2$};
    \node (xu3)	at (3,13) {$u^x_3$};
    \node (xu4)	at (4,13) {$u^x_4$};
    \node (xu5)	at (5,13) {$u^x_5$};

    \node (xd6)	at (6,10) {$d^x_6$};
    \node (xb6)	at (6,11) {$b^x_6$};
    \node (xt6)	at (6,12) {$t^x_6$};
    \node (xu6)	at (6,13) {$u^x_6$};
    
    \node (yd3)	at (0,5) {$d^y_{7}$};
    \node (yd4)	at (1,5) {$d^y_{8}$};
    \node (yd5)	at (2,5) {$d^y_{9}$};
    \node (yd6)	at (3,5) {$d^y_{10}$};
    \node (yd7)	at (4,5) {$d^y_{11}$};
    \node (yd8)	at (5,5) {$d^y_{12}$};
    \node (yd9)	at (6,5) {$d^y_{13}$};
    \node (yb3)	at (0,6) {$b^y_{7}$};
    \node (yb4)	at (1,6) {$b^y_{8}$};
    \node (yb5)	at (2,6) {$b^y_{9}$};
    \node (yb6)	at (3,6) {$b^y_{10}$};
    \node (yb7)	at (4,6) {$b^y_{11}$};
    \node (yb8)	at (5,6) {$b^y_{12}$};
    \node (yb9)	at (6,6) {$b^y_{13}$};
    \node (yt3)	at (0,7) {$t^y_{7}$};
    \node (yt4)	at (1,7) {$t^y_{8}$};
    \node (yt5)	at (2,7) {$t^y_{9}$};
    \node (yt6)	at (3,7) {$t^y_{10}$};
    \node (yt7)	at (4,7) {$t^y_{11}$};
    \node (yt8)	at (5,7) {$t^y_{12}$};
    \node (yt9)	at (6,7) {$t^y_{13}$};
    \node (yu3)	at (0,8) {$u^y_{7}$};
    \node (yu4)	at (1,8) {$u^y_{8}$};
    \node (yu5)	at (2,8) {$u^y_{9}$};
    \node (yu6)	at (3,8) {$u^y_{10}$};
    \node (yu7)	at (4,8) {$u^y_{11}$};
    \node (yu8)	at (5,8) {$u^y_{12}$};
    \node (yu9)	at (6,8) {$u^y_{13}$};
    
    \node (zd7)	at (0,0) {$d^z_{7}$};
    \node (zd8)	at (1,0) {$d^z_{8}$};
    \node (zd9)	at (2,0) {$d^z_{9}$};
    \node (zd10)	at (3,0) {$d^z_{10}$};
    \node (zd11)	at (4,0) {$d^z_{11}$};
    \node (zd12)	at (5,0) {$d^z_{12}$};
    \node (zd13)	at (6,0) {$d^z_{13}$};
    \node (zb7)	at (0,1) {$b^z_{7}$};
    \node (zb8)	at (1,1) {$b^z_{8}$};
    \node (zb9)	at (2,1) {$b^z_{9}$};
    \node (zb10)	at (3,1) {$b^z_{10}$};
    \node (zb11)	at (4,1) {$b^z_{11}$};
    \node (zb12)	at (5,1) {$b^z_{12}$};
    \node (zb13)	at (6,1) {$b^z_{13}$};
    \node (zt7)	at (0,2) {$t^z_{7}$};
    \node (zt8)	at (1,2) {$t^z_{8}$};
    \node (zt9)	at (2,2) {$t^z_{9}$};
    \node (zt10)	at (3,2) {$t^z_{10}$};
    \node (zt11)	at (4,2) {$t^z_{11}$};
    \node (zt12)	at (5,2) {$t^z_{12}$};
    \node (zt13)	at (6,2) {$t^z_{13}$};
    \node (zu7)	at (0,3) {$u^z_{7}$};
    \node (zu8)	at (1,3) {$u^z_{8}$};
    \node (zu9)	at (2,3) {$u^z_{9}$};
    \node (zu10)	at (3,3) {$u^z_{10}$};
    \node (zu11)	at (4,3) {$u^z_{11}$};
    \node (zu12)	at (5,3) {$u^z_{12}$};
    \node (zu13)	at (6,3) {$u^z_{13}$};
    
    % x
    % x horizontal
    \draw (xt0) -- (xt1) -- (xt2) -- (xt3) -- (xt4) -- (xt5) -- (xt6);
    \draw (xb0) -- (xb1) -- (xb2) -- (xb3) -- (xb4) -- (xb5) -- (xb6);
    
    % x vertical
    \draw (xd0) -- (xb0) -- (xt0) -- (xu0);
    \draw (xd1) -- (xb1) -- (xt1) -- (xu1);
    \draw (xd2) -- (xb2) -- (xt2) -- (xu2);
    \draw (xd3) -- (xb3) -- (xt3) -- (xu3);
    \draw (xd4) -- (xb4) -- (xt4) -- (xu4);
    \draw (xd5) -- (xb5) -- (xt5) -- (xu5);
    \draw (xd6) -- (xb6) -- (xt6) -- (xu6);
    
    % x diagonal
    \draw (xd0) -- (xb1) -- (xd2) -- (xb3) -- (xd4) -- (xb5) -- (xd6);
    \draw (xb0) -- (xd1) -- (xb2) -- (xd3) -- (xb4) -- (xd5) -- (xb6);
    
    \draw (xu0) -- (xt1) -- (xu2) -- (xt3) -- (xu4) -- (xt5) -- (xu6);
    \draw (xt0) -- (xu1) -- (xt2) -- (xu3) -- (xt4) -- (xu5) -- (xt6);
    
    % x variable square
    \draw[line width=1mm]  (xb0) -- (xb1) -- (xt1) -- (xt0) -- (xb0);
    
    % y
    \draw (yt3) -- (yt4) -- (yt5) -- (yt6) -- (yt7) -- (yt8) -- (yt9);
    \draw (yb3) -- (yb4) -- (yb5) -- (yb6) -- (yb7) -- (yb8) -- (yb9);
    
    \draw (yd3) -- (yb3) -- (yt3) -- (yu3);
    \draw (yd4) -- (yb4) -- (yt4) -- (yu4);
    \draw (yd5) -- (yb5) -- (yt5) -- (yu5);
    \draw (yd6) -- (yb6) -- (yt6) -- (yu6);
    \draw (yd7) -- (yb7) -- (yt7) -- (yu7);
    \draw (yd8) -- (yb8) -- (yt8) -- (yu8);
    \draw (yd9) -- (yb9) -- (yt9) -- (yu9);
    
    % y diagonal
    \draw (yd3) -- (yb4) -- (yd5) -- (yb6) -- (yd7) -- (yb8) -- (yd9);
    \draw (yb3) -- (yd4) -- (yb5) -- (yd6) -- (yb7) -- (yd8) -- (yb9);
    
    \draw (yu3) -- (yt4) -- (yu5) -- (yt6) -- (yu7) -- (yt8) -- (yu9);
    \draw (yt3) -- (yu4) -- (yt5) -- (yu6) -- (yt7) -- (yu8) -- (yt9);
    
    % y variable square
    \draw[line width=1mm]  (yb4) -- (yb5) -- (yt5) -- (yt4) -- (yb4);

    % z
    \draw (zt7) -- (zt8) -- (zt9) -- (zt10) -- (zt11) -- (zt12) -- (zt13);
    \draw (zb7) -- (zb8) -- (zb9) -- (zb10) -- (zb11) -- (zb12) -- (zb13);

    \draw (zd7) -- (zb7) -- (zt7) -- (zu7);
    \draw (zd8) -- (zb8) -- (zt8) -- (zu8);
    \draw (zd9) -- (zb9) -- (zt9) -- (zu9);
    \draw (zd10) -- (zb10) -- (zt10) -- (zu10);
    \draw (zd11) -- (zb11) -- (zt11) -- (zu11);
    \draw (zd12) -- (zb12) -- (zt12) -- (zu12);
    \draw (zd13) -- (zb13) -- (zt13) -- (zu13);
    
    % z diagonal
    \draw (zd7) -- (zb8) -- (zd9) -- (zb10) -- (zd11) -- (zb12) -- (zd13);
    \draw (zb7) -- (zd8) -- (zb9) -- (zd10) -- (zb11) -- (zd12) -- (zb13);
    
    \draw (zu7) -- (zt8) -- (zu9) -- (zt10) -- (zu11) -- (zt12) -- (zu13);
    \draw (zt7) -- (zu8) -- (zt9) -- (zu10) -- (zt11) -- (zu12) -- (zt13);
    
    % z variable square
    \draw[line width=1mm]  (zb8) -- (zb9) -- (zt9) -- (zt8) -- (zb8);

    % CLAUSE!!!!
    \node (cv1) at (11,8) {$v^c_1$}; \node (cv2) at (11,6) {$v^c_2$};
    \node (cv3) at (10,7) {$v^c_3$}; \node (cv4) at (12,7) {$v^c_4$};
    
    \node (cu1) at (14,6) {$u^c_1$}; \node (cu2) at (14,8) {$u^c_2$};
    \node (cu3) at (12,4) {$u^c_3$}; \node (cu4) at (10,4) {$u^c_4$};
    
    \draw (cv1) -- (cv4) -- (cv2) -- (cv3) -- (cv1) -- (cu2) -- (cu1) -- (cv2);
    \draw (cv3) -- (cu4) -- (cu3) -- (cv4);
    
    % x -- clause
    \node[scale=0.1,draw=none] (xt1h1) at (2,13.5) {};
    \draw (xt1) to[out=60,in=180] (xt1h1);
    \node[scale=0.1,draw=none] (xt1h2) at (6,13.5) {};
    \draw (xt1h1) to[out=0,in=180] (xt1h2);
    \draw (xt1h2) to[out=0,in=90,looseness=1] (cv1);
    
    \node[scale=0.1,draw=none] (xb0h1) at (1,9.5) {};
    \draw (xb0) to[out=-60,in=180] (xb0h1);
    \node[scale=0.1,draw=none] (xb0h2) at (8,9.5) {};
    \draw (xb0h1) to[out=0,in=180] (xb0h2);
    \draw (xb0h2) to[out=0,in=120] (cu1);

    % y -- clause
    \node[scale=0.1,draw=none] (yt4h1) at (2,8.5) {};
    \draw (yt4) to[out=60,in=180] (yt4h1);
    \node[scale=0.1,draw=none] (yt4h2) at (9,8.5) {};
    \draw (yt4h1) to[out=0,in=180] (yt4h2);
    \draw (yt4h2) to[out=0,in=90,looseness=1] (cv2);

    \node[scale=0.1,draw=none] (yb5h1) at (3,4.5) {};
    \draw (yb5) to[out=-60,in=180] (yb5h1);
    \node[scale=0.1,draw=none] (yb5h2) at (9,4.5) {};
    \draw (yb5h1) to[out=0,in=180] (yb5h2);
    \draw (yb5h2) to[out=0,in=-120] (cu2);

    % z -- clause
    \node[scale=0.1,draw=none] (zt9h1) at (3,3.5) {};
    \draw (zt9) to[out=60,in=180] (zt9h1);
    \node[scale=0.1,draw=none] (zt9h2) at (7,3.5) {};
    \draw (zt9h1) to[out=0,in=180] (zt9h2);
    \draw (zt9h2) to[out=0,in=-120] (cv3);
    
    \node[scale=0.1,draw=none] (zb8h1) at (2,-0.5) {};
    \draw (zb8) to[out=-60,in=180] (zb8h1);
    \node[scale=0.1,draw=none] (zb8h2) at (7,-0.5) {};
    \draw (zb8h1) to[out=0,in=180] (zb8h2);
    \draw (zb8h2) to[out=0,in=-90] (cu3);
    
    % kK_2 clause gadget
    \node[scale=0.3] (k11) at (12.50, 3.50) {};
    \node[scale=0.3] (k12) at (13.00, 4.00) {};
    
    \node[scale=0.3] (k21) at (12.75, 3.25) {};
    \node[scale=0.3] (k22) at (13.25, 3.75) {};
    
    \node[scale=0.3] (k31) at (13.00, 3.00) {};
    \node[scale=0.3] (k32) at (13.50, 3.50) {};
    
    \node[scale=0.3] (k41) at (13.25, 2.75) {};
    \node[scale=0.3] (k42) at (13.75, 3.25) {};
    
    \node[scale=0.3] (k51) at (13.50, 2.50) {};
    \node[scale=0.3] (k52) at (14.00, 3.00) {};
    
    \draw (cu4) -- (k11) -- (k12) -- (cv4);
    \draw (cu4) -- (k21) -- (k22) -- (cv4);
    \draw (cu4) -- (k31) -- (k32) -- (cv4);
    \draw (cu4) -- (k41) -- (k42) -- (cv4);
    \draw (cu4) -- (k51) -- (k52) -- (cv4);
    
    \draw [decorate,decoration={brace},xshift=-1em,yshift=-1em]
    (13.50,2.5) -- (12.50, 3.50) node
    [scale=1.5,draw=none,midway,xshift=-.65em,yshift=-1.25em,rotate=45]
    {$kK_2$};
    
  \end{tikzpicture}
  \caption{The connections for a clause $c = x \lor \neg y \lor z$.  In this
    example, $c$ is the first clause of appearance for $x$ thus $x$ is
    connected to $G^c$ via the 0th square.  For $y$ and $z$, we assume
    that $c$ is the third clause they appear, thus $y$ and $z$ use the
    8th square.}
  \label{fig:c4-free-compl-connections-enlarged}
\end{figure}
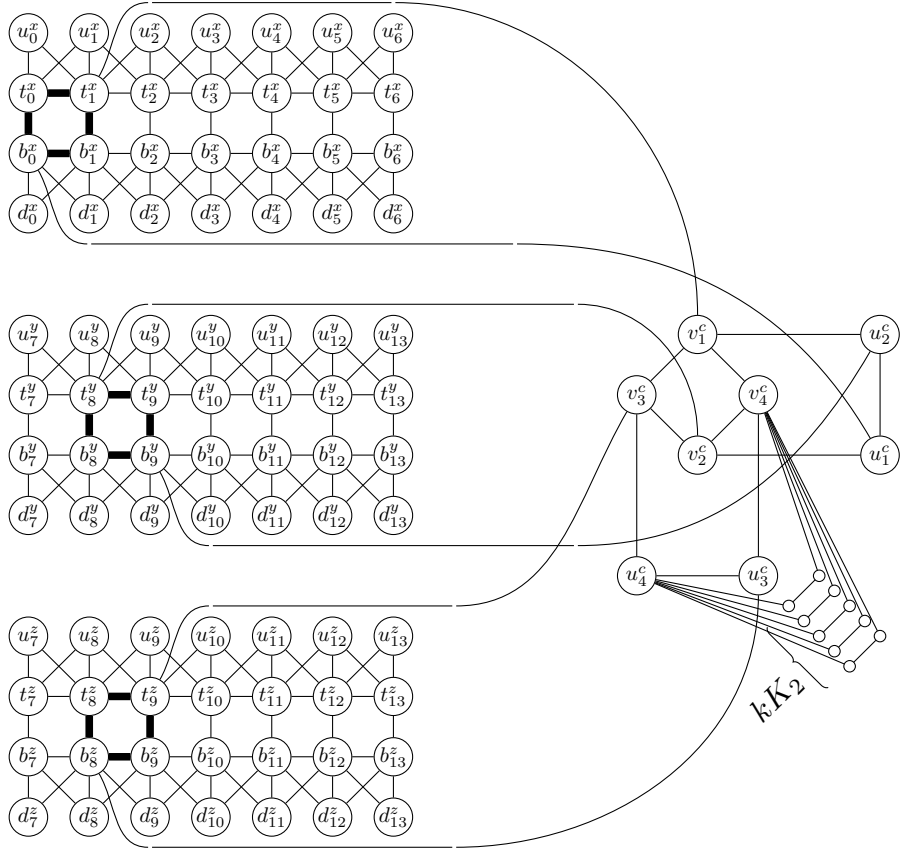

% 
% end of connection of c4 completion
% 

\begin{claim}
  \label{claim:c4-compl-clause-gadget-two-edges}
  For each clause gadget $G^c$ for a clause $c \in \mathcal{C}(\phi)$, we
  need to add at least \emph{two edges} between vertices of $G^c$ to
  eliminate all induced $C_4$s in $G^c$.  Moreover, there are exactly
  three ways of adding exactly two edges to $G^c$ so that the
  resulting graph does not contain any induced $C_4$: by adding
  $\{v_1^c v_2^c,v_1^c u_1^c\}$, $\{v_1^c v_2^c,v_2^c u_2^c\}$, or
  $\{v_3^c v_4^c,v_3^c u_3^c\}$.
\end{claim}
\begin{proof}[of claim]
  There is a four-cycle $v^c_1v^c_4v^c_2v^c_3$ which needs to be
  eliminated, either by adding the edge $v^c_1v^c_2$
  (Figure~\ref{fig:c4compl-clause-gadget-1-2}) or $v^c_3v^c_4$
  (Figure~\ref{fig:c4compl-clause-gadget-3-4}).  In any case we create
  a new $C_4$, either $v^c_1u^c_2u^c_1v^c_2$ in the former case, and
  $v^c_4u^c_3u^c_4v^c_3$ in the latter case.  In the former case we
  can eliminate the created $C_4$ by adding $v^c_1u^c_1$ or
  $v^c_2u^c_2$, and in the latter case we can eliminate it by adding
  $v^c_3u^c_3$.  Note that in the latter case we cannot add
  $v^c_4u^c_4$, since then we would create $k_\phi$ new induced
  four-cycles.  A direct check shows that all the three aforementioned
  completion sets lead to a $C_4$-free graph.
\end{proof}

\begin{lemma}
  Given a \pname{3Sat} instance $\phi$, we have that $(G_\phi, k_\phi)$ is a
  \yes{} instance for \pname{$C_4$-Free Completion} for $k_\phi = 14
  |\mathcal{C}(\phi)|$ if and only if $\phi$ is satisfiable.
\end{lemma}
\begin{proof}
  From right to left, suppose $\phi$ is satisfiable.  Let $\alpha\colon
  \mathcal{V}(\phi) \to \{\mathtt{true}, \mathtt{false}\}$ be a
  satisfying assignment for $\phi$.  For every variable $x \in
  \mathcal{V}(\phi)$, if $\alpha(x) = \mathtt{true}$, we add edges
  $t_i^x b^x_{i+1}$ to $S$ for $i \in \{0, \ldots, 4p_x-1\}$ and if
  $\alpha(x) = \mathtt{false}$, we add edges $t_{i+1}^x b^x_{i}$ to
  $S$ for $i \in \{0, \ldots, 4p_x-1\}$.
  
  For a clause $c$ in $\mathcal{C}(\phi)$, if the first literal satisfies
  the clause, we add the edges $v_1^c v_2^c$ and $v_1^c u_1^c$ to $S$.
  If the second literal satisfies the clause, we add $v_1^c v_2^c$ and
  $v_2^c u_2^c$ to $S$ and if it is the third literal, we add $v_3^c
  v_4^c$ and $v_3^c u_3^c$ to $S$.  If more than one literal satisfies
  the clause, we pick any.  In total this makes
  $12|\mathcal{C}(\phi)|$ edges added to the variable gadgets and
  $2|\mathcal{C}(\phi)|$ edges added to the clause gadgets.
  
  Suppose now for a contradiction that $G_\phi+S$ contains a cycle $L$ of
  length four.  In Claims~\ref{claim:c4-compl-variable-gadget-p-edges}
  and~\ref{claim:c4-compl-clause-gadget-two-edges} it is already
  verified that $L$ is not completely contained in a variable or
  clause gadget.  Each vertex has at most one incident edge ending
  outside the gadget of the vertex and there are only edges between
  variable and clause gadgets.  Thus $L$ consist of one edge from a
  variable gadget and one from a clause gadget and two edges between.
  We can observe that $L$ then must contain either $v^c_1u^c_1$,
  $v^c_2u^c_2$, or $v^c_3u^c_3$ of the clause gadget, see
  Figure~\ref{fig:c4-free-compl-connections-enlarged}.  Let us assume
  without loss of generality that $L$ contains the edge $v^c_1u^c_1$.
  By the construction of the set $S$ this implies that the literal of
  the first variable $x$ of $c$ satisfies $c$.  If $x$ is non-negated
  in $c$, then we have that $\alpha(x)=\mathtt{true}$ and that $v^c_1
  t_{j+1}^x$ and $u^c_1 b_j^x$ are edges of $L$.  To complete the
  cycle $t_{j+1}^x b_j^x$ must be an edge of $L$; however, by the
  definition of $S$ we have added the edge $t_j^x b_{j+1}^x$ to $S$
  instead of $t_{j+1}^x b_j^x$, and we obtain a contradiction.  The
  case where $x$ is negated is symmetric.
    
  \medskip
  
  From left to right, suppose $(G_\phi, k_\phi)$ is a \yes{} instance for
  $k_\phi = 14 |\mathcal{C}(\phi)| $ and let $S$ be such that $G_\phi
  + S$ is $C_4$-free with $|S| \leq k_\phi$.  By Corollary
  \ref{cor:variableEdges} and Claim
  \ref{claim:c4-compl-clause-gadget-two-edges} we know that we need to
  use at least $12 |\mathcal{C}(\phi)|$ edges to fix the variable
  gadgets and we need to use at least $2 |\mathcal{C}(\phi)|$ edges
  for the clause gadgets.  Since $|S| \leq k_\phi$, we infer that
  $|S|=k_\phi$, that we use at exactly $4p_x$ edges to fix each
  variable gadgets $G^x$ (and that the orientation of the added edges
  must be the same within the gadget), that we use exactly two edges
  for each clause gadget $G^c$, and that $S$ contains no edges other
  than the mentioned above.
  
  We now define an assignment $\alpha$ for $\mathcal{V}(\phi)$ and prove that
  it is indeed a satisfying assignment.  If $S$ contains the edge
  $t^x_0 b^x_1$, we let $\alpha(x) = \mathtt{true}$, and if $S$
  contains the edge $t^x_1 b^x_0$ we let $\alpha(x) = \mathtt{false}$.
  Let $c \in \mathcal{C}(\phi)$ be a clause and suppose that $c$ is
  not satisfied.  We know by Claim
  \ref{claim:c4-compl-clause-gadget-two-edges} that the gadget for $c$
  contains $\{v^c_1v^c_2,v^c_1u^c_1\}$ or $\{v^c_1v^c_2,v^c_2u^c_2\}$,
  or $\{v^c_3v^c_4,v^c_3u^c_3\}$.
  
  Without loss of generality assume that $G^c$ contains $\{v^c_1v^c_2,
  v^c_1u^c_1\}$ and that $x$ is the first variable in $c$, and it
  appears non-negated.  Since $x$ does not satisfy $c$, we infer that
  $\alpha(x)=\mathtt{false}$.  This means that $t^x_1 b^x_0\in S$, and
  since the orientation of the added edges in the gadget $G^x$ is the
  same, then also $t^x_{i+1}b^x_i\in S$.  As a result, both edges
  $t^x_{i+1}b^x_i$ and $v^c_1u^c_1$ are present in $G_\phi+S$.  But
  then we have an induced four-cycle $v^c_1 u^c_1 b^x_i t^x_{i+1}
  v^c_1$, contradicting the assumption that $G_\phi+S$ was $C_4$-free.
  The cases for $y$, $z$ and negative literals are symmetric.  This
  concludes the proof.
\end{proof}

Similarly as before, the proof of Theorem~\ref{thm:c4completion-exp}
can be completed as follows: combining the presented reduction with an
algorithm for \pname{$C_4$-Free Completion} working in
$2^{o(k)}n^{\bigO(1)}$ time would give an algorithm for \pname{3Sat}
working in $2^{o(n+m)}(n + m)^{\bigO(1)}$ time, which contradicts ETH
by the results of Impagliazzo, Paturi and
Zane~\cite{impagliazzo2001which}.

% 
% 
% 
% P_4 free completion
% 
% 
% 
\subsection{$P_4$-free completion is not solvable in subexponential
  time}
In this section we show that there is no subexponential algorithm for
\fd{} for $\F = \{P_4\}$ unless the ETH fails.  Let us recall that
since $\overline{P_4} = P_4$, the problems \pname{$P_4$-Free Edge
  Deletion} and \pname{$P_4$-Free Completion} are polynomial time
equivalent, and that this graph class more commonly goes under the
name \emph{cographs}.  In other words, we aim to convince the reader
of the following.

\begin{theorem}
  \label{thm:p4-completion-exp}
  The problem \pname{$P_4$-Free Completion} is not solvable in
  $2^{o(k)}n^{\bigO(1)}$ time unless ETH fails.
\end{theorem}

We reduce from \pname{3Sat} to the complement problem
\pname{$P_4$-Free Edge Deletion}.  Let~$\phi$ be the input
\pname{3Sat} formula, where we again assume that every clause of
$\phi$ contains exactly three literals corresponding to pairwise
different variables.  For a variable $x \in \mathcal{V}(\phi)$ we
denote by~$p_x$ the number of clauses in~$\phi$ containing~$x$.  Note
that since each clause contains exactly three variables, we have that
$\sum_{x \in \mathcal{V}(\phi)} p_x=3|\mathcal{C}(\phi)|$.  We
construct a graph~$G_\phi$ such that for $k_\phi =
4|\mathcal{C}(\phi)| + \sum_{x \in \mathcal{V}(\phi)} 4p_x =
16|\mathcal{C}(\phi)|$,~$\phi$ is satisfiable if and only if
$(G_\phi,k_\phi)$ is a \yes{} instance of \pname{$P_4$-Free Edge
  Deletion}.  Since the complement of~$P_4$ is~$P_4$, this will prove
the theorem.

\begin{figure}[t]
  \centering
  \begin{tikzpicture}[every node/.style={circle, draw, scale=.3},
    scale=.5]

    % LEFT SIDE STACK
    \node (s11) at (2.5,4.0) {};
    \node (s12) at (2.0,4.0) {};
    \node (s13) at (1.5,4.0) {};
    \node (s14) at (1.0,4.0) {};
    \node (s15) at (0.5,4.0) {};

    % FIRST PAIR
    \node (b11) at (3.0,5.0) {};
    \node (b12) at (3.5,6.0) {};
    \node (b13) at (4.0,5.0) {};
    \node (b14) at (4.5,6.0) {};
    \node (b15) at (5.0,5.0) {};
    
    \node (t11) at (3.5,8.0) {};
    \node (t12) at (4.5,8.0) {};
    
    \node (s21) at (6.0,5.5) {};
    \node (s22) at (6.0,6.0) {};
    \node (s23) at (6.0,6.5) {};
    \node (s24) at (6.0,7.0) {};
    \node (s25) at (6.0,7.5) {};

    % SECOND PAIR
    \node (b21) at (7.0,5.0) {};
    \node (b22) at (7.5,6.0) {};
    \node (b23) at (8.0,5.0) {};
    \node (b24) at (8.5,6.0) {};
    \node (b25) at (9.0,5.0) {};
    
    \node (t21) at (7.5,8.0) {};
    \node (t22) at (8.5,8.0) {};
    
    \node (s31) at (10.0,5.5) {};
    \node (s32) at (10.0,6.0) {};
    \node (s33) at (10.0,6.5) {};
    \node (s34) at (10.0,7.0) {};
    \node (s35) at (10.0,7.5) {};

    % THIRD PAIR
    \node (b31) at (11.0,5.0) {};
    \node (b32) at (11.5,6.0) {};
    \node (b33) at (12.0,5.0) {};
    \node (b34) at (12.5,6.0) {};
    \node (b35) at (13.0,5.0) {};
    
    \node (t31) at (11.5,8.0) {};
    \node (t32) at (12.5,8.0) {};
    
    % RIGHT SIDE STACK
    \node (s41) at (13.5,4.0) {};
    \node (s42) at (14.0,4.0) {};
    \node (s43) at (14.5,4.0) {};
    \node (s44) at (15.0,4.0) {};
    \node (s45) at (15.5,4.0) {};

    % FOURTH PAIR
    \node (b41) at (13.0,3.0) {};
    \node (b42) at (12.5,2.0) {};
    \node (b43) at (12.0,3.0) {};
    \node (b44) at (11.5,2.0) {};
    \node (b45) at (11.0,3.0) {};
            
    \node (t41) at (12.5,0.0) {};
    \node (t42) at (11.5,0.0) {};

    \node (s51) at (10.0,2.5) {};
    \node (s52) at (10.0,2.0) {};
    \node (s53) at (10.0,1.5) {};
    \node (s54) at (10.0,1.0) {};
    \node (s55) at (10.0,0.5) {};

    % FIFTH PAIR
    \node (b51) at (9.0,3.0) {};
    \node (b52) at (8.5,2.0) {};
    \node (b53) at (8.0,3.0) {};
    \node (b54) at (7.5,2.0) {};
    \node (b55) at (7.0,3.0) {};
            
    \node (t51) at (8.5,0.0) {};
    \node (t52) at (7.5,0.0) {};

    \node (s61) at (6.0,2.5) {};
    \node (s62) at (6.0,2.0) {};
    \node (s63) at (6.0,1.5) {};
    \node (s64) at (6.0,1.0) {};
    \node (s65) at (6.0,0.5) {};

    % SIXTH PAIR
    \node (b61) at (5.0,3.0) {};
    \node (b62) at (4.5,2.0) {};
    \node (b63) at (4.0,3.0) {};
    \node (b64) at (3.5,2.0) {};
    \node (b65) at (3.0,3.0) {};
            
    \node (t61) at (4.5,0.0) {};
    \node (t62) at (3.5,0.0) {};

    \draw[Blue] (b11) -- 
                (b12) --
                (b13) --
                (b14) --
                (b15) --
                (b13) --
                (b11);

    \draw[Green,thick](b12) -- 
                (t11);

    \draw[red, thick]  (b14) -- 
                (t12);

    \draw[Blue] (b15) -- 
                (s21) --
                (b21) --
                (s22) --
                (b15) --
                (s23) --
                (b21) --
                (s24) --
                (b15) --
                (s25) --
                (b21) --
                (b15);

                % 2

    \draw[Blue] (b21) -- 
                (b22) --
                (b23) --
                (b24) --
                (b25) --
                (b23) --
                (b21);
                  
    \draw[Green,thick](b22) -- 
                (t21);
                  
    \draw[red, thick]  (b24) -- 
                (t22);

    \draw[Blue] (b25) -- 
                (s31) --
                (b31) --
                (s32) --
                (b25) --
                (s33) --
                (b31) --
                (s34) --
                (b25) --
                (s35) --
                (b31) --
                (b25);

                % 3
    \draw[Blue] (b31) -- 
                (b32) --
                (b33) --
                (b34) --
                (b35) --
                (b33) --
                (b31);
                  
    \draw[Green,thick](b32) -- 
                (t31);
                  
    \draw[red, thick]  (b34) -- 
                (t32);

    \draw[Blue] (b35) -- 
                (s41) --
                (b41) --
                (s42) --
                (b35) --
                (s43) --
                (b41) --
                (s44) --
                (b35) --
                (s45) --
                (b41) --
                (b35);

                % 4
    \draw[Blue] (b41) -- 
                (b42) --
                (b43) --
                (b44) --
                (b45) --
                (b43) --
                (b41);
                  
    \draw[Green,thick](b42) -- 
                (t41);
                  
    \draw[red, thick]  (b44) -- 
                (t42);

    \draw[Blue] (b45) -- 
                (s51) --
                (b51) --
                (s52) --
                (b45) --
                (s53) --
                (b51) --
                (s54) --
                (b45) --
                (s55) --
                (b51) --
                (b45);

                % 5
    \draw[Blue] (b51) -- 
                (b52) --
                (b53) --
                (b54) --
                (b55) --
                (b53) --
                (b51);
                  
    \draw[Green,thick](b52) -- 
                (t51);
                  
    \draw[red, thick]  (b54) -- 
                (t52);

    \draw[Blue] (b55) -- 
                (s61) --
                (b61) --
                (s62) --
                (b55) --
                (s63) --
                (b61) --
                (s64) --
                (b55) --
                (s65) --
                (b61) --
                (b55);

                % 6
    \draw[Blue] (b61) -- 
                (b62) --
                (b63) --
                (b64) --
                (b65) --
                (b63) --
                (b61);
                  
    \draw[Green,thick](b62) -- 
                (t61);
                  
    \draw[red, thick]  (b64) -- 
                (t62);

    \draw[Blue] (b65) -- 
                (s11) --
                (b11) --
                (s12) --
                (b65) --
                (s13) --
                (b11) --
                (s14) --
                (b65) --
                (s15) --
                (b11) --
                (b65);
  \end{tikzpicture}
  \caption{Variable gadget $G_x$ for a variable appearing in six clauses
    in $\phi$, i.e., $p_x = 6$.  Deleting the leftmost edges in each
    tower pair corresponds to setting $x$ to false and deleting the
    rightmost edge in each tower pair corresponds to setting $x$ to
    true.}
  \label{fig:p4-variable-gadget}
\end{figure}
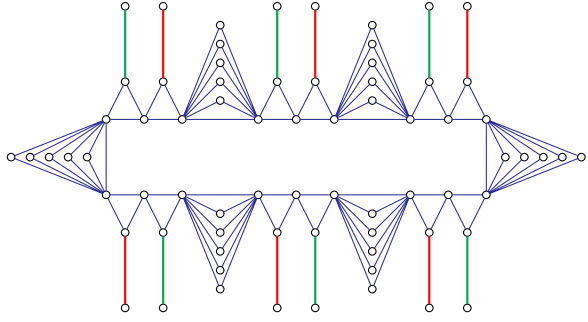

\paragraph{Variable gadget}
For each variable $x \in \mathcal{V}(\phi)$, we create a gadget~$G^x$ which
looks like the one given in Figure~\ref{fig:p4-variable-gadget}.
Before providing the construction formally, let us first describe it
informally.  We call a triangle with a pendant vertex a \emph{tower},
where the triangle will be referred to as the \emph{base} of the
tower, and the pendant vertex the \emph{spike} of the tower.  The
towers will always come in pairs, and they are joined in one of the
vertices in the bases (two vertices are identified, see
Figure~\ref{fig:p4-variable-gadget}).  Pairs of towers will be
separated by~$k'$ (defined below) triangles sharing an edge.  The
vertices not shared between the~$k'$ triangles will be called the
\emph{stack}, whereas the edge shared among the triangles will be
called the \emph{shortcut}.

The gadget~$G^x$ for a variable~$x$ consists of~$p_x$ pairs of towers
arranged in a cycle, one for each clause~$x$ appears in, where every
two consecutive pairs are separated by a shortcut edge and a stack of
vertices.  The stack is chosen to be big enough ($k'=k_\phi+3$
vertices) so that we will never delete the edge that connects the two
towers on each side of the stack, nor any edge incident to a vertex
from the stack.  We will refer to the two towers in the pairs as
Tower~1 (the one with lower index) and Tower~2.

Formally, let~$\phi$ be an instance of \pname{3Sat}.  The budget for the
output instance will be $k_\phi = 4 |\mathcal{C}(\phi)| + \sum_{x \in
  \mathcal{V}(\phi)}4p_x = 16 |\mathcal{C}(\phi)|$.  Let
$k'=k_\phi+3$.  For a variable~$x$ which appears in~$p_x$ clauses, we
create vertices~$s^x_{i,j}$ for $i \in \{1, \dots, p_x\}$ and $j \in
\{1, \dots, k'\}$.  These will be the vertices for the stacks.  For
the spikes of the towers, we add vertices~$t^x_{i,1}$ and~$t^x_{i,2}$
for $i \in \{1, \dots, p_x\}$.  For the base of the towers, we add
vertices~$b^x_{i,j}$ for $j \in \{1, \dots, 5\}$ and $i \in \{1,
\dots, p_x\}$.  These are all the vertices of the gadget~$G^x$ for $x
\in \mathcal{V}(\phi)$.

The vertices denoted by~$t$ are the two spikes in the tower, i.e.,
$t^x_{i, 1}$ is the spike of the Tower~1 of the~$i$th pair for
variable~$x$.  The vertices denoted by~$b$ are for the bases (there
are five vertices in the bases of the two towers).

Now we add the edges to~$G^x$, see
Figure~\ref{fig:p4-variable-gadget-enlarged}:
\begin{itemize}
\item For the stack, we add edges $s^x_{i, j}b^x_{i, 1}$ for all $i
  \in \{1, \dots, p_x\}$ and $j \in \{1, \dots, k'\}$ (right side of
  the stack) and edges $b^x_{i,5}s^x_{i+1, j}$ for $i \in \{1, \dots,
  p_x\}$ and $j \in \{1, \dots, k'\}$ (the left side of the next
  stack), where the indices behave cyclically modulo~$p_x$.
\item For the bases, we add the edges $b^x_{i,1}b^x_{i,2}$,
  $b^x_{i,1}b^x_{i,3}$, $b^x_{i,2}b^x_{i,3}$, $b^x_{i,3}b^x_{i,4}$,
  $b^x_{i,3}b^x_{i,5}$ and $b^x_{i,4}b^x_{i,5}$ for $i \in \{1, \dots,
  p_x\}$.  To attach the towers, we add the edges $b^x_{i,2}t^x_{i,1}$
  and $b^x_{i,4}t^x_{i,2}$.  The set of these eight edges will be
  denoted by~$R^x_i$.
\item The last edges to add are the shortcut edges
  $b^x_{i,5}b^x_{i+1,1}$ for $i \in \{1, \dots, p_x\}$, where again
  the indices behave cyclically modulo~$p_x$.
\end{itemize}

\begin{figure}[t]
  \centering
  \begin{tikzpicture}[every node/.style={circle, draw, scale=.65},
    scale=.68, minimum size=2.5em, inner sep=0pt]
    
    \node[draw=none] (startvertex) at (-2,0) {$\cdots$};

    \node (s11) at (0,1.0) {$s^x_{i,1}$};
    \node (s12) at (0,2.25) {$s^x_{i,2}$};
    
    \node[draw=none,rotate=90] (dots) at (0,3.25) {$\cdots$};
    
    \node (s1k) at (0,4.5) {$s^x_{i,k'}$};
    
    \node (b11) at (2,0) {$b^x_{i,1}$};
    \node (b12) at (3,1.5) {$b^x_{i,2}$};
    \node (b13) at (4,0) {$b^x_{i,3}$};
    \node (b14) at (5,1.5) {$b^x_{i,4}$};
    \node (b15) at (6,0) {$b^x_{i,5}$};
    
    \node (t11) at (3,4.5) {$t^x_{i,1}$};
    \node (t12) at (5,4.5) {$t^x_{i,2}$};

    \node (s21) at (8,1.0) {$s^x_{i+1,1}$};
    \node (s22) at (8,2.25) {$s^x_{i+1,2}$};
    
    \node[draw=none,rotate=90] (dots2) at (8,3.25) {$\cdots$};
    
    \node (s2k) at (8,4.5) {$s^x_{i+1,k'}$};

    \draw (s11) -- (startvertex) -- (s12);
    \draw (s1k) -- (startvertex);    
    \draw (b11) -- (startvertex);

    \draw (s11) -- (b11) -- (b12) -- (b13) -- (b14) -- (b15) -- (s21);
    \draw (s12) -- (b11) -- (b13) -- (b15);
    \draw (s12) -- (b11) -- (s1k);
    \draw (s22) -- (b15) -- (s2k);
    
    \draw[line width=1mm,Green] (b12) -- (t11);
    \draw[line width=1mm,Red]   (b14) -- (t12);
   
    \draw (b12) -- (t11);
    \draw (b14) -- (t12);
    
    \node (b21) at (10,0) {$b^x_{i+1,1}$};
    \node (b22) at (11,1.5) {$b^x_{i+1,2}$};
    \node (b23) at (12,0) {$b^x_{i+1,3}$};
    \node (b24) at (13,1.5) {$b^x_{i+1,4}$};
    \node (b25) at (14,0) {$b^x_{i+1,5}$};
    
    \node (t21) at (11,4.5) {$t^x_{i+1,1}$};
    \node (t22) at (13,4.5) {$t^x_{i+1,2}$};

    \node (s31) at (16,1.0) {$s^x_{i+2,1}$};
    \node (s32) at (16,2.25) {$s^x_{i+2,2}$};
    
    \node[draw=none,rotate=90] (dots3) at (16,3.25) {$\cdots$};
    
    \node (s3k) at (16,4.5) {$s^x_{i+2,k'}$};

    \node[draw=none] (endvertex) at (18,0) {$\cdots$};

    \draw (b25) -- (s31) -- (endvertex);
    \draw (b25) -- (s32) -- (endvertex);
    \draw (b25) -- (s3k) -- (endvertex);
    \draw (b25) --  (endvertex);

    \draw (s21) -- (b21) -- (b22) -- (b23) -- (b24) -- (b25);
    \draw (s22) -- (b21) -- (b23) -- (b25);
    \draw (s22) -- (b21) -- (s2k);
    
    \draw[line width=1mm,Green] (b22) -- (t21);
    \draw[line width=1mm,Red]   (b24) -- (t22);
   
    \draw (b22) -- (t21);
    \draw (b24) -- (t22);
    
    % shortcut
    \draw (b15) -- (b21);
    
  \end{tikzpicture}
  \caption{Variable gadget~$G^x$.  The counter~$i$ ranges from~$1$
    to~$p_x$, the number of clauses~$x$ appears in.  This figure does
    not illustrate that the gadget is a cycle, see
    Figure~\ref{fig:p4-variable-gadget} for a zoomed-out version.}
  \label{fig:p4-variable-gadget-enlarged}
\end{figure}
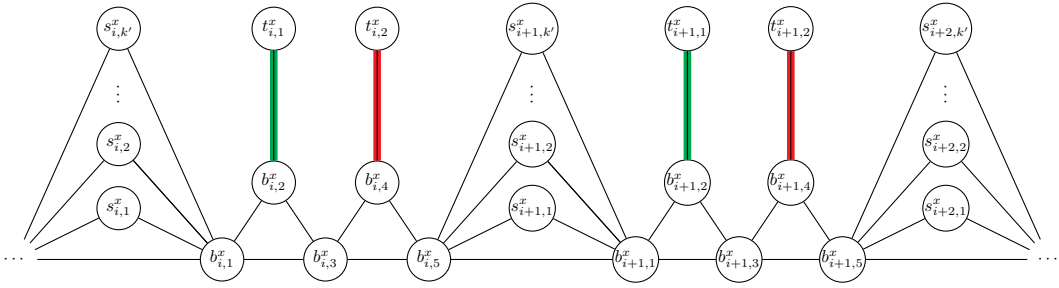

\paragraph{Elimination from variable gadgets}
We will now show that there are exactly two ways of eliminating
all~$P_4$s occurring in a variable gadget using at most~$4p_x$ edges.
To state this claim formally, we need to control how the variable
gadget is situated in a larger construction of the whole output
instance that will be defined later.  We say that a variable
gadget~$G^x$ is \emph{properly embedded} in the output
instance~$G_\phi$ if~$G^x$ is an induced subgraph of~$G_\phi$, and
moreover the only vertices of~$G^x$ that are incident to edges outside
$G^x$ are the spikes of the towers, i.e., vertices~$t^x_{i,1}$ and
$t^x_{i,2}$ for $i\in \{1,2,\ldots,p_x\}$.  This property will be
satisfied for gadgets~$G_x$ for all $x\in \mathcal{V}(\phi)$ in the
next steps of the construction.  Using this notion, we can infer
properties of the variable gadget irrespective of what the whole
output instance~$G_\phi$ constructed later looks like.

We first show that an inclusion minimal deletion set~$S$ that has size
at most~$k_\phi$ cannot touch the stacks nor the shortcut edges.

\begin{claim}\label{cl:not-touching-stacks}
  Assume gadget~$G^x$ is embedded properly in the output graph~$G_\phi$,
  and that~$S$ is an inclusion minimal~$P_4$-free edge deletion set
  in~$G_\phi$ of size at most~$k_\phi$.  Then~$S$ does not contain any
  edge of type $b^x_{i,5}b^x_{i+1,1}$ (a shortcut edge), nor any edge
  incident to a vertex of the form~$s^x_{i,j}$.
\end{claim}
\begin{proof}[of claim] 
  Suppose first that a shortcut edge $b^x_{i,5}b^x_{i+1,1}$ belongs
  to~$S$.  (See Figure~\ref{fig:p4-variable-gadget-enlarged} for
  indices.)  Let $S' = S \setminus \{b^x_{i,5}b^x_{i+1,1}\}$.
  Since~$S$ was inclusion minimal, the graph~$G_\phi-S'$ must contain
  an induced~$P_4$ that contains the edge $b^x_{i,5}b^x_{i+1,1}$;
  denote this~$P_4$ by~$L$.  By the assumption that~$G^x$ is properly
  embedded in~$G_\phi$ we have that~$L$ is entirely contained
  in~$G^x$.  Since the stack between pairs of towers~$i$ and~$i+1$ has
  height $k'=k_\phi+3$, we know that there are at least three vertices
  of the form $s^x_{i+1,j}$ for some $j \leq k$ which are not incident
  to an edge in~$S$.  Since~$L$ passes through~$2$ vertices apart from
  $b^x_{i,5}$ and $b^x_{i+1,1}$, we infer that one of these vertices,
  say $s^x_{i+1,j_0}$, is not incident to any edge of~$S$, nor it lies
  on~$L$.  Create~$L'$ by replacing the edge $b^x_{i,5}b^x_{i+1,1}$
  with the path $b^x_{i,5}-s^x_{i+1,j_0}-b^x_{i+1,1}$ on~$L$.  We
  infer that~$L'$ is an induced~$P_5$ in~$G_\phi-S$, which in
  particular contains an induced~$P_4$.  This is a contradiction to
  the definition of~$S$.
  
  Second, without loss of generality suppose now that the edge
  $b^x_{i,5}s^x_{i+1,j}$ belongs to~$S$ for some $j\in
  \{1,2,\ldots,k'\}$.  Let $S'=S\setminus
  \{b^x_{i,5}s^x_{i+1,j},b^x_{i+1,1}s^x_{i+1,j}\}$; note here that the
  edge $b^x_{i+1,1}s^x_{i+1,j}$ might had not belonged to~$S$, but if
  it had, then we remove it when constructing~$S'$.  Since~$S$ was
  inclusion minimal, the graph $G_\phi-S'$ must contain an induced
  $P_4$ that contain the vertex $s^x_{i+1,j}$, so also one of the
  vertices $b^x_{i,5}$ or $b^x_{i+1,1}$; denote this~$P_4$ by~$L$.
  Again, by the definition of proper embedding we have that~$L$ is
  entirely contained in~$G^x$.  By the same argumentation as before we
  infer that there exists a vertex $s^x_{i+1,j_0}$ such that
  $s^x_{i+1,j_0}$ is not traversed by~$L$ and is not incident to an
  edge of~$S$.  Since vertices $s^x_{i+1,j_0}$ and $s^x_{i+1,j}$ are
  twins in $G_\phi-S'$, it follows that the path~$L'$ constructed
  from~$L$ by substituting $s^x_{i+1,j}$ with $s^x_{i+1,j_0}$ is an
  induced~$P_4$ in~$G_\phi-S$.  This is a contradiction to the
  definition of~$S$.
\end{proof}

Now we show that every minimal deletion set~$S$ must use at least~$4$
edges in each pair of towers, and if it uses exactly~$4$ edges then
there are exactly~$4$ ways how the intersection of~$S$ with this pair
of towers can look like.

\begin{claim}\label{cl:elimination-types}
  Assume that the gadget~$G^x$ is embedded properly in the output
  graph~$G_\phi$, and that~$S$ is an inclusion minimal $P_4$-free edge
  deletion set in~$G_\phi$ of size at most~$k_\phi$.  Then for each $i
  \in \{1, 2, \ldots, p_x\}$ it holds that $|R^x_i \cap S| \geq 4$,
  and if $|R^x_i \cap S| = 4$ then either:
  \begin{itemize}
  \item[Elimination~$A$:] $R^x_i \cap S$ consists of the edges of the
    base of Tower~1 and the spike of Tower~2, or
  \item[Elimination~$B$:] $R^x_i \cap S$ consists of the edges of the
    base of Tower~2 and the spike of Tower~1, or
  \item[Elimination~$C$:] $R^x_i \cap S$ consists of the edges of both
    spikes and of the base of Tower~1 apart from the edge
    $b^x_{i,1}b^x_{i,2}$, or
  \item[Elimination~$D$:] $R^x_i \cap S$ consists of the edges of both
    spikes and of the base of Tower~2 apart from the edge
    $b^x_{i,4}b^x_{i,5}$.
  \end{itemize}
\end{claim}

We refer to Figure~\ref{fig:p4-variable-gadget-kill} for visualization
of all the four types of eliminations.  We will say that $R^x_i \cap
S$ {\emph{realizes Elimination~$X$}} for~$X$ being $A$, $B$, $C$,
or~$D$, if $R^x_i \cap S$ is as described in the statement of
Claim~\ref{cl:elimination-types}.  Similarly, we say that the $i$th
pair of towers {\emph{realizes Elimination~$X$}} if $R^x_i \cap S$
does.

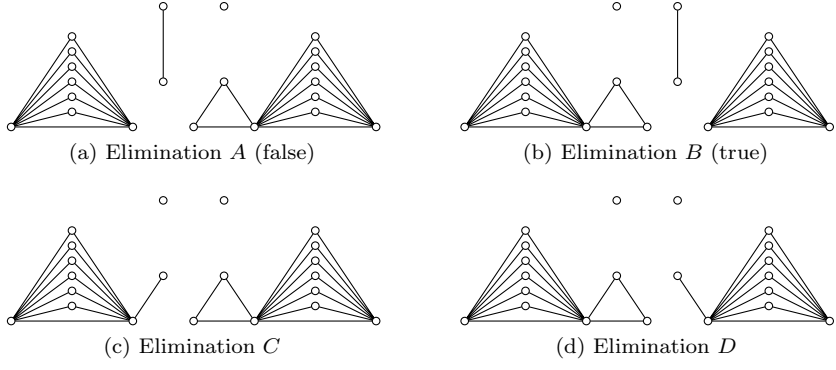
\begin{figure}[t]
  \centering
  \subfloat[Elimination~$A$ (false)] {
    \centering
    \begin{tikzpicture}[every
      node/.style={circle,draw,scale=.3},scale=.4]
      
      \node (s) at (-2, 0) {};
      \node (t) at (10, 0) {};

      \node (i) at (0,0.5) {};
      \node (j) at (0,1) {};
      \node (k) at (0,1.5) {};
      \node (l) at (0,2) {};
      \node (m) at (0,2.5) {};
      \node (n) at (0,3) {};
      
      \node (a) at (2,0) {};
      \node (b) at (4,0) {};
      \node (c) at (6,0) {};
      \node (d) at (3,1.5) {};
      \node (e) at (5,1.5) {};
      \node (f) at (3,4) {};
      \node (g) at (5,4) {};
      
      \node (i2) at (8,0.5) {};
      \node (j2) at (8,1) {};
      \node (k2) at (8,1.5) {};
      \node (l2) at (8,2) {};
      \node (m2) at (8,2.5) {};
      \node (n2) at (8,3) {};
      
      \draw (i) -- (a) -- (j);
      \draw (k) -- (a) -- (l);
      \draw (m) -- (a) -- (n);
      
      \draw (d) -- (f);
      
      \draw (b) -- (c) -- (e) -- (b);
      
      \draw (i2) -- (c) -- (j2);
      \draw (k2) -- (c) -- (l2);
      \draw (m2) -- (c) -- (n2);
      
      \draw (a) -- (s);
      \draw (c) -- (t);

      \draw (i) -- (s) -- (j);
      \draw (k) -- (s) -- (l);
      \draw (m) -- (s) -- (n);
      
      \draw (i2) -- (t) -- (j2);
      \draw (k2) -- (t) -- (l2);
      \draw (m2) -- (t) -- (n2);
    \end{tikzpicture}
    \label{fig:p4-variable-gadget-kill-a}
  }\hspace{2em}
  \subfloat[Elimination~$B$ (true)] {
    \centering
    \begin{tikzpicture}[every
      node/.style={circle,draw,scale=.3},scale=.4]

      \node (s) at (-2, 0) {};
      \node (t) at (10, 0) {};

      \node (i) at (0,0.5) {};
      \node (j) at (0,1) {};
      \node (k) at (0,1.5) {};
      \node (l) at (0,2) {};
      \node (m) at (0,2.5) {};
      \node (n) at (0,3) {};
      
      \node (a) at (2,0) {};
      \node (b) at (4,0) {};
      \node (c) at (6,0) {};
      \node (d) at (3,1.5) {};
      \node (e) at (5,1.5) {};
      \node (f) at (3,4) {};
      \node (g) at (5,4) {};
      
      \node (i2) at (8,0.5) {};
      \node (j2) at (8,1) {};
      \node (k2) at (8,1.5) {};
      \node (l2) at (8,2) {};
      \node (m2) at (8,2.5) {};
      \node (n2) at (8,3) {};
      
      \draw (i) -- (a) -- (j);
      \draw (k) -- (a) -- (l);
      \draw (m) -- (a) -- (n);
      
      \draw (a) -- (b) -- (d) -- (a);

      \draw (e) -- (g);
      
      \draw (i2) -- (c) -- (j2);
      \draw (k2) -- (c) -- (l2);
      \draw (m2) -- (c) -- (n2);
      
      \draw (a) -- (s);
      \draw (c) -- (t);

      \draw (i) -- (s) -- (j);
      \draw (k) -- (s) -- (l);
      \draw (m) -- (s) -- (n);
      
      \draw (i2) -- (t) -- (j2);
      \draw (k2) -- (t) -- (l2);
      \draw (m2) -- (t) -- (n2);
    \end{tikzpicture}
    \label{fig:p4-variable-gadget-kill-b}
  }\\
  \subfloat[Elimination~$C$] {
    \centering
    \begin{tikzpicture}[every
      node/.style={circle,draw,scale=.3},scale=.4]
      
      \node (s) at (-2, 0) {};
      \node (t) at (10, 0) {};
      
      \node (i) at (0,0.5) {};
      \node (j) at (0,1) {};
      \node (k) at (0,1.5) {};
      \node (l) at (0,2) {};
      \node (m) at (0,2.5) {};
      \node (n) at (0,3) {};
      
      \node (a) at (2,0) {};
      \node (b) at (4,0) {};
      \node (c) at (6,0) {};
      \node (d) at (3,1.5) {};
      \node (e) at (5,1.5) {};
      \node (f) at (3,4) {};
      \node (g) at (5,4) {};
      
      \node (i2) at (8,0.5) {};
      \node (j2) at (8,1) {};
      \node (k2) at (8,1.5) {};
      \node (l2) at (8,2) {};
      \node (m2) at (8,2.5) {};
      \node (n2) at (8,3) {};
      
      \draw (i) -- (a) -- (j);
      \draw (k) -- (a) -- (l);
      \draw (m) -- (a) -- (n);
      
      \draw (a) -- (d);
      
      \draw (b) -- (c) -- (e) -- (b);
      
      \draw (i2) -- (c) -- (j2);
      \draw (k2) -- (c) -- (l2);
      \draw (m2) -- (c) -- (n2);
      
      \draw (a) -- (s);
      \draw (c) -- (t);

      \draw (i) -- (s) -- (j);
      \draw (k) -- (s) -- (l);
      \draw (m) -- (s) -- (n);
      
      \draw (i2) -- (t) -- (j2);
      \draw (k2) -- (t) -- (l2);
      \draw (m2) -- (t) -- (n2);

    \end{tikzpicture}
    \label{fig:p4-variable-gadget-kill-c}
  }\hspace{2em}
  \subfloat[Elimination~$D$] {
    \centering
    \begin{tikzpicture}[every node/.style={circle, draw, scale=.3},
      scale=.4]
            
      \node (s) at (-2, 0) {};
      \node (t) at (10, 0) {};

      \node (i) at (0,0.5) {};
      \node (j) at (0,1) {};
      \node (k) at (0,1.5) {};
      \node (l) at (0,2) {};
      \node (m) at (0,2.5) {};
      \node (n) at (0,3) {};
      
      \node (a) at (2,0) {};
      \node (b) at (4,0) {};
      \node (c) at (6,0) {};
      \node (d) at (3,1.5) {};
      \node (e) at (5,1.5) {};
      \node (f) at (3,4) {};
      \node (g) at (5,4) {};
      
      \node (i2) at (8,0.5) {};
      \node (j2) at (8,1) {};
      \node (k2) at (8,1.5) {};
      \node (l2) at (8,2) {};
      \node (m2) at (8,2.5) {};
      \node (n2) at (8,3) {};
      
      \draw (i) -- (a) -- (j);
      \draw (k) -- (a) -- (l);
      \draw (m) -- (a) -- (n);
      
      \draw (a) -- (b) -- (d) -- (a);

      \draw (c) -- (e);
      
      \draw (i2) -- (c) -- (j2);
      \draw (k2) -- (c) -- (l2);
      \draw (m2) -- (c) -- (n2);
      
      \draw (a) -- (s);
      \draw (c) -- (t);

      \draw (i) -- (s) -- (j);
      \draw (k) -- (s) -- (l);
      \draw (m) -- (s) -- (n);
      
      \draw (i2) -- (t) -- (j2);
      \draw (k2) -- (t) -- (l2);
      \draw (m2) -- (t) -- (n2);
    \end{tikzpicture}
    \label{fig:p4-variable-gadget-kill-d}
  }
  \caption{The four different ways of eliminating a tower pair in a
    variable gadget.  Only Eliminations~$A$ and~$B$ yield optimum
    deletion sets in an entire variable gadget.  They all use exactly
    four edges per pair of towers, as is evident in the figure.}
  \label{fig:p4-variable-gadget-kill}
\end{figure}

\begin{proof}[of Claim~\ref{cl:elimination-types}]
  By Claim~\ref{cl:not-touching-stacks} we infer that~$S$ does not
  contain any edge incident to stacks~$i$ and~$i+1$, nor any of the
  shortcut edges incident to the considered pair of towers.  We
  consider four cases, depending on how the set $S \cap
  \{t_{i,1}^xb_{i,2}^x,t_{i,2}^xb_{i,4}^x\}$ looks like.  In each case
  we prove that $|R^x_i \cap S|\geq 4$, and that $|R^x_i \cap S|=4$
  implies that one of four listed elimination types is used.
  
  First assume that $S \cap
  \{t_{i,1}^xb_{i,2}^x,t_{i,2}^xb_{i,4}^x\}=\emptyset$ and observe
  that $s_{i,1}^x-b_{i,1}^x-b_{i,2}^x-t_{i,1}^x$ and
  $s_{i+1,1}^x-b_{i,5}^x-b_{i,4}^x-t_{i,2}^x$ are induced~$P_4$s
  in~$G_\phi$.  Since on each of these~$P_4$s there is only one edge
  that is not assumed to be not belonging to~$S$, it follows that both
  $b_{i,1}^xb_{i,2}^x$ and $b_{i,4}^xb_{i,5}^x$ must belong to~$S$.
  Suppose that $b_{i,1}^xb_{i,3}^x\notin S$.  Then we infer that
  $b_{i,2}^xb_{i,3}^x,b_{i,3}^xb_{i,4}^x,b_{i,3}^xb_{i,5}^x\in S$,
  since otherwise any of these edges would form an induced~$P_4$ in
  $G_\phi-S$ together with edges $b_{i,1}^xb_{i,3}^x$ and
  $b_{i,1}^xs_{i,1}^x$.  We infer that in this case $|R^x_i \cap
  S|\geq 5$, and a symmetric conclusion can be drawn when
  $b_{i,3}^xb_{i,5}^x\notin S$.  We are left with the case when
  $b_{i,1}^xb_{i,3}^x,b_{i,3}^xb_{i,5}^x\in S$.  But then~$S$ must
  include also one of the edges $b_{i,2}^xb_{i,3}^x$ or
  $b_{i,3}^xb_{i,4}^x$ so that the induced~$P_4$
  $t^x_{i,1}-b^x_{i,2}-b^x_{i,3}-b^x_{i,4}$ is destroyed.  Hence, in
  all the considered cases we conclude that $|R^x_i \cap S|\geq 5$.
  
  Second, assume that $S \cap
  \{t_{i,1}^xb_{i,2}^x,t_{i,2}^xb_{i,4}^x\}=\{t_{i,2}^xb_{i,4}^x\}$.
  The same reasoning as in the previous paragraph shows that
  $b_{i,1}^xb_{i,2}^x$ must belong to~$S$.  Again, if
  $b_{i,1}^xb_{i,3}^x\notin S$, then all the edges
  $b_{i,2}^xb_{i,3}^x,b_{i,3}^xb_{i,4}^x,b_{i,3}^xb_{i,5}^x$ must
  belong to~$S$, and so $|R^x_i \cap S|\geq 5$.  Assume then that
  $b_{i,1}^xb_{i,3}^x\in S$.  Note now that we have two induced
  $P_4s$: $t^x_{i,1}-b^x_{i,2}-b^x_{i,3}-b^x_{i,4}$ and
  $t^x_{i,1}-b^x_{i,2}-b^x_{i,3}-b^x_{i,5}$ that share the edge
  $t_{i,1}^xb_{i,2}^x$ about which we assumed that it does not belong
  to~$S$, and the edge $b^x_{i,2}b^x_{i,3}$.  To remove both
  these~$P_4$s we either remove at least two more edges, which results
  in conclusion that $|R^x_i \cap S|\geq 5$, or remove the edge
  $b^x_{i,2}b^x_{i,3}$, which results in Elimination~$A$.
  
  The third case when $S \cap
  \{t_{i,1}^xb_{i,2}^x,t_{i,2}^xb_{i,4}^x\}=\{t_{i,1}^xb_{i,2}^x\}$ is
  symmetric to the second case, and leads to a conclusion that either
  $|R^x_i \cap S|\geq 5$ or $R^x_i \cap S$ realizes Elimination~$B$.
  
  Finally, assume that $t_{i,1}^xb_{i,2}^x,t_{i,2}^xb_{i,4}^x\in S$.
  Observe that we have an induced~$P_4$
  $s^x_{i,1}-b^x_{i,1}-b^x_{i,3}-b^x_{i,5}$ in~$G_\phi$, so one of the
  edges $b^x_{i,1}b^x_{i,3}$ or $b^x_{i,3}b^x_{i,5}$ must be included
  in~$S$.  Assume first that $b^x_{i,1}b^x_{i,3}\in S$.  Consider now
  $P_4$s $s^x_{i,1}-b^x_{i,1}-b^x_{i,2}-b^x_{i,3}$ and
  $b^x_{i,2}-b^x_{i,3}-b^x_{i,5}-s^x_{i+1,1}$.  Both these~$P_4$s need
  to be destroyed by~$S$ since after removing $b^x_{i,1}b^x_{i,3}$ the
  first~$P_4$ becomes induced, while the second is induced already
  in~$G_\phi$.  Moreover, these~$P_4$s share only the edge
  $b^x_{i,2}b^x_{i,3}$, which means that either $|R^x_i \cap S|\geq 5$
  or $b^x_{i,2}b^x_{i,3}\in S$ and $R^x_i \cap S$ realizes
  Elimination~$C$.  The case when $b^x_{i,3}b^x_{i,5}\in S$ is
  symmetric and leads to a conclusion that either $|R^x_i \cap S|\geq
  5$ or $R^x_i\cap S$ realizes Elimination~$D$.
\end{proof}

Finally, we are able to prove that the variable gadget~$G^x$ requires
at least~$4p_x$ edge deletions, and that there are only two ways of
destroying all~$P_4$s by using exactly~$4p_x$ edge deletions: either
by applying Elimination~$A$ or Elimination~$B$ to all the pairs of
towers.

\begin{claim}\label{cl:deletion-var-gadget}
  Suppose a gadget~$G^x$ is embedded properly in the output
  graph~$G_\phi$, and that~$S$ is an inclusion minimal $P_4$-free edge
  deletion set in~$G_\phi$ of size at most~$k_\phi$.  Then
  $|E(G^x)\cap S|\geq 4p_x$, and if $|E(G^x)\cap S|=4p_x$, then either
  $R^x_i\cap S$ realizes Elimination~$A$ for all $i\in
  \{1,2,\ldots,p_x\}$, or $R^x_i\cap S$ realizes Elimination~$B$ for
  all $i\in \{1,2,\ldots,p_x\}$.
\end{claim}
\begin{proof}[of claim]
  By Claims~\ref{cl:not-touching-stacks}
  and~\ref{cl:elimination-types} we have that~$S$ does not contain any
  shortcut edge or edge incident to a stack vertex, and moreover that
  $|R^x_i \cap S|\geq 4$ for all $i\in \{1,2,\ldots,p_x\}$.  Since
  sets~$R^x_i$ are pairwise disjoint, it follows that $|E(G^x) \cap
  S|\geq 4p_x$.  Moreover, if $|E(G^x) \cap S|=4p_x$, then $|R^x_i
  \cap S|=4$ for all $i\in \{1,2,\ldots,p_x\}$ and, by
  Claim~\ref{cl:elimination-types}, for all $i\in \{1,2,\ldots,p_x\}$
  the set $R^x_i \cap S$ must realize Elimination~$A$, $B$, $C$,
  or~$D$.
  
  We say that one pair of towers is \emph{followed} by another, if the
  former has index~$i$, and the latter has index~$i+1$ (of course,
  modulo~$p_x$).  To obtain the conclusion that either all the sets
  $R^x_i \cap S$ realize Elimination~$A$ or all of them realize
  Elimination~$B$, we observe that when some pair of towers realize
  Elimination~$A$, $C$, or~$D$, then the following pair must realize
  Elimination~$A$.  Indeed, otherwise the graph~$G_\phi-S$ would
  contain an induced~$P_4$ of the form
  $b^x_{i,4}-b^x_{i,5}-b^x_{i+1,1}-b^x_{i+1,3}$, where the~$i$th pair
  of towers is the considered pair that realizes Elimination~$A$, $C$,
  or~$D$.  Now observe that since the pairs of towers are arranged on
  a cycle, then either all pairs of towers realize Elimination~$B$, or
  at least one realizes Elimination~$A$, $C$, or~$D$, which means that
  the following pair realizes Elimination~$A$, and so all the pairs
  must realize Elimination~$A$.
\end{proof}

\paragraph{Clause gadget}
We now move on to construct the clause gadget~$G^c$ for a clause $c
\in \mathcal{C}(\phi)$.  Assume that $c = \ell_x \lor \ell_y \lor
\ell_z$, where~$\ell_r$ is a literal of variable~$r$ for $r\in
\{x,y,z\}$.  We create seven vertices: one vertex~$u^c$ and
vertices~$u^r_2$ and~$u^r_3$ for $r = x,y,z$.  We also add the edges
$u^cu^r_2$, $u^cu^r_3$ and $u^r_2u^r_3$.  Now, for non-negated $r \in
\{x,y,z\}$ in~$c$, where~$c$ is the~$i$th clause~$r$ appears in, we
add edges $u^r_2 t^r_{i,1}$ and $u^r_3 t^r_{i,1}$ (recall that
$t^r_{i,1}$ is the spike of Tower~1 in tower pair~$i$).  If~$r$
appears negated, we add the edges $u^r_2 t^r_{i,2}$ and $u^r_3
t^r_{i,2}$ instead, see Figure~\ref{fig:p4-connection-gadget}.
Let~$M^c$ be the set comprising all the~$15$ created edges, including
the ones incident to the spikes of the towers.  By~$M^{c,r}$ for $r\in
\{x,y,z\}$ we denote the subset of~$M^c$ containing~$5$ edges that are
incident to vertex~$u^r_2$ or~$u^r_3$.

\begin{figure}[t]
  \centering \subfloat[Clause gadget for a clause $c = x \lor \neg y \lor
  z$.  The dashed vertices are the connection points in the variable
  gadgets.  Observe that since $y$ appears negated, we attach it to
  Tower~2 in its pair.  The clause $c$ is the $i_\ell$th clause $\ell$
  appears in.] { \centering
    \begin{tikzpicture}[every node/.style={circle, draw, scale=.7,
        inner sep=0, minimum size=2em}, scale=.7]
      
      \node[dashed] (v1) at (3,7) {$t^z_{i_z,1}$};
      \node (v2) at (2,5)         {$u^z_2$};    
      \node (v3) at (4,5)         {$u^z_3$};    
      
      \node[dashed] (v4) at (0,0) {$t^x_{i_x,1}$};
      \node (v5) at (2.5,0.5)     {$u^x_2$};    
      \node (v6) at (0.5,2.5)     {$u^x_3$};    
      
      \node[dashed] (v7) at (6,0) {$t^y_{i_y,2}$};
      \node (v8) at (5.5,2.5)     {$u^y_2$};    
      \node (v9) at (3.5,0.5)     {$u^y_3$};    
      
      \node (v10) at (3,3) {$u^c$};
      
      \draw (v1) -- (v2) -- (v10) -- (v3) -- (v1);
      \draw (v4) -- (v5) -- (v10) -- (v6) -- (v4);
      \draw (v7) -- (v8) -- (v10) -- (v9) -- (v7);
      \draw (v2) -- (v3);
      \draw (v5) -- (v6);
      \draw (v8) -- (v9);
    \end{tikzpicture}
    \label{fig:p4-clause-gadget}
  }
  \hspace{.02\textwidth}
  \subfloat[Clause gadget elimination when $c$ is satisfied by variable
      $z$.] {
    \centering
    \begin{tikzpicture}[every node/.style={circle, draw, scale=.7,
        inner sep=0, minimum size=2em}, scale=.7]
      
      \node[dashed] (v1) at (3,7) {$t^z_{i_z,1}$};
      \node (v2) at (2,5)         {$u^z_2$};    
      \node (v3) at (4,5)         {$u^z_3$};    
      
      \node[dashed] (v4) at (0,0) {$t^x_{i_x,1}$};
      \node (v5) at (2.5,0.5)     {$u^x_2$};    
      \node (v6) at (0.5,2.5)     {$u^x_3$};    
      
      \node[dashed] (v7) at (6,0) {$t^y_{i_y,2}$};
      \node (v8) at (5.5,2.5)     {$u^y_2$};    
      \node (v9) at (3.5,0.5)     {$u^y_3$};    
      
      \node (v10) at (3,3) {$u^c$};
      
      \draw (v1) -- (v2) -- (v10) -- (v3) -- (v1);
      \draw (v4) -- (v5) -- (v6) -- (v4);
      \draw (v7) -- (v8) -- (v9) -- (v7);
      \draw (v2) -- (v3);
      \draw (v5) -- (v6);
      \draw (v8) -- (v9);
    \end{tikzpicture}
    \label{fig:p4-clause-gadget-killed}
  }
  \caption{Clause gadget $G^c$ for a clause $c = x \lor \neg y \lor
    z$.  To the left it is before elimination, to the right an optimal
    elimination when satisfied by $z$.}
\end{figure}
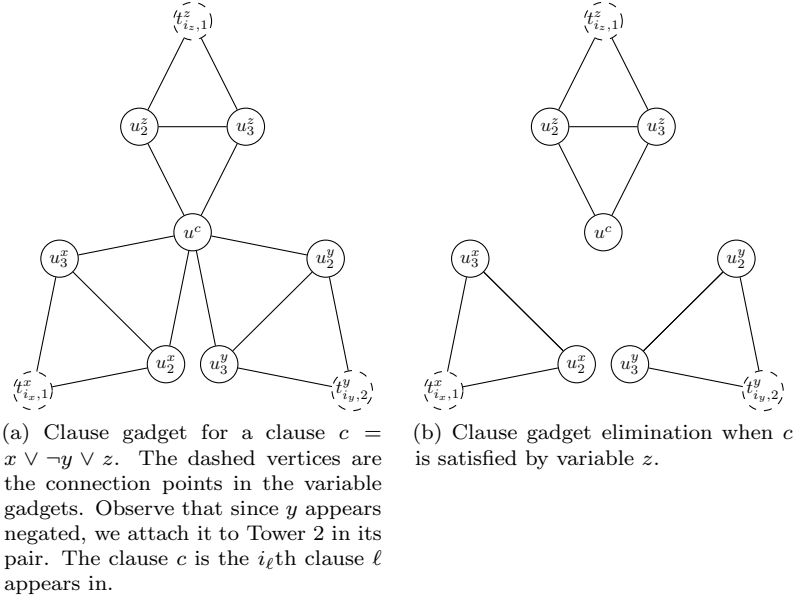

This concludes the construction of the graph~$G_\phi$; note that all the
variable gadgets are properly embedded in~$G_\phi$.  Before showing
the correctness of the reduction, we prove the following claims about
the number of edges needed for the clause gadgets:

\begin{claim}\label{cl:deletion-clause-gadget}
  Assume that~$S$ is a $P_4$-free deletion set of graph~$G_\phi$.
  Let~$c$ be a clause of~$\phi$, and assume that $x,y,z$ are the
  variables appearing in~$c$.  Then $|S \cap M^c|\geq 4$, and if $|S
  \cap M^c| = 4$, then $S \cap M^{c,r} = \emptyset$ for some $r \in
  \{x,y,z\}$ (see Figure~\ref{fig:p4-clause-gadget-killed} for an
  example where $S \cap M^{c,z} = \emptyset$).
\end{claim}
\begin{proof}[of claim]
  To simplify the notation, let $t^x,t^y,t^z$ be the corresponding
  vertices of the variable gadgets that are incident to edges
  of~$M^c$.
  
  If $|S \cap M^{c,r}|\geq 2$ for all $r\in \{x,y,z\}$, then $|S \cap M^c|\geq 6$
  and we are done.  Assume then without loss of generality that $|S
  \cap M^{c,x}| \leq 1$.  Hence, at least one of the paths $t^x -
  u^x_2 - u^c$ and $t^x - u^x_3 - u^c$ does not contain an edge
  of~$S$.  Assume without loss of generality that it is $t^x - u^x_2 -
  u^c$.  Now observe that in~$G_\phi$ we have~$4$ induced~$P_4$
  created by prolonging this~$P_3$ by vertex $u^y_2$, $u^y_3$, $u^z_2$
  or~$u^z_3$.  Since $t^x - u^x_2 - u^c$ is disjoint with~$S$, it
  follows that all the four edges connecting these vertices with~$u^c$
  must belong to~$S$.  Hence $|S \cap M^c|\geq 4$, and if $|S \cap M^c| =
  4$ then~$M^{c,x}$ must be actually disjoint with~$S$.
\end{proof}

\begin{figure}[t]
  \centering
  \begin{tikzpicture}[every node/.style={circle, draw, scale=.6},
    scale=.6]
    
    \node (v1) at (7,10) {};
    \node (v2) at (6,8) {};
    \node (v3) at (8,8) {};
    
    \node (v4) at (4,4) {};
    \node (v5) at (5,6) {};
    \node (v6) at (6,4) {};
    
    \node (v7) at (10,4) {};
    \node (v8) at (8,4) {};
    \node (v9) at (9,6) {};
    
    \node (v10) at (7,6) {};

    \draw[Blue] (v1) -- (v2) -- (v10) -- (v3) -- (v1);
    \draw[Blue] (v4) -- (v5) -- (v10) -- (v6) -- (v4);
    \draw[Blue] (v7) -- (v8) -- (v10) -- (v9) -- (v7);
    \draw[Blue] (v2) -- (v3);
    \draw[Blue] (v5) -- (v6);
    \draw[Blue] (v8) -- (v9);

    % X gadget
    \node[draw=none,scale=2] (xxx) at (2,5) {$x$};
    \node (xt2) at (5,3) {};
    \node (xb1) at (2,3) {};
    \node (xb2) at (3,3) {};
    \node (xb3) at (3,2) {};
    \node (xb4) at (4,2) {};
    \node (xb5) at (4,1) {};

    \draw[Blue] (xb1) -- (xb2) -- (xb3) -- (xb1);
    \draw[Blue] (xb3) -- (xb4) -- (xb5) -- (xb3);
    \draw[thick,Green] (xb2) -- (v4);
    \draw[thick,Red] (xb4) -- (xt2);

    % Y gadget
    \node[draw=none,scale=2] (yyy) at (12,4.5) {$\neg y$};
    \node (yt1) at (9,3) {};
    \node (yb1) at (10,1) {};
    \node (yb2) at (10,2) {};
    \node (yb3) at (11,2) {};
    \node (yb4) at (11,3) {};
    \node (yb5) at (12,3) {};

    \draw[Blue] (yb1) -- (yb2) -- (yb3) -- (yb1);
    \draw[Blue] (yb3) -- (yb4) -- (yb5) -- (yb3);
    \draw[thick,Green] (yb2) -- (yt1);
    \draw[thick,Red] (yb4) -- (v7);

    % Z gadget
    \node[draw=none,scale=2] (zzz) at (3.5,10.5) {$z$};
    \node (zt2) at (6,10) {};
    \node (zb1) at (7.5,12.414) {}; % .414 is appx sqrt(2)
    \node (zb2) at (7,11.414) {};
    \node (zb3) at (6.5,12.414) {};
    \node (zb4) at (6,11.414) {};
    \node (zb5) at (5.5,12.414) {};

    \draw[Blue] (zb1) -- (zb2) -- (zb3) -- (zb1);
    \draw[Blue] (zb3) -- (zb4) -- (zb5) -- (zb3);
    \draw[thick,Green] (zb2) -- (v1);
    \draw[thick,Red] (zb4) -- (zt2);
  \end{tikzpicture}
  \caption{For a clause $c = x \lor \neg y \lor z$, we obtain the
    above connection.  For negated variables, the rightmost spike is
    attached to the gadgets, otherwise the leftmost spike is attached.
    If~$x$ is being evaluated to a value satisfying~$c$, the edge
    spike between~$G^x$ and~$G^c$ is deleted.}
  \label{fig:p4-connection-gadget}
\end{figure}
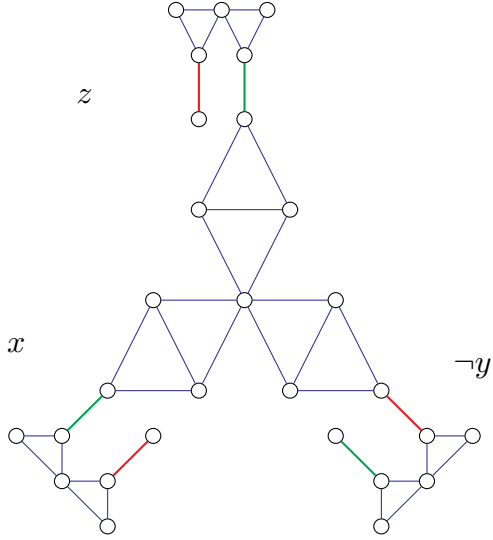

We are finally ready to prove the following lemma, which implies
correctness of the reduction.
\begin{lemma}
  Given an input instance~$\phi$ to \pname{3Sat},~$\phi$ is satisfiable if
  and only if the constructed graph~$G_\phi$ has a~$P_4$ deletion set
  of size $k_\phi = 16|\mathcal{C}(\phi)|$.
\end{lemma}
\begin{proof}
  From left to right, suppose~$\phi$ is satisfiable by an assignment~$\alpha$,
  and let~$G_\phi$ and~$k_\phi$ be as above.  If a variable~$x$ is
  assigned false in~$\alpha$, we delete as in
  Figure~\ref{fig:p4-variable-gadget-kill-a}, that is, we apply
  Elimination~$A$ to all the pairs of towers in the variable
  gadget~$G^x$.  Otherwise we delete as in
  Figure~\ref{fig:p4-variable-gadget-kill-b}, that is, we apply
  Elimination~$B$ to all the pairs of towers in the variable
  gadget~$G^x$.  In other words, if~$x$ assigned to false
  (Elimination~$A$), we delete the edges $t^x_{i,2}b^x_{i,4}$,
  $b^x_{i,1}b^x_{i,2}$, $b^x_{i,1}b^x_{i,3}$, $b^x_{i,2}b^x_{i,3}$,
  otherwise, when~$x$ is assigned to true (Elimination~$B$), we delete
  the edges $t^x_{i,1}b^x_{i,2}$, $b^x_{i,3}b^x_{i,4}$,
  $b^x_{i,3}b^x_{i,5}$, $b^x_{i,4}b^x_{i,5}$, for all $i \in
  \{1,\dots,p_x\}$.
  
  Furthermore, for every clause $c = \ell_x \lor \ell_y \lor \ell_z$ we
  choose an arbitrary variable whose literal satisfies~$c$, say~$r$.
  We remove the edges $u^{r'}_2u^c$ and $u^{r'}_3u^c$ for $r'\neq r$.
  We have thus used exactly four edge removals per clause, $4
  |\mathcal{C}(\phi)|$ in total, and for each $x \in
  \mathcal{V}(\phi)$ we have removed~$4p_x$ edges.  This sums up
  exactly to $4 |\mathcal{C}(\phi)| + \sum_{x \in \mathcal{V}(\phi)}
  4p_x = 4 |\mathcal{C}(\phi)| + 4\sum_{x \in \mathcal{V}(\phi)} p_x =
  4 |\mathcal{C}(\phi)| + 4 \cdot 3 |\mathcal{C}(\phi)|= 16
  |\mathcal{C}(\phi)| = k_\phi$ edge removals.
  
  We now claim that~$G_\phi$ is $P_4$-free.  A direct check shows that
  there is no induced~$P_4$ left inside any variable gadget, nor
  inside any clause gadget.  Therefore, any induced~$P_4$ left must
  necessarily contain vertex of the form $t^x_{i,q}$ for some $x\in
  \mathcal{V}(\phi)$, $i\in \{1,2,\ldots,p_x\}$, and $q\in \{1,2\}$,
  together with the edge of the spike incident to this vertex and one
  of the edges of gadget~$G^c$ incident to this vertex, where~$c$ is
  the~$i$th clause~$x$ appears in.  Assume without loss of generality
  that~$q=1$, so~$x$ appears in~$c$ positively.  Since we did not
  delete the spike edge $t^x_{i,1}b^x_{i,2}$, we infer that
  $\alpha(x)=\mathtt{false}$.  Therefore~$x$ does not satisfy~$c$, so
  we must have deleted edges $u^cu^x_2$ and $u^cu^x_3$.  Thus in the
  remaining graph $G_\phi-S$ the connected component of the vertex
  $t^x_{i,q}$ is a triangle with a pendant edge, which is $P_4$-free.
  We conclude that~$G_\phi-S$ is indeed $P_4$-free.

  \medskip
  
  From right to left, suppose now that~$G_\phi$ is the constructed graph
  from a fixed~$\phi$ and that for~$k_\phi$ as above, we have that
  $(G_\phi,k_\phi)$ is a \yes{} instance of \pname{$P_4$-Free Edge
    Deletion}.  Let~$S$ be a~$P_4$ deletion set of size at
  most~$k_\phi$, and without loss of generality assume that~$S$ is
  inclusion minimal.  By Claims~\ref{cl:deletion-var-gadget}
  and~\ref{cl:deletion-clause-gadget}, set ~$S$ must contain at
  least~$4p_x$ edges in each set~$E(G^x)$, and at least four edges in
  each set $M^c$.  Since $4 |\mathcal{C}(\phi)| + \sum_{x \in
    \mathcal{V}(\phi)} 4p_x = k_\phi$, we infer that~$S$ contains
  exactly~$4p_x$ edges in each set~$E(G^x)$, and exactly four edges in
  each set~$M^c$.  By Claim~\ref{cl:deletion-var-gadget} we infer that
  for each variable~$x$, all the pairs of towers in~$G^x$ realize
  Elimination~$A$, or all of them realize Elimination~$B$.  Let
  $\alpha\colon \mathcal{V}(\phi) \to \{\mathtt{true},
  \mathtt{false}\}$ be an assignment that assigns value
  $\mathtt{false}$ if Elimination~$A$ is used throughout the
  corresponding gadget, and value $\mathtt{true}$ otherwise.  We claim
  that~$\alpha$ satisfies~$\phi$.
  
  Consider a clause $c \in \mathcal{C}(\phi)$ and assume that $x,y,z$ are
  variables appearing in~$c$.  By
  Claim~\ref{cl:deletion-clause-gadget} we infer that there exists $r
  \in \{x,y,z\}$ such that $S \cap M^{c,r}=\emptyset$.  Assume without
  loss of generality that $r = x$, and that~$x$ appears positively
  in~$c$.  Moreover, assume that~$c$ is the~$i_x$th clause~$x$ appears
  in.  We claim that $\alpha(x) = \mathtt{true}$, and thus~$c$ is
  satisfied by~$x$.  Indeed, otherwise the edge
  $t^x_{i_x,1}b^x_{i_x,2}$ would not be deleted, and thus
  $b^x_{i_x,2}-t^x_{i_x,1}-u^x_2-u^c$ would be an induced~$P_4$ in
  $G_\phi-S$; this is a contradiction to the definition of~$S$.
\end{proof}

Again, the proof of Theorem~\ref{thm:p4-completion-exp} follows:
combining the presented reduction with an algorithm for
\pname{$P_4$-Free Edge Deletion} working in $2^{o(k)}n^{\bigO(1)}$
time would give an algorithm for \pname{3Sat} working in $2^{o(n+m)}(n
+ m)^{\bigO(1)}$ time, which contradicts ETH by the results of
Impagliazzo, Paturi and Zane~\cite{impagliazzo2001which}.

It is easy to verify that in the presented reduction, both the graph
$G_\phi$ and $G_\phi-S$ for $S$ being the deletion set constructed for
a satisfying assignment for $\phi$ are actually $C_4$-free.  Thus the
same reduction also shows that \pname{$\{C_4, P_4\}$-free Deletion} is
not solvable in $2^{o(k)}n^{\bigO(1)}$ time unless ETH fails; Since
$\overline{P_4} = P_4$ and $\overline{C_4} = 2K_2$, it follows that
\pname{$\{2K_2, P_4\}$-Free Completion} is hard under ETH as well.  In
other words we derive the following result: \pname{Co-Trivially
  Perfect Completion} is not solvable in subexponential time unless
ETH fails.

\begin{theorem}
  \label{thm:2k2-p4-completion-exp}
  The problem \pname{$\{2K_2, P_4\}$-Free Completion} is not solvable
  in $2^{o(k)}n^{\bigO(1)}$ time unless ETH fails.
\end{theorem}

\section{Conclusion and future work}
\label{sec:conclusion}

In this paper, we provided several upper and lower subexponential
parameterized bounds for \fd.  The most natural open question would be
to ask for a dichotomy characterizing for which sets $\F$, \fd{}
problems are in $\cclass{P}$, in \cclass{SUBEPT}, and not in
\cclass{SUBEPT} (under ETH).  Keeping in mind the lack of such
characterization concerning classes $\cclass{P}$ and $\cclass{NP}$, an
answer to this question can be very non-trivial.  Even a more modest
task---deriving general arguments explaining what causes a completion
problem to be in \cclass{SUBEPT}---is an important open question.

\bigskip

\noindent
Similarly, from an algorithmic perspective obtaining generic
subexponential algorithms for completion problems would be a big step
forwards.  With the current knowledge, for different cases of~$\F$,
the algorithms are built on different ideas like chromatic coding,
potential maximal cliques, $k$-cuts, etc.\ and each new case requires
special treatment.

\bigskip

\noindent
Another interesting property is that all the graph classes for which
subexponential algorithms for completion problems are known, are
tightly connected to chordal graphs.  Indeed, all the known algorithms
exploit existence of a chordal-like decomposition of the target
completed graph.  Are there natural \cclass{NP}-hard graph
modification problems admitting subexponential time algorithms where
the graph class target is \emph{not} related to chordal graphs?

\bigskip

\noindent
Finally, in this paper we have presented \cclass{SUBEPT} lower bounds
(under ETH) for \fd{} for several different cases of $\F$, but we lack
a method for proving tight lower bounds on the running time for
problems that actually are in \cclass{SUBEPT}.  For instance, it may
be the case that \pname{Trivially Perfect Completion} or
\pname{Chordal Completion} can be solved in time $2^{\bigO(k^{1/4})}
n^{\bigO(1)}$.  As Fomin and Villanger~\cite{fomin2012subexponential}
observed, in the case of \pname{Chordal Completion} known
\cclass{NP}-hardness reductions provide lower bounds much weaker than
the current upper bound of $2^{\bigO(\sqrt{k}\log k)} n^{\bigO(1)}$.
However, we feel that a $2^{o(\sqrt{k})} n^{\bigO(1)}$ running time
should be impossible to achieve, since such an algorithm would
immediately imply the existence of an exact algorithm with running
time~$2^{o(n)}$.  Is it possible to prove $2^{o(\sqrt{k})}
n^{\bigO(1)}$ lower bounds under ETH for \pname{Trivially Perfect
  Completion}, \pname{Chordal Completion}, and other completion
problems to subclasses of chordal graphs known to be contained in
\cclass{SUBEPT}?

\bibliographystyle{plain}

\end{document}